\documentclass[lettersize,journal]{IEEEtran}
\usepackage{amsmath,amsfonts}
\usepackage{algorithm}
\usepackage{array}
\usepackage[caption=false,font=footnotesize]{subfig}
\usepackage{textcomp}
\usepackage{stfloats}
\usepackage{url}
\usepackage{verbatim}
\usepackage{graphicx}
\usepackage{pgffor}
\usepackage{cite}
\usepackage{todonotes}
\usepackage[hidelinks]{hyperref}
\usepackage{diagbox}
\definecolor{kit-green}{RGB}{0, 150, 130}
\colorlet{KITgreen}{kit-green}
\colorlet{kit-green100}{kit-green}
\colorlet{kit-green90}{kit-green!90!white}
\colorlet{kit-green80}{kit-green!80!white}
\colorlet{kit-green70}{kit-green!70!white}
\colorlet{kit-green60}{kit-green!60!white}
\colorlet{kit-green50}{kit-green!50!white}
\colorlet{kit-green40}{kit-green!40!white}
\colorlet{kit-green30}{kit-green!30!white}
\colorlet{kit-green25}{kit-green!25!white}
\colorlet{kit-green20}{kit-green!20!white}
\colorlet{kit-green15}{kit-green!15!white}
\colorlet{kit-green10}{kit-green!10!white}
\colorlet{kit-green5}{kit-green!5!white}

\definecolor{kit-blue}{RGB}{70, 100, 170}
\colorlet{KITblue}{kit-blue}
\colorlet{kit-blue100}{kit-blue}
\colorlet{kit-blue90}{kit-blue!90!white}
\colorlet{kit-blue80}{kit-blue!80!white}
\colorlet{kit-blue70}{kit-blue!70!white}
\colorlet{kit-blue60}{kit-blue!60!white}
\colorlet{kit-blue50}{kit-blue!50!white}
\colorlet{kit-blue40}{kit-blue!40!white}
\colorlet{kit-blue30}{kit-blue!30!white}
\colorlet{kit-blue25}{kit-blue!25!white}
\colorlet{kit-blue20}{kit-blue!20!white}
\colorlet{kit-blue15}{kit-blue!15!white}
\colorlet{kit-blue10}{kit-blue!10!white}
\colorlet{kit-blue5}{kit-blue!5!white}

\definecolor{kit-royalblue}{RGB}{0, 45, 76}
\colorlet{kit-royalblue100}{kit-royalblue}
\colorlet{kit-royalblue90}{kit-royalblue!90!white}
\colorlet{kit-royalblue80}{kit-royalblue!80!white}
\colorlet{kit-royalblue70}{kit-royalblue!70!white}
\colorlet{kit-royalblue60}{kit-royalblue!60!white}
\colorlet{kit-royalblue50}{kit-royalblue!50!white}
\colorlet{kit-royalblue40}{kit-royalblue!40!white}
\colorlet{kit-royalblue30}{kit-royalblue!30!white}
\colorlet{kit-royalblue25}{kit-royalblue!25!white}
\colorlet{kit-royalblue20}{kit-royalblue!20!white}
\colorlet{kit-royalblue15}{kit-royalblue!15!white}
\colorlet{kit-royalblue10}{kit-royalblue!10!white}
\colorlet{kit-royalblue5}{kit-royalblue!5!white}

\definecolor{kit-iceblue100}{RGB}{30, 53, 69}
\definecolor{kit-iceblue70}{RGB}{68, 94, 111}
\definecolor{kit-iceblue50}{RGB}{168, 185, 196}
\definecolor{kit-iceblue30}{RGB}{218, 225, 230}

\definecolor{kit-red}{RGB}{162, 34, 35}
\colorlet{KITred}{kit-red}
\colorlet{kit-red100}{kit-red}
\colorlet{kit-red90}{kit-red!90!white}
\colorlet{kit-red80}{kit-red!80!white}
\colorlet{kit-red70}{kit-red!70!white}
\colorlet{kit-red60}{kit-red!60!white}
\colorlet{kit-red50}{kit-red!50!white}
\colorlet{kit-red40}{kit-red!40!white}
\colorlet{kit-red30}{kit-red!30!white}
\colorlet{kit-red25}{kit-red!25!white}
\colorlet{kit-red20}{kit-red!20!white}
\colorlet{kit-red15}{kit-red!15!white}
\colorlet{kit-red10}{kit-red!10!white}
\colorlet{kit-red5}{kit-red!5!white}

\definecolor{kit-yellow}{RGB}{252, 229, 0}
\colorlet{KITyellow}{kit-yellow}
\colorlet{kit-yellow100}{kit-yellow}
\colorlet{kit-yellow90}{kit-yellow!90!white}
\colorlet{kit-yellow80}{kit-yellow!80!white}
\colorlet{kit-yellow70}{kit-yellow!70!white}
\colorlet{kit-yellow60}{kit-yellow!60!white}
\colorlet{kit-yellow50}{kit-yellow!50!white}
\colorlet{kit-yellow40}{kit-yellow!40!white}
\colorlet{kit-yellow30}{kit-yellow!30!white}
\colorlet{kit-yellow25}{kit-yellow!25!white}
\colorlet{kit-yellow20}{kit-yellow!20!white}
\colorlet{kit-yellow15}{kit-yellow!15!white}
\colorlet{kit-yellow10}{kit-yellow!10!white}
\colorlet{kit-yellow5}{kit-yellow!5!white}

\definecolor{kit-orange}{RGB}{223, 155, 27}
\colorlet{KITorange}{kit-orange}
\definecolor{kit-orange}{RGB}{223, 155, 27}
\colorlet{kit-orange100}{kit-orange}
\colorlet{kit-orange90}{kit-orange!90!white}
\colorlet{kit-orange80}{kit-orange!80!white}
\colorlet{kit-orange70}{kit-orange!70!white}
\colorlet{kit-orange60}{kit-orange!60!white}
\colorlet{kit-orange50}{kit-orange!50!white}
\colorlet{kit-orange40}{kit-orange!40!white}
\colorlet{kit-orange30}{kit-orange!30!white}
\colorlet{kit-orange25}{kit-orange!25!white}
\colorlet{kit-orange20}{kit-orange!20!white}
\colorlet{kit-orange15}{kit-orange!15!white}
\colorlet{kit-orange10}{kit-orange!10!white}
\colorlet{kit-orange5}{kit-orange!5!white}

\definecolor{kit-lightgreen}{RGB}{140, 182, 60}
\colorlet{KITlightgreen}{kit-lightgreen}
\colorlet{kit-lightgreen100}{kit-lightgreen}
\colorlet{kit-lightgreen90}{kit-lightgreen!90!white}
\colorlet{kit-lightgreen80}{kit-lightgreen!80!white}
\colorlet{kit-lightgreen70}{kit-lightgreen!70!white}
\colorlet{kit-lightgreen60}{kit-lightgreen!60!white}
\colorlet{kit-lightgreen50}{kit-lightgreen!50!white}
\colorlet{kit-lightgreen40}{kit-lightgreen!40!white}
\colorlet{kit-lightgreen30}{kit-lightgreen!30!white}
\colorlet{kit-lightgreen25}{kit-lightgreen!25!white}
\colorlet{kit-lightgreen20}{kit-lightgreen!20!white}
\colorlet{kit-lightgreen15}{kit-lightgreen!15!white}
\colorlet{kit-lightgreen10}{kit-lightgreen!10!white}
\colorlet{kit-lightgreen5}{kit-lightgreen!5!white}

\definecolor{kit-purple}{RGB}{163, 16, 124}
\colorlet{KITpurple}{kit-purple}
\colorlet{kit-purple100}{kit-purple}
\colorlet{kit-purple90}{kit-purple!90!white}
\colorlet{kit-purple80}{kit-purple!80!white}
\colorlet{kit-purple70}{kit-purple!70!white}
\colorlet{kit-purple60}{kit-purple!60!white}
\colorlet{kit-purple50}{kit-purple!50!white}
\colorlet{kit-purple40}{kit-purple!40!white}
\colorlet{kit-purple30}{kit-purple!30!white}
\colorlet{kit-purple25}{kit-purple!25!white}
\colorlet{kit-purple20}{kit-purple!20!white}
\colorlet{kit-purple15}{kit-purple!15!white}
\colorlet{kit-purple10}{kit-purple!10!white}
\colorlet{kit-purple5}{kit-purple!5!white}

\definecolor{kit-brown}{RGB}{167, 130, 46}
\colorlet{KITbrown}{kit-brown}
\colorlet{kit-brown100}{kit-brown}
\colorlet{kit-brown90}{kit-brown!90!white}
\colorlet{kit-brown80}{kit-brown!80!white}
\colorlet{kit-brown70}{kit-brown!70!white}
\colorlet{kit-brown60}{kit-brown!60!white}
\colorlet{kit-brown50}{kit-brown!50!white}
\colorlet{kit-brown40}{kit-brown!40!white}
\colorlet{kit-brown30}{kit-brown!30!white}
\colorlet{kit-brown25}{kit-brown!25!white}
\colorlet{kit-brown20}{kit-brown!20!white}
\colorlet{kit-brown15}{kit-brown!15!white}
\colorlet{kit-brown10}{kit-brown!10!white}
\colorlet{kit-brown5}{kit-brown!5!white}

\definecolor{kit-cyan}{RGB}{35, 161, 224}
\colorlet{KITcyan}{kit-cyan}
\colorlet{KITcyanblue}{kit-cyan}
\colorlet{kit-cyan100}{kit-cyan}
\colorlet{kit-cyan90}{kit-cyan!90!white}
\colorlet{kit-cyan80}{kit-cyan!80!white}
\colorlet{kit-cyan70}{kit-cyan!70!white}
\colorlet{kit-cyan60}{kit-cyan!60!white}
\colorlet{kit-cyan50}{kit-cyan!50!white}
\colorlet{kit-cyan40}{kit-cyan!40!white}
\colorlet{kit-cyan30}{kit-cyan!30!white}
\colorlet{kit-cyan25}{kit-cyan!25!white}
\colorlet{kit-cyan20}{kit-cyan!20!white}
\colorlet{kit-cyan15}{kit-cyan!15!white}
\colorlet{kit-cyan10}{kit-cyan!10!white}
\colorlet{kit-cyan5}{kit-cyan!5!white}

\definecolor{kit-gray}{RGB}{0, 0, 0}
\colorlet{KITgray}{kit-gray}
\colorlet{kit-gray100}{kit-gray}
\colorlet{kit-gray90}{kit-gray!90!white}
\colorlet{kit-gray80}{kit-gray!80!white}
\colorlet{kit-gray70}{kit-gray!70!white}
\colorlet{kit-gray60}{kit-gray!60!white}
\colorlet{kit-gray50}{kit-gray!50!white}
\colorlet{kit-gray40}{kit-gray!40!white}
\colorlet{kit-gray30}{kit-gray!30!white}
\colorlet{kit-gray25}{kit-gray!25!white}
\colorlet{kit-gray20}{kit-gray!20!white}
\colorlet{kit-gray15}{kit-gray!15!white}
\colorlet{kit-gray10}{kit-gray!10!white}
\colorlet{kit-gray5}{kit-gray!5!white}

\definecolor{KITpalegreen}{RGB}{130,190,60}
\colorlet{kit-maigreen100}{KITpalegreen}
\colorlet{kit-maigreen70}{KITpalegreen!70}
\colorlet{kit-maigreen50}{KITpalegreen!50}
\colorlet{kit-maigreen30}{KITpalegreen!30}
\colorlet{kit-maigreen15}{KITpalegreen!15}

\usepackage{amssymb,amsthm,mathrsfs,bm,bbm,dsfont}
\usepackage{siunitx}
\usepackage{enumerate}
\usepackage{booktabs}
\usepackage{acronym}
\usepackage{diagbox}
\usepackage[noend]{algorithmic}
\usepackage{xcolor}
\usepackage{pgf, tikz, pgfplots}
\pgfplotsset{compat=1.17}
\usepackage{pgfplotstable}
\usepgfplotslibrary{statistics,groupplots}

\usetikzlibrary{
  arrows.meta, calc, fit, matrix, positioning, shapes, shadows,
  trees, mindmap, tikzmark, angles, quotes, babel, spy
}

\tikzset{
  myblock/.style={
    rectangle, draw, thin,
    minimum width=1.75cm, minimum height=0.72cm,
    font=\small, align=center
  },
  mywideblock/.style={
    rectangle, draw, thin,
    minimum width=3cm, minimum height=0.5cm,
    font=\small, align=center
  }
}

\newcommand{\myline}[2]{%
  \path (#1.east) -- (#2.west) coordinate[pos=0.4](mid);
  \draw[-latex] (#1.east) -| (mid) |- (#2.west);
}

\DeclareMathOperator*{\aC}{\mathcal{C}_\mathsf{a}}
\DeclareMathOperator*{\sC}{\mathcal{C}_\mathsf{s}}

\pgfplotscreateplotcyclelist{kitcolors_cycle}{
    {draw=kit-green100},
    {draw=kit-green80}, 
    {draw=kit-blue100},
    {draw=kit-blue80},
    {draw=kit-royalblue100},
    {draw=kit-royalblue80},
    {draw=kit-red100},
    {draw=kit-red80},
    {draw=kit-orange100},
    {draw=kit-orange80},
    {draw=kit-lightgreen100},
    {draw=kit-lightgreen80},
    {draw=kit-purple100},
    {draw=kit-purple80},
    {draw=kit-brown100},
    {draw=kit-brown80},
    {draw=kit-cyan100},
    {draw=kit-cyan80},
}
\pgfplotscreateplotcyclelist{allred}{
    {KITred}
}
\pgfplotscreateplotcyclelist{rbr}{
    {KITred},
    {KITred},
    {KITred},
    {KITred},
    {KITred},
    {KITorange},
    {KITred},
    {KITred},
    {KITred},
    {KITred},
    {KITred},
    {KITred}
}
\pgfplotscreateplotcyclelist{brr}{
    {KITorange},
    {KITred},
    {KITred},
    {KITred},
    {KITred},
    {KITred},
    {KITred},
    {KITred},
    {KITred},
    {KITred},
    {KITred},
    {KITred}
}

\acrodef{BP}{belief propagation}
\acrodef{LDPC}{low‑density parity‑check}%
\acrodef{RL}{raptor-like}
\acrodef{BCH}{Bose--Chaudhuri--Hocquenghem}
\acrodef{AED}{automorphism ensemble decoding}
\acrodef{MBBP}{multiple bases \ac{BP}}
\acrodef{SCED}{subcode ensemble decoding}
\acrodef{aSCED}{affine subcode ensemble decoding}
\acrodef{GAED}{generalized \ac{AED}}
\acrodef{EED}{endomorphism ensemble decoding}
\acrodef{SED}{scheduling ensemble decoding}
\acrodef{ML}{maximum likelihood}
\acrodef{SC}{successive cancellation}
\acrodef{RRD}{random redundant decoding}

\acrodef{BPSK}{binary phase‑shift keying}
\acrodef{BI-AWGN}{binary‑input additive white Gaussian noise}

\acrodef{SPA}{sum-product algorithm}
\acrodef{NMSA}{normalized min-sum algorithm}
\acrodef{NSPA}{normalized sum-product algorithm}
\acrodef{FER}{frame error rate}
\acrodef{FEP}{frame error probability}
\acrodef{LER}{list error rate}
\acrodef{LEP}{list error probability}
\acrodef{TEC}{total edge count}
\acrodef{SNR}{signal‑to‑noise ratio}
\acrodef{LLR}{log-likelihood ratio}

\acrodef{PCM}{parity‑check matrix}
\acrodef{GPCM}{generalized parity‑check matrix}
\acrodef{PCRB}{parity-check row block}
\acrodef{ssPCM}{structured sparse \ac{PCM}}
\acrodef{OSD}{ordered statistics decoding}

\acrodef{LC}{linear covering}

\newlength{\ferfigheight}      %
\begin{document}
\bstctlcite{IEEEexample:BSTcontrol}

\newif\ifshowrevisions
\showrevisionstrue    %
\showrevisionsfalse %

\newif\ifshowminrevisions
\showminrevisionstrue    %
\showminrevisionsfalse %

\ifshowrevisions
  \newcommand{\revi}[1]{\textcolor{KITgreen}{#1}}
  \newcommand{\revii}[1]{\textcolor{KITblue}{#1}}
  \newcommand{\reviii}[1]{\textcolor{KITred}{#1}}
  \newcommand{\reviv}[1]{\textcolor{KITbrown}{#1}}
\else
  \newcommand{\revi}[1]{#1}
  \newcommand{\revii}[1]{#1}
  \newcommand{\reviii}[1]{#1}
  \newcommand{\reviv}[1]{#1}
\fi

\ifshowminrevisions
  \newcommand{\minrevi}[1]{\textcolor{KITgreen}{#1}}
\else
  \newcommand{\minrevi}[1]{#1}
\fi

\newtheorem{theorem}{Theorem}
\newtheorem{lemma}{Lemma}
\newtheorem{proposition}{Proposition}
\newtheorem{corollary}{Corollary}
\newtheorem{remark}{Remark}
\newtheorem{designprinciple}{Design Principle}
\title{Affine Subcode Ensemble Decoding of\\ Linear Block Codes}
\newcommand{\jmcol}[1]{\textcolor{black}{#1}}
\newcommand{\pbcol}[1]{\textcolor{black}{#1}}
\author{Jonathan Mandelbaum,~\IEEEmembership{Graduate Student Member,~IEEE}, Paul Bezner,~\IEEEmembership{Graduate Student Member,~IEEE}, Holger~Jäkel,~\IEEEmembership{Member,~IEEE}, Stephan ten Brink,~\IEEEmembership{Fellow,~IEEE},  Laurent Schmalen,~\IEEEmembership{Fellow,~IEEE}%
\thanks{This work has received funding from the German Federal Ministry of
Research, Technology and Space (BMFTR) within the projects Open6GHub and
Open6GHub+ (grant agreements 16KISK010, 16KIS2405 and 16KISK019).
Parts of this paper have been
presented at the IEEE Int. Symp. Inf. Theory 2025 \cite{mandelbaum_subcode_2025} and IEEE Int.
Symp. on Topics in Coding 2025 \cite{sced_polar}.\\
Jonathan Mandelbaum, Holger Jäkel, and Laurent Schmalen are with the Communications Engineering Lab (CEL),
Karlsruhe Institute of Technology (KIT), Hertzstr. 16, 76187 Karlsruhe,
Germany (e-mail: first.lastname@kit.edu).
Paul Bezner and Stephan ten Brink are with the Institute of Telecommunications, Pfaffenwaldring 47, University of Stuttgart, 70569 Stuttgart, Germany (e-mail:
bezner,tenbrink@inue.uni-stuttgart.de)
}}%

\maketitle

\begin{abstract}
In the short block length regime, ensemble decoding schemes with their inherently parallel structure can improve error correction performance and reduce latency compared to stand-alone suboptimal decoders such as \ac{BP}.  
In this work, we introduce \ac{aSCED}, which uses an ensemble of decoders operating on linear block codes and both linear and strictly affine subcodes.  
This generalizes the recently proposed \ac{SCED}, which is restricted to linear subcodes.  
We derive \ac{BP} update rules for affine subcodes and show that \ac{aSCED} simplifies ensemble design compared to \ac{SCED}, multiple bases \ac{BP}, and automorphism ensemble decoding.  
Monte Carlo simulations of \revii{four} low-density parity-check (LDPC) codes and \revii{four} \ac{BCH} codes demonstrate improved error correction performance of \ac{aSCED} over competing existing ensemble schemes.  
\revii{ Notably, for one $5$G LDPC code and two \ac{BCH} codes, \ac{aSCED} combined with algorithms to construct high-performance parity-check matrices can achieve near-\acl{ML} performance.}
\end{abstract}

\begin{IEEEkeywords}
Belief Propagation Decoding, Ensemble Decoding, Universal Codes, Code Agnostic Decoding
\end{IEEEkeywords}

\section{Introduction}
\IEEEPARstart{L}{ow}-density parity-check (LDPC) codes decoded with \acf{BP}, a low-complexity iterative message passing algorithm,  achieve excellent error correction performance in the large block length regime~\cite{MCT08}.
However, in the short block length regime, the decoding performance of \ac{BP} becomes unsatisfactory \cite{10571997,8594709}.
This hinders the application of LDPC codes in emerging applications such as the Internet of Things (IoT) and ultra-reliable low-latency communication (URLLC) in future wireless systems.
In contrast, polar codes, proposed by Arıkan in 2009 \cite{polar:arikan09}, paired with \ac{SC} list decoding \cite{polar:talvardy:scl}, are well known for their strong performance at short block lengths.

There exist two independent key research directions to enable the use of \ac{BP} in the short block length regime:
\begin{enumerate}
    \item \textbf{\Ac{PCM} design}: Many short codes with large minimum Hamming distance, e.g., \acf{BCH} codes, exhibit dense Tanner graphs with short cycles that degrade \ac{BP} performance. 
    Different approaches to enhance BP performance aim at improving sparsity or eliminating short cycles in the graphical representation of the code~\cite{vasic_4cyc_orthogon,yifei_sspcm,yedidia_4cyc_removal,kumar_milen_4cyc_rem,yifei_univ_bp,kraft_setmatrix}.  
    \item \textbf{Ensemble decoding}: Employing multiple suboptimal decoding algorithms in parallel---not limited to \ac{BP}---and selecting the best estimate to improve overall the error correction performance \cite{MBBP1,stuttgart_ldpc_aed,krieg2024comparativestudyensembledecoding_scc}.
\end{enumerate}

Ensemble decoding can improve the performance in the short-block regime, enhancing both latency and \ac{FER} \cite{AED_RMcodes,stuttgart_ldpc_aed,MBBP2}. 
Prior work on ensemble decoding methods achieves diversity by decoding over modified representations of noise realizations or graphs \cite{mandelbaum_subcode_2025,stuttgart_ldpc_aed, HalfordRDD,krieg2024comparativestudyensembledecoding_scc,AED_RMcodes, geiselhart2023ratecompatible,8723089,7360541,MBBP1,MBBP2,MBBP3_withLeaking,mandelbaum2024endomorphisms,mandelbaum2023generalized}.

Homomorphism-based ensemble decoders, such as \ac{RRD}, \ac{AED}, \ac{GAED}, and \ac{EED}, leverage known code symmetries to transform the noise representation.
If the automorphism group is known and its effect is not absorbed by a symmetry of the decoder, \ac{AED} can yield performance gains \cite{AED_RMcodes,geiselhart2023ratecompatible,krieg2024comparativestudyensembledecoding_scc,stuttgart_ldpc_aed}.
 Related noise-altering ensemble decoding schemes include noise-aided ensemble decoding \cite{8723089} and saturated \ac{BP} \cite{7360541}.
Other methods, e.g., \ac{MBBP} and \ac{SED}, improve performance by decoding over ensembles of distinct graphs or \ac{BP} variants.
 \ac{MBBP} employs a set of equivalent PCMs, i.e., structurally different graphs, while \ac{SED} varies the scheduling order of layered \ac{BP} updates \cite{MBBP2,krieg2024comparativestudyensembledecoding_scc}.
However, when considering an arbitrary code, \ac{MBBP} relies on an NP-complete search for minimum-weight codewords of the dual code, severely limiting its applicability for codes of medium-to-long block lengths~\cite{intractabilityofminimumdistance,kraft_setmatrix}.

In \cite{mandelbaum_subcode_2025}, we introduced \acf{SCED}, an ensemble decoding scheme using different linear subcodes of the code $\mathcal{C}$.
PCMs of subcodes are obtained by appending \emph{linearly independent} rows to the original \ac{PCM}, thereby forming PCMs for \ac{BP}‑based \ac{SCED}.  
It is essential to ensure a \ac{LC}, i.e., that the union of the linear subcodes covers the code~$\mathcal{C}$\cite{clark_covering_numbers}; otherwise, some codewords are not decodable.

For arbitrary codes, \ac{SCED} simplifies ensemble design. Unlike \ac{MBBP}, it only requires identifying low-weight linearly independent rows \revi{that are appended to the \ac{PCM}}, which is easier than searching for minimum-weight dual codewords \revi{to construct the diverse multiple basis representations}.
Moreover, unlike \ac{AED}, \ac{SCED} does not require any knowledge of the code structure beyond its \ac{PCM}.
For the codes investigated in \cite{mandelbaum_subcode_2025}, Monte Carlo simulations confirm that \ac{SCED} consistently outperforms AED, which itself surpasses other ensemble methods~\cite{mandelbaum_subcode_2025,krieg2024comparativestudyensembledecoding_scc}.

\reviv{
In this work, our contributions are summarized as follows:
\begin{enumerate}[(i)]
    \item We generalize \ac{SCED} to \ac{aSCED}, where the ensemble of decoders does not only decode linear subcodes, but also suitably defined affine subcodes, i.e., linear subcodes that are subject to an affine offset.
    \item We show how BP decoding can be adapted to decode affine subcodes by mapping the affine offset to a fixed syndrome offset. The resulting update rules are closely related to syndrome-based BP decoding; however, in contrast to conventional syndrome-based decoding, the syndrome offset depends only on the affine subcode and is independent of the received sequence.
    \item We provide theoretical results for \ac{aSCED} showing that \ac{aSCED} reduces the number of subcodes required for coverage and simplifies the ensemble design compared to \ac{SCED}. Moreover, since the linear subcode and its
    associated affine subcodes share the same Tanner graph and differ only in their check-node offsets, \ac{aSCED} enables potential graph and hardware reuse.
    \item We clarify that the all-zero codeword assumption can be used for \ac{aSCED} in the absence of \ac{ML} errors. This provides a basis for efficient simulation and future analytical studies.
\end{enumerate}
The flexibility provided by affine subcodes proves beneficial for BP-based decoding of both LDPC and algebraic codes. Our numerical results can be summarized as follows:
\begin{enumerate}[(i)]
    \item For the $5$G LDPC code, \ac{aSCED} enables an ensemble design that leverages its Raptor-like (RL) structure.
    \item For BCH codes, we demonstrate that linear and affine subcodes can improve the design of PCMs suitable for BP by enriching the dual search space.
    \item For the considered BCH and 5G LDPC examples, \ac{aSCED} achieves substantial gains over stand-alone BP and competing ensemble decoders. In several cases, %
    the proposed decoders can achieve near-\ac{ML} performance.
\end{enumerate}
}

The remainder of this paper is organized as follows. In Sec.~\ref{sec:prelim}, we briefly review linear codes and \ac{BP} decoding, before introducing a general framework of ensemble decoding in Sec.~\ref{sec:ed}. %
In Sec.~\ref{sec:asced_sec}, we present \revi{aSCED} and prove some of its properties.
Additional structural benefits enabled by the usage of subcodes are discussed in Sec.~\ref{sec:structural_advantages}.
Numerical results supporting these structural advantages and demonstrating the improved error-correction performance of \ac{aSCED} are provided in Sec.~\ref{sec:res}.
Finally, Sec.~\ref{sec:conclusion} concludes the paper.

\textbf{Notation:} %
Vectors and matrices are denoted by lowercase and uppercase bold letters, respectively, e.g., $\bm{x}$ and $\bm{X}$.
The $i$th element of $\bm{x}$ is denoted by $x_i$, the $i$th row of $\bm{X}$ by $\bm{X}_{i,:}$, and the element in row $i$ and column $j$ of $\bm{X}$ by $X_{i,j}$.
Sets are denoted by calligraphic letters, e.g.,  $\mathcal{S}$, and we use the notation $[m]:=\{1,\ldots,m\}$.
For index sets $\mathcal{A}$ and $\mathcal{B}$, we define the submatrix $\bm{M}_{\mathcal{A},\mathcal{B}}:=(M_{i,j})_{i\in\mathcal{A},j\in\mathcal{B}}$, 
the subvector $\bm{x}_\mathcal{A}=(x_i)_{i\in \mathcal{A}}$
and the vector $\bm{1}_\mathcal{A}:=\sum_{i\in \mathcal{A}}\bm{e}_i$, where $\bm{e}_i$ denotes the $i$th canonical unit vector.
Finally, $\bm{1}_{a\times b}$ and $\bm{0}_{a\times b}$ denote the all-one and all-zero matrices of size $a \times b$, respectively.

\section{Preliminaries}\label{sec:prelim}

\subsection{Linear Codes and Subcodes}
A linear block code ${\mathcal{C}(n,k)=:\mathcal{C}}$ with $ {k\leq n}$ is a \mbox{$k$-dimensional} subspace of the vector space $\mathbb{F}^n$ over a field~$\mathbb{F}$. When clear from the context, the block length ${n\in \mathbb{N}}$ and information length ${k\in \mathbb{N}}$ are omitted.
Every linear code can be defined as the span of the rows of a generator matrix $\bm{G}\in\mathbb{F}^{k\times n}$ or by the null space of a \ac{PCM} $\bm{H}\in\mathbb{F}^{m\times n}$, i.e., 
\[
\mathcal{C}:=\left\{\bm{uG}:\bm{u}\in \mathbb{F}^k\right\}=\left\{\bm{x}\in \mathbb{F}^n:\bm{H}\bm{x}^\mathsf{T}=\bm{0}\right\},
\]
where $\mathrm{rank}(\bm{G})=k$ and $\mathrm{rank}(\bm{H})=n-k\leq m$. We denote the code characterized by $\bm{H}$ as $\mathcal{C}(\bm{H})$. A \ac{PCM} with $m>\mathrm{rank}(\bm{H})$ is called \emph{overcomplete}.  The \revi{\emph{weight} of $\bm{H}$, or \emph{total edge count} of the respective Tanner graph, is} ${\mathrm{wt}(\bm{H}):=\sum_{j,i } H_{ji}}$. 
The dual code of $\mathcal{C}$, denoted by $\mathcal{C}^\bot$, is defined by interchanging the generator matrix and the \ac{PCM}, i.e., $\bm{H}^\bot=\bm{G}$. 
We refer to codewords of the dual code simply as \emph{dual codewords}.

Following \cite{mandelbaum_subcode_2025}, a \emph{linear subcode}
$\sC(n,k')$ of $\mathcal{C}$ with $k'\leq k$ is a $k'$-dimensional subspace of $\mathcal{C}$.
We define an \emph{strictly affine subcode} $\aC(n,k')$ of $\mathcal{C}$ by a tuple $(\sC(n,k'),\bm{x}_\mathsf{a})$ as
\begin{equation}\label{eq:def:aff_code}
\aC(n,k'):=\left\{\bm{x}_\mathsf{a}+\bm{x}:\bm{x}\in \sC(n,k')\right\}=:\bm{x}_\mathsf{a}+ \sC(n,k'),    
\end{equation}
where $\sC(n,k')$ is a linear subcode of $\mathcal{C}$ and the \emph{affine offset} $\bm{x}_\mathsf{a}$ satisfies
${\bm{x}_\mathsf{a}\in \mathcal{C}}$ but $\bm{x}_\mathsf{a}\notin \sC(n,k')$ requiring $k'<k$.  
Thus, a strictly affine subcode $\aC$ coincides with a coset within the quotient space $\mathcal{C}/\sC$. %
Let $\bm{H}_\mathsf{s}$ denote a \ac{PCM} of $\sC$, then we define the \emph{affine syndrome} ${\bm{s}_\mathsf{a}:=\bm{H}_\mathsf{s}\bm{x}_\mathsf{a}^\mathsf{T}}$.
Since every codeword $\bm{x}\in\aC$ can be written as $\bm{x}=\bm{x}_\mathsf{a}+\bm{x}_\mathsf{s}$ with $\bm{x}_\mathsf{s}\in\sC$, all codewords within the strictly affine subcode have the same non-zero affine syndrome, i.e., ${\bm{H}_\mathsf{s}\bm{x}^\mathsf{T}=\bm{s}_\mathsf{a}}\neq \bm{0}$ for all $\bm{x}\in\aC$.

Linear subcodes and strictly affine subcodes are subsets of $\mathcal{C}$, i.e., $\sC\subseteq_\mathsf{s} \mathcal{C}$ and $\aC\subset_\mathsf{a} \mathcal{C}$, where the subscripts indicate the relation of a linear and a strictly affine subcode to the original code $\mathcal{C}$, respectively. 
A linear subcode $\sC(n,k')$ is a \emph{proper} subcode if $k'<k$. In this case, we replace $\subseteq_\mathsf{s}$ by $\subset_\mathsf{s}$.
 In contrast, a strictly affine subcode does not constitute a subspace of $\mathcal{C}$ since it does not contain the all-zero codeword and therefore is not closed under addition.
Throughout this paper, we exclusively consider linear subcodes and strictly affine subcodes. 
For brevity, we omit the term strictly and refer to the latter simply as \emph{affine subcodes}.

We consider the transmission of a binary codeword ${\bm{x}\in\mathcal{C}}$ over a binary-input memoryless symmetric-output channel (BiMSC) with output alphabet $\mathcal{Y}$  \cite[Ch.~4]{MCT08}.
At the receiver, the corresponding channel output is denoted by ${\bm{y}\in\mathcal{Y}^n}$ and can be mapped onto the bit-wise \ac{LLR} vector $\bm{\lambda}:=(l(y_i))_{i=1}^n\in \mathbb{R}^n$ with $l(y_i)$ given by \cite[Def.~4.5]{MCT08}.
\subsection{Belief Propagation Decoding}\label{sec:prelim:bp}
\ac{BP} is an iterative message passing algorithm operating on the Tanner graph of a code\cite{MCT08}.
A Tanner graph, associated with a \ac{PCM} $\bm{H}$ of a linear code, is a bipartite graph, consisting of variable nodes (VNs) and check nodes (CNs). A VN $\mathsf{v}_i$ is associated with the $i$th code bit, and a CN $\mathsf{c}_j$ is associated with the $j$th parity check equation. 
An edge connects VN $\mathsf{v}_i$ and CN $\mathsf{c}_j$ iff $H_{j,i}\neq0$. 
Define ${\mathcal{N}(i):=\{j\in [m]:H_{j,i}\neq0\}}$ and ${\mathcal{M}(j):=\{i\in [n]:H_{j,i}\neq0\}}$. 

During \ac{BP}, \ac{LLR} messages computed by the nodes are iteratively passed along the edges of the Tanner graph. 
For each ${i\in[n]}$, the VN update equation yields the VN-to-CN messages 
\[
\lambda_{i\rightarrow j}=\lambda_i+\sum_{j'\in \mathcal{N}(i)\setminus \{j\}} \lambda_{i\leftarrow j'}=:\lambda_i^{\mathrm{tot}}-\lambda_{i\leftarrow j}\,,\,
\, j\in \mathcal{N}(i),
\]
where $\lambda_i = l(y_i)$ denotes the bit-wise \ac{LLR}, $\lambda_{i\leftarrow j}$ is the message passed from CN $\mathsf{c}_j$ to VN $\mathsf{v}_i$, and  
$\lambda_i^{\mathrm{tot}}$ represents the \emph{a posteriori \ac{LLR} estimate} at VN $\mathsf{v}_i$.

Several \ac{BP} decoding variants exist, differing primarily in the CN update rule.
In this work, we consider the \ac{NSPA} and the \ac{NMSA}. In the \ac{NSPA}, for each CN $\mathsf{c}_j$, the CN-to-VN messages are given by
\[
\lambda_{i\leftarrow j}=\alpha\cdot 2\cdot\tanh^{-1}\!\!{\left(
\prod_{i'\in\mathcal{M}(j)\setminus \{i\}}\!\!\!\!\!\tanh{
\left(\frac{\lambda_{i'\rightarrow j}}{2}\right)}\right)},\,\, i\in \mathcal{M}(j),
\]
where $\alpha\in \mathbb{R}$ is a normalization constant. For $\alpha=1$, we recover to the \ac{SPA} \cite{MCT08}.

The \ac{NMSA}, which is particularly suited for hardware implementation, replaces the CN update rule of CN $\mathsf{c}_j$ by \cite{de_nmsa_fossorier}
\[
\lambda_{i\leftarrow j}=\alpha \cdot\!\!\!\!\!\!\!\! \prod_{i'\in\mathcal{M}(j)\setminus \{i\}} \!\!\!\!\!\!\!\!\mathrm{sign}\left(\!\lambda_{i'\rightarrow j}\!\right)\cdot\!\! \min_{i'\in\mathcal{M}(j)\setminus \{i\}}\!\left|\lambda_{i'\rightarrow j} \right|,\, i\in \mathcal{M}(j),
\]
where $\alpha\in \mathbb{R}$ is a normalization constant. 
In a flooding schedule, after initializing all VN-to-CN messages as $\lambda_{i\rightarrow j}=\tilde{\lambda}_i$, one \ac{BP} iteration consists of updating all CNs, followed by updating all VNs.
After each iteration, the current codeword estimate $\hat{\bm{x}}$ is obtained by the hard decision
\[
\hat{x}_i=
\begin{cases}
    1,&\lambda_i^{\mathrm{tot}}<0,\\
    0,& \text{otherwise},
\end{cases}\quad i \in \left[n\right].
\]
The iterative procedure continues until $\bm{H}\hat{\bm{x}}^\mathsf{T}=\bm{0}$
or until a maximum number of iterations has been conducted. If no valid codeword is found, the current hard decision is returned as the codeword estimate. 
We denote this estimate resulting from \ac{BP} using $\bm{H}$ as $\hat{\bm x}:=\text{BP}(\bm{H},\bm{y})=\text{BP}(\bm{H},\bm{\lambda})$.

\section{Ensemble Decoding}\label{sec:ed}

\subsection{Basics of Ensemble Decoding}\label{sec:ed:basics}

Ensemble decoding schemes aim at improving the performance of a single suboptimal decoding algorithm by employing multiple parallel decoding paths\cite{krieg2024comparativestudyensembledecoding_scc}.
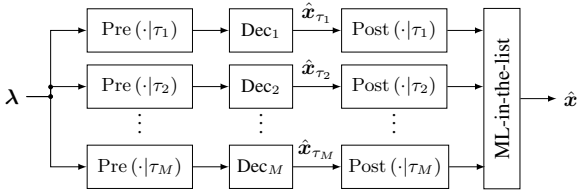
\begin{figure}[t]
    \centering
    \scalebox{0.8}{
\begin{tikzpicture}[
  node distance=2mm and 6mm,
  >=latex
]

  \node (input) {$\bm{\lambda}$};

  \node[myblock, right=1cm of input, yshift=1.1cm] (pi1) {$\mathrm{Pre}\left(\cdot|\tau_1\right)$};
  \node[myblock, below=of pi1] (pi2) {$\mathrm{Pre}\left(\cdot|\tau_2\right)$};
  \node[myblock, below=of pi2, yshift=-4mm] (pik) {$\mathrm{Pre}\left(\cdot|\tau_M\right)$};

  \node[rectangle, draw, thin,
    minimum width=1.05cm, minimum height=0.72cm,
    font=\small, align=center, right=of pi1] (debox_1) {$\text{Dec}_1$};
  \node[rectangle, draw, thin,
    minimum width=1.05cm, minimum height=0.72cm,
    font=\small, align=center, right=of pi2] (debox_2) {$\text{Dec}_2$};
  \node[rectangle, draw, thin,
    minimum width=1.05cm, minimum height=0.72cm,
    font=\small, align=center, right=of pik] (debox_K) {$\text{Dec}_M$};

  \node[myblock, right=0.8cm of debox_1] (pi_inv_1) {$\mathrm{Post}\left(\cdot|\tau_1\right)$};
  \node[myblock, right=0.8cm of debox_2] (pi_inv_2) {$\mathrm{Post}\left(\cdot|\tau_2\right)$};
  \node[myblock, right=0.8cm of debox_K] (pi_inv_k) {$\mathrm{Post}\left(\cdot|\tau_M\right)$};

  \node[right=of pi_inv_1] (ml_1) {};
  \node[right=of pi_inv_2] (ml_2) {};
  \node[right=of pi_inv_k] (ml_K) {};

  \coordinate (midy) at ($(pi1.north)!0.5!(pik.south)$);
  \node[
    rectangle, draw, thin,
    minimum width=0.6cm,
    minimum height=3cm,
    right=of pi_inv_2
  ] (ml_box)
  at (midy -| pi_inv_2.east)
  {\rotatebox{90}{ML-in-the-list}};
  \node[right=of ml_box] (output) {$\hat{\bm{x}}$};

  \node[below=of pi2, yshift=0.45cm] (dots_pi) {$\vdots$};
  \node[below=of debox_2, yshift=0.45cm] (dots_dec) {$\vdots$};
  \node[below=of pi_inv_2, yshift=0.45cm] (dots_inv) {$\vdots$};

  \myline{input}{pi1};
  \myline{input}{pi2};
  \myline{input}{pik};

  \filldraw[black] (0.458cm+0.23cm*0.8,0.16cm) circle (0.9pt);
  \filldraw[black] (0.458cm+0.23cm*0.8,0) circle (0.9pt);

  \draw[-latex] (pi1) -- (debox_1);
  \draw[-latex] (pi2) -- (debox_2);
  \draw[-latex] (pik) -- (debox_K);

  \draw[-latex] (debox_1) -- (pi_inv_1) node [midway, above, anchor=south] {$\hat{\bm{x}}_{\tau_1}$};
  \draw[-latex] (debox_2) -- (pi_inv_2) node [midway, above, anchor=south] {$\hat{\bm{x}}_{\tau_2}$};
  \draw[-latex] (debox_K) -- (pi_inv_k) node [midway, above, anchor=south] {$\hat{\bm{x}}_{\tau_M}$};

  \draw[-latex] (pi_inv_1) -- (ml_1);
  \draw[-latex] (pi_inv_2) -- (ml_2);
  \draw[-latex] (pi_inv_k) -- (ml_K);

  \draw[-latex] (ml_box) -- (output);

\end{tikzpicture}
}
    \caption{Structure of parallel paths employed in ensemble decoding.}
    \label{fig:general_ed}
\end{figure}
Fig.~\ref{fig:general_ed} illustrates the general structure of ensemble decoding. The \ac{LLR} vector $\bm{\lambda}$ is provided as the input to $M$ parallel paths.
Each path, indexed by ${\ell\in [M]}$, comprises a preprocessing function ${\mathrm{Pre}\left(\cdot|\tau_\ell\right):\mathbb{R}^n\rightarrow \mathbb{R}^n}$, a decoder ${\mathrm{Dec}_\ell:\mathbb{R}^n\rightarrow\mathbb{F}_2^n}$, and a corresponding postprocessing function ${\mathrm{Post}\left(\cdot|\tau_\ell\right):\mathbb{F}_2^n\rightarrow\mathbb{F}_2^n}$. 
Hereby, the pre- and postprocessing functions are based on a homomorphism $\tau_\ell:\mathcal{C}\rightarrow\mathcal{C}$, and we denote $\bm{x}_{\tau_\ell}:=\tau_\ell(\bm{x})$. 
The preprocessing function aims at reproducing the effect of $\tau_\ell(\cdot)$ on the transmitted codeword $\bm{x}$ in the \ac{LLR} domain. 
Thus, the $\ell$th decoder returns an estimate $\hat{\bm{x}}_{\tau_\ell}$ for $\bm{x}_{\tau_\ell}$ and the postprocessing identifies $\hat{\bm{x}}_\ell$ such that $\tau_\ell(\hat{\bm{x}}_\ell)=\hat{\bm{x}}_{\tau_\ell}$ \cite{mandelbaum2023generalized}.

Typically, an ensemble decoding scheme either
\begin{enumerate}
    \item[(i)] uses $M$ different decoders $\mathrm{Dec}_\ell$, but omits pre- and postprocessing, i.e., ${\tau_\ell=\mathrm{id},\, \ell\in[M]}$, or
    \item[(ii)] employs the same decoder across all paths, i.e., ${\mathrm{Dec}_\ell=\mathrm{Dec},\, \ell\in[M]}$, but uses $M$ distinct homomorphisms $\tau_\ell$ with associated pairs of pre- and postprocessing functions $(\mathrm{Pre}(\bm{\lambda}|\tau_\ell),\mathrm{Post}(\hat{\bm{x}}_{\tau_\ell}|\tau_\ell))$, $\ell\in[M]$.
\end{enumerate}

 Following \cite{mandelbaum_subcode_2025}, ensemble decoding generates a list
\[
\mathcal{L}:=\begin{cases}
    \{\hat{\bm{x}}_\ell:\ell\in[M],\,\,\hat{\bm{x}}_\ell\in \mathcal{C}\},& \text{if } \exists \ell\in[M]:\hat{\bm{x}}_\ell\in \mathcal{C} \\
    \{\hat{\bm{x}}_\ell:\ell\in[M]\},& \text{otherwise,}
\end{cases}
\]
and then, using the LLR vector $\bm{\lambda}$, selects the final estimate $\hat{\bm{x}}$ according to an {\emph{\ac{ML}-in-the-list rule}}\cite{AED_RMcodes}
\begin{align*}
    \hat{\bm{x}}:= \arg \max_{\bm{x}\in \mathcal{L}}(\bm{1}_{1\times n}-2\bm{x)}^\mathsf{T}\bm{\lambda}.%
\end{align*}

Next, we characterize several ensemble decoding schemes based on how they select distinct decoders and pre- and postprocessing functions to enhance the decoding performance.

\subsection{Multiple Basis Belief Propagation Decoding}
\ac{MBBP} decoding was first proposed by Hehn et al. in \cite{MBBP2} and later extended in \cite{MBBP1,MBBP3_withLeaking}. 
In \ac{MBBP}, the ensemble consists of $M$ parallel \ac{BP} decoders, each instantiated with a distinct \ac{PCM} $\bm{H}_\ell$, resulting in ${\mathrm{Dec}_\ell=\mathrm{BP}(\bm{H}_\ell):\mathbb{R}^n\rightarrow\mathbb{F}_2^n}$. Note that \ac{MBBP} employs no pre- or post-processing functions, so it belongs to category (i).

\ac{MBBP} requires $M$ PCMs suitable for \ac{BP} decoding, i.e., $M$ \jmcol{matrices ideally composed of minimum-weight dual codewords}. In general, for an arbitrary code, finding such matrices is an NP-complete problem\cite{intractabilityofminimumdistance}.
However, for short- to medium-length codes, efficient algorithms for finding low-weight codewords exist, e.g., \cite{leon_search}.
In particular, Hehn et al.\ investigate overcomplete \ac{PCM} designs and information exchange between parallel \ac{BP} decoders \cite{MBBP1,MBBP2}.
More recently, Shen et al. showed that \ac{MBBP}, in combination with a refined \ac{PCM} construction, achieves strong error correction for various short non-LDPC codes \cite{yifei_univ_bp}.

\subsection{Homomorphism-based Ensemble Decoding}
\revi{
The use of code automorphisms to improve decoding performance has a long history, dating back to the classical permutation decoding of MacWilliams for hard-decision decoding of systematic codes \cite{permutation_decoding_macwilliams}.
Subsequently, automorphisms have been exploited in various iterative decoding schemes.
For example, \ac{RRD} repeatedly applies code automorphisms to modify the Tanner graph during \ac{BP} \cite{HalfordRDD}, while the improved \ac{RRD} algorithm extends this concept to an \ac{MBBP}-like ensemble decoding architecture \cite{dimnik_Improved_random_redundant}. Furthermore, automorphisms have been used to generate redundant \acp{PCM} for \ac{MBBP} decoding of cyclic high-density parity-check codes \cite{MBBP1}.
}

\revi{More recently,} the introduction of \ac{AED} for Reed--Muller and polar codes \cite{AED_RMcodes} has \revi{renewed the interest in} homomorphism-based ensemble decoding schemes for modern wireless communications, e.g., \cite{aed_6g_urllc}.
Such schemes belong to category (ii), employing the same decoder in each path, i.e., $\mathrm{Dec}_\ell=\mathrm{Dec}$ for all $\ell \in [M]$, while selecting $M$ distinct homomorphisms $\tau_\ell$.
\ac{AED} builds on the definition of the automorphism group of a code $\mathcal{C}$ prevalent in coding theory
\[
\mathrm{Aut}(\mathcal{C}):=  \!  
\left\{\! \tau^{(a)}: \mathcal{C}\!\to\!\mathcal{C}, \bm{x}\!\mapsto\!a\tau(\bm{x})\!:\! \tau\!\in \mathrm{S}_n, a\in\mathbb{F}\!\setminus\!\{0\}    \right\},
\]
where $\mathrm{S}_n$ denotes the \emph{symmetric group} of $[n]$ and ${a\tau(\bm{x}) := a \cdot (x_{\tau(1)}, \ldots, x_{\tau(n)})^\mathsf{T}}$. %
For binary codes, ${a=1}$, and we simply write ${\tau:=\tau^{(1)}}$\cite{mandelbaum2023generalized}.
Hence, \ac{AED} employs homomorphisms $\tau_\ell$ that are permuting codewords. 
For each path $\ell\in[M]$, a distinct permutation automorphism ${\tau_\ell\in \mathrm{Aut}(\mathcal{C})}$ is selected
and the pre- and postprocessing functions are 
${\mathrm{Pre}(\bm{\lambda}|\tau_\ell):=\left(\lambda_{\tau_\ell(i)}\right)_{i=1}^n}$ 
and ${\mathrm{Post}(\bm{x}):=\tau_\ell^{-1}(\bm{x})}$, respectively, where ${\tau_\ell^{-1}\in \mathrm{Aut}(\mathcal{C})}$ is the inverse of $\tau_\ell$.

\ac{AED} can yield substantial decoding gains compared to stand-alone decoding, provided that the effects of the selected automorphisms are not absorbed by the symmetries of the decoders \cite{AED_RMcodes,PolarAEDPillet,stuttgart_ldpc_aed}.
The authors in \cite{stuttgart_ldpc_aed} investigate \ac{AED} to enhance stand-alone \ac{BP} decoding of quasi-cyclic (QC) LDPC codes.
Due to the QC construction of their \ac{PCM} $\bm{H}$, and noting that ${\mathrm{Aut}(\mathcal{C})=\mathrm{Aut}(\mathcal{C^\bot})}$\cite{MacWilliamsSloane}, certain permutation automorphisms are known a priori.
It was observed that the \ac{AED} path using $\mathrm{Dec}=\mathrm{BP}(\bm{H})$ and $\tau_\ell$ is equivalent to \ac{BP} decoding with $\tau_\ell^{-1}(\bm{H})$, where the QC permutation $\tau_\ell^{-1}$ is permuting columns of $\bm{H}$, while both pre- and postprocessing become the identity map $\mathrm{id}$ \cite{EnhancingCyclic}. 
For QC LDPC codes $\tau_\ell^{-1}(\bm{H})$ is typically equivalent to $\bm{H}$ up to row permutations, such that $\mathrm{BP}(\bm{H})$ and $\mathrm{BP}(\tau_\ell^{-1}\bm{(\bm{H})})$ yield identical decoding results preventing an ensemble gain.
Hence, \ac{BP} decoding based on $\bm{H}$ is equivariant with respect to these QC automorphisms.
A similar effect can occur for polar codes under \ac{SC} decoding \cite{AED_RMcodes}.
Nevertheless, \cite{stuttgart_ldpc_aed} proposes to break the decoder path equivariance for QC LDPC codes by modifying $\bm{H}$, e.g., by removing a row or appending a new linearly dependent row.

Recent works generalize the concept of \ac{AED}.
For instance, by dropping the constraint that $\tau$ must be a permutation, \ac{GAED} \cite{mandelbaum2023generalized} adopts the more general notion of the automorphism group prevalent in linear algebra.
\Ac{EED} further removes the requirement of bijectivity and instead relies on mappings from the \emph{set of endomorphisms}~\cite{mandelbaum2024endomorphisms} 
\[
\mathrm{End}(\mathcal{C}):=    
        \left\{ \tau: \mathcal{C}\rightarrow\mathcal{C},\,\, \text{ $\tau$ linear} \right\}.
\]
Both \ac{GAED} and \ac{EED} require a preprocessing to mimic the effect of the respective linear map in the \ac{LLR} domain. 
This preprocessing causes information loss, which limits the applicability of \ac{GAED} and EED to arbitrary codes.

\subsection{Linear Subcode Ensemble Decoding}
\subsubsection{Principle Idea}
\ac{SCED} \cite{mandelbaum_subcode_2025} belongs to ensemble decoding schemes of category~(ii).
In particular, \ac{SCED} possibly uses proper linear subcodes ${\mathcal{C}_\ell\subset_\mathsf{s} \mathcal{C}}$ and their respective decoder $\mathrm{Dec}_\ell$ in the parallel decoding paths~\cite{mandelbaum_subcode_2025}.
Note that for every path $\ell\in[M]$, there exist codewords $\bm{x}\in \mathcal{C}$ such that ${\bm{x}\notin \mathcal{C}_\ell}$, implying that the $\ell$th path cannot decode this codeword $\bm{x}$.
To ensure that all codewords can potentially be decoded, the ensemble of linear subcodes should constitute a \ac{LC} of $\mathcal{C}$.
Following \cite{clark_covering_numbers}, we define the \ac{LC} of a code $\mathcal{C}$ as a set of subcodes $\{\mathcal{C}_\ell\subseteq_\mathsf{s}\mathcal{C}:\ell\in {[M]}\}$ satisfying 
\begin{equation}
\bigcup_{\ell=1}^M \mathcal{C}_\ell= \mathcal{C},\label{eq:general_cover}
\end{equation}
where $M\geq3$ needs to hold \cite{clark_covering_numbers}.
If a well-performing decoder for $\mathcal{C}$ exists, one may include ${\mathcal{C}_1=\mathcal{C}}$ to guarantee an \ac{LC}\cite{mandelbaum_subcode_2025}.
The motivation for using subcode decoding is twofold. 
First, the minimum Hamming distance of an arbitrary subcode ${\mathcal{C}_\ell\subseteq \mathcal{C}}$
 is lower bounded by the minimum Hamming distance of $\mathcal{C}$.
 Hence, a subcode $\mathcal{C}_\ell$ can eventually correct more errors than the original code~$\mathcal{C}$.
  Second, efficient subcode decoders may exist, which can be leveraged within an ensemble.

\jmcol{Note that for ${\tau_\ell\in \mathrm{End}(\mathcal{C})}$, it holds that $\tau_\ell(\mathcal{C})\!\subseteq_\mathsf{s}\!\mathcal{C}$ \cite{mandelbaum2024endomorphisms}. 
Thus, \ac{EED} operates over linear subcodes of $\mathcal{C}$ while achieving ensemble diversity through distinct endomorphisms. 
In contrast, \ac{SCED} explicitly employs distinct linear subcode decoders avoiding the information loss induced by preprocessing.}

\subsubsection{\ac{BP}-based \ac{SCED}}
In particular, \ac{SCED} with \ac{BP} decoding is investigated in \cite{mandelbaum_subcode_2025}.
In this setting, each path ${\ell\in[M]}$ employs a \ac{PCM} $\bm{H}_\ell$ describing a subcode $\mathcal{C}_\ell$, such that ${\mathrm{Dec}_\ell=\mathrm{BP}(\bm{H}_\ell)}$.
Given the \ac{PCM} $\bm{H}$ of a binary code $\mathcal{C}$, the \ac{PCM} $\bm{H}_\ell$ of a subcode ${\mathcal{C}(\bm{H}_\ell) \subseteq_\mathsf{s} \mathcal{C}}$ can be constructed by appending to $\bm{H}$ a matrix $\bm{M}_\ell\in \mathbb{F}_2^{m_\ell\times n}$ whose $m_\ell$ rows are possibly linearly independent of the rows of $\bm{H}$, i.e., 
\begin{equation}\label{eq:inducing_subcode}
\bm{H}_\ell =
\begin{pmatrix}
\bm{H} \\
\bm{M}_\ell
\end{pmatrix}
=:
[\bm{H} ; \bm{M}_\ell],
\end{equation}
where $[\bm{H} ; \bm{M}_\ell]$ denotes the vertical stacking of matrices. 
Furthermore, we define the \emph{rank deficiency} of the subcode as \[
\Delta_\ell:=\mathrm{rank}(\bm{H}_\ell)-\mathrm{rank}(\bm{H})\geq 0.\]
Then, $\mathcal{C}_\ell=\mathcal{C}(\bm{H}_\ell)$ is a $(k-\Delta_\ell)$-dimensional linear subcode of $\mathcal{C}$.
Thus, if there exists ${i\in[m_\ell]}$ such that the $i$th row of $\bm{M}_\ell$ is linearly independent of all rows in $\bm{H}$, then $\Delta_\ell>0$
and $\mathrm{BP}(\bm{H}_\ell)$ is a decoder of a proper subcode $\mathcal{C}_\ell\subset_\mathsf{s}\mathcal{C}$.
\pbcol{If $\Delta_\ell=0$ for all $\ell \in[M]$, then \ac{BP}-based \ac{SCED} reduces to \ac{MBBP}, highlighting that \ac{BP}-based \ac{SCED} naturally generalizes \ac{MBBP}.} 
Importantly, in contrast to \ac{MBBP}, \ac{SCED} allows appending linearly independent rows to create diversity.
Thereby, it can avoid the exhaustive search of
minimum-weight dual codewords, while still providing ensemble gains \cite{mandelbaum_subcode_2025}.
\revi{When an ensemble is constructed from different basis representations of the subcode, an exhaustive search for low-weight dual codewords remains necessary.}
\revi{However, in this case} the appended rows can yield further structural advantage for \ac{BP} decoding since the search is performed over $\mathcal{C}_\ell^\bot$ rather than $\mathcal{C}^\bot$, \revi{which may have a more favorable weight-distribution}, as shown in Sec.~\ref{sec:structural_advantages}.

For LDPC codes, appending even a single linearly independent row, i.e., $m_\ell=1$, already yields ensembles with substantial decoding performance gains \cite{mandelbaum_subcode_2025}.
This aligns with \cite{globecom_when_does}, where Laendner et al. observed that a single linearly \emph{dependent} row can disrupt trapping sets, an effect that may also occur when appending linearly \emph{independent} rows.

However, the design of \ac{SCED} ensembles still faces certain challenges.
For instance, some codewords, e.g., the all-zero codeword, are decodable by a larger number of decoding paths than others \cite{mandelbaum_subcode_2025}. 
To overcome this imbalance and to further simplify the ensemble design while retaining the advantages of \ac{SCED}, we seek an approach that operates on subcodes while guaranteeing that every codeword is recoverable by the same number of decoding paths.
To this end, we propose \ac{aSCED}, which extends \ac{SCED} by decoding over affine subcodes.

\section{Affine Subcode Ensemble Decoding}\label{sec:asced_sec}
\subsection{Idea and Some Properties of \ac{aSCED}}\label{sec:asced}
We propose \ac{aSCED} for an arbitrary linear code $\mathcal{C}$. The central idea is that \ac{aSCED} generalizes \ac{SCED} by allowing affine subcodes, in addition to linear subcodes, each with its respective decoder.
Hence, \ac{aSCED} belongs to category~(i) and employs $M$ distinct decoders for subcodes $\mathcal{C}_\ell$ that are either linear or affine.
As in \ac{SCED}, to ensure that all codewords can potentially be decoded, we require that the set ${\{\mathcal{C}_\ell\subseteq\mathcal{C}:\ell\in[M]\}}$ of linear and affine subcodes satisfies~(\ref{eq:general_cover}). 
As shown next, the use of affine subcodes yields beneficial properties compared to \ac{SCED}.

\begin{theorem}\label{prop:affine_cover}
Let $\mathcal{C}$ be a linear binary code.
Then, there exists a covering satisfying~(\ref{eq:general_cover}) using  $M=2$ proper subcodes.
This is the minimum number of proper subcodes for a covering and can only be attained if the covering consists of a linear and an affine subcode.
\end{theorem}
\begin{IEEEproof}
Since ${\sC(n,k-1)}$ is a $(k-1)$-dimensional subspace of $\mathcal{C}$, the set of cosets with respect to ${\sC(n,k-1)}$, contains exactly two elements due to Lagrange's Theorem.
Since the coset $\bm{0} + {\sC(n, k-1)}$ constitutes the neutral element of the quotient space, \reviii{every element ${\bm{x}_\mathsf{a}\in \mathcal{C}\setminus{\sC(n, k-1)}}$ can be used as representative of the second coset} such that ${\mathcal{C}/{\sC(n,k-1)}=\{\bm{0}+ {\sC(n, k-1)}, \bm{x}_\mathsf{a}+\sC(n,k-1)\}}$.
Defining $\mathcal{C}_1=\bm{0}+{\sC(n,k-1)}$ and $\mathcal{C}_2=\bm{x}_\mathsf{a}+{\sC(n,k-1)}$, we finally obtain a covering $\mathcal{C}=\mathcal{C}_1\cup\mathcal{C}_2$ consisting of a proper linear subcode $\mathcal{C}_1$ and a proper affine subcode $\mathcal{C}_2$.
Obviously, a single proper subcode cannot satisfy~(\ref{eq:general_cover}). Note that neither two proper linear subcodes nor two affine subcodes can satisfy ~(\ref{eq:general_cover}).
The latter cannot cover the all-zero codeword, and the former fact is proven in \cite{clark_covering_numbers}. 
\end{IEEEproof}
Theorem~\ref{prop:affine_cover} demonstrates that the smallest number of subcodes needed to satisfy~(\ref{eq:general_cover}) is $M=2$, i.e., when allowing affine subcodes, $M$ can be reduced relative to $M=3$ when only proper linear subcodes are allowed.

\begin{remark}\label{remark:block}
    \reviii{For subcodes $\sC$ with rank deficiency $\Delta>1$,} %
    the quotient space $\mathcal{C}/\sC(n, k')$ contains $2^\Delta$ distinct cosets $\bm{0}+\sC$, $\bm{x}_{\mathsf{a}, 1}+\sC,\ldots,\bm{x}_{\mathsf{a}, 2^\Delta-1}+\sC$. %
    This is the minimum number of subcodes to achieve a covering based on $\sC$. 
\end{remark}

\begin{remark}
    \revi{Theorem \ref{prop:affine_cover}} can be extended to non-binary codes over $\mathbb{F}_q$ by using a $(k-1)$-dimensional linear subspace and $q-1$ affine subcodes, i.e., requiring $M=q$.
\end{remark}

We next show that \ac{aSCED} can match \ac{ML} decoding when employing sufficiently many paths.

\begin{proposition}\label{lemma:aSCED-ML}
    For a linear binary code $\mathcal{C}$, there exists an \ac{aSCED} with $M=2^k$ paths whose \ac{ML}-in-the-list decision coincides with \ac{ML} decoding.
\end{proposition}
\begin{IEEEproof}
\pbcol{Consider the singleton linear subcode ${\sC=\{\bm{0}\}}$.
According to Remark~\ref{remark:block}, it induces  ${M=2^k}$ singleton affine subcodes ${\mathcal{C}_\ell=\bm{x}_\ell+\sC}$, $\ell \in [M]$.
Consider an \ac{aSCED} ensemble with $M$ paths, where the $\ell$th path, $\ell\in[M]$, employs the decoder ${\mathrm{Dec}_\ell:\mathbb{R}^n\rightarrow\mathcal{C}_\ell}$.
Each path outputs its unique codeword, so the candidate list is $\mathcal{L}= \mathcal{C}$}.
\end{IEEEproof}
Proposition~\ref{lemma:aSCED-ML} establishes that an \ac{aSCED} approaches \ac{ML} performance as the number of paths increases; the construction above relies on ${M=2^k}$ paths, which is infeasible for practical codes. 
Although the resulting decoding complexity is prohibitively large, this theoretical result indicates that \ac{ML}-like behavior is, in principle, attainable.
Moreover, the results in Sec.~\ref{sec:structural_advantages_results} show that, when using subcode \acp{ssPCM}, a substantially smaller number of paths is sufficient to reach \ac{ML} performance in certain cases. 

To efficiently implement \ac{aSCED}, we require efficient decoders for both linear and affine subcodes, which are not guaranteed to exist.
In the following, we show that suitable decoders exist for \ac{BP} decoding.

\begin{remark}
\revi{In this work, we focus on \ac{BP}-based \ac{aSCED}. In \cite{sced_polar}, we considered \ac{aSCED} for polar codes, using the observation that suitable pre-transforms of the input to the polar graph can induce different subcodes \cite{trifonov_subcodes_2016}.}

\end{remark}

\subsection{\ac{BP}-based \ac{aSCED}}
We show how \ac{BP} decoding of binary codes can be adapted such that it can be used with affine subcodes within the context of \ac{aSCED}. 
 While \ac{BP} decoding of a linear subcode simply uses the \ac{PCM} of the subcodes without modifying the update rules \cite{mandelbaum_subcode_2025},
 \ac{BP} decoding of an affine subcode requires an adapted CN update rule to account for the fixed affine offset.
\subsubsection{\ac{BP} Decoding of Affine Subcodes}

The following lemma formalizes the modified \ac{BP} update equations when decoding over an affine subcode $\aC$.

\begin{lemma}\label{prop:affine_bp}
    Let $\aC(n,k')=\bm{x}_\mathsf{a}+\sC(n,k')$ be an affine subcode of a binary code $\mathcal{C}$, where $\sC(n,k')$ is a proper linear subcode with \ac{PCM} $\bm{H}_\mathsf{s}$.
    Define $\bm{s}_\mathsf{a}:=\bm{H}_\mathsf{s}\bm{x}_\mathsf{a}^\mathsf{T}$.
    Then, affine \ac{BP} decoding, denoted by $\mathrm{BP}_{\bm{s}_\mathsf{a}}(\bm{H}_\mathsf{s},\bm{\lambda})$, operates on the Tanner graph associated with $\bm{H}_\mathsf{s}$ using the adapted CN update 
    \begin{equation}\label{eq:affine_update_eq}
    (-1)^{(\bm{s}_\mathsf{a})_j}\cdot \lambda_{i\leftarrow j}
    .
    \end{equation}
    The initialization of VNs with the bit-wise LLRs and the VN update rules remain unchanged.
     Decoding terminates when $\bm{H}_\mathsf{s}\hat{\bm{x}}=\bm{s}_\mathsf{a}$ or when reaching the maximum number of iterations.
     When employed within an ensemble, decoding of each path still terminates when $\bm{H}\hat{\bm{x}}=\bm{0}$.
\end{lemma}
\begin{IEEEproof}
Recall that all codewords $\bm{x}\in\aC$ satisfy  
${\bm{H}_\mathsf{s}\bm{x}^\mathsf{T}=\bm{s}_\mathsf{a}}$. 
Rearranging gives ${\bm{H}_\mathsf{s}\bm{x}^\mathsf{T}+\bm{s}_\mathsf{a}=\bm{0}}$,
which corresponds to \ac{BP} decoding over the Tanner graph of $\bm{H}_\mathsf{s}$,
but with each CN $j$ incorporating a fixed binary offset $(\bm{s}_\mathsf{a})_j$.
This offset $(\bm{s}_\mathsf{a})_j$ is equivalent to injecting an \ac{LLR} of $(-1)^{(\bm{s}_\mathsf{a})_j}\cdot \infty$ at the CN.
Finally, since both standard CN update rules incorporate odd functions, we obtain the stated update rule.
\end{IEEEproof}
\begin{remark}
    The adapted CN update (\ref{eq:affine_update_eq}) can be used for all CNs since it reduces to the standard CN update for $\bm{s}_\mathsf{a}=\bm{0}$.
\end{remark}
\begin{remark}
    Affine \ac{BP} shares similarities with syndrome-based \ac{BP} decoding \cite[Ch. 47]{MacKay2003ITILA}. Unlike syndrome-based \ac{BP} decoding, the \emph{affine syndrome} $\bm{s}_\mathsf{a}$ is fixed and independent of the received sequence. 
\end{remark}
\begin{remark}
 The affine offsets $\bm{x}_\mathsf{a}$ defining the cosets can be derived from the non-zero columns of $\bm{H}_\mathsf{s}\bm{G}^\mathsf{T}$.
\end{remark}

\subsubsection{\ac{BP}-based \ac{aSCED}}
With affine \ac{BP} decoding established, \ac{BP}-based \ac{aSCED} can be realized by applying \ac{BP}, either in its standard or affine-modified form, to each linear or affine subcode in the ensemble.
Next, we propose a crucial design principle for constructing an efficient ensemble.
\begin{designprinciple}\label{dsgn:cover_with_cosets}
Let $\bm{H}_\mathsf{s}$ be the \ac{PCM} of a proper subcode $\sC(n,k')$ of $\mathcal{C}$, and let $\Delta:=k-k'$ the \emph{rank deficiency} of its \ac{PCM}. 
Following Remark~\ref{remark:block}, identify $2^\Delta-1$ distinct cosets $\bm{x}_{\mathsf{a}, i}+\sC=:\aC_{,i}$. 
Then, we propose to include the set
$\{
\sC,\aC_{,1},\ldots,\aC_{,2^\Delta-1}
\}$ satisfying (\ref{eq:general_cover}) as one batch, denoted as \ac{aSCED} batch, in the ensemble.
The whole \ac{aSCED} ensemble possibly consists of multiple \ac{aSCED} batches.
\end{designprinciple}
Given a proper subcode $\sC\subseteq_\mathsf{s} \mathcal{C}$ with \ac{PCM} $\bm{H}_\mathsf{s}$, the design principle is to construct an \ac{aSCED} batch, illustrated in Fig.~\ref{fig:asced_block} 
that ensures full coverage of $\mathcal{C}$ without degrading the performance of individual decoding paths.
\ac{aSCED} may comprise multiple \ac{aSCED} batches of different sizes. In this case, the total number of paths is $M=\sum_{\ell\in[L]}2^{\Delta_\ell}$, where $L$ denotes the number of \ac{aSCED} batches and $\Delta_\ell$ is the rank deficiency of the $\ell$th batch. If all batches have the same size, this simplifies to $M=2^\Delta \cdot L$.

When following Design Principle~\ref{dsgn:cover_with_cosets}, every codeword $\bm{x}\in \mathcal{C}$ is decodable by $L$ paths.
In contrast to \ac{SCED} \cite{mandelbaum_subcode_2025}, where ensuring equal protection for all codewords may degrade the performance of individual decoding paths, Theorem~\ref{th:error_prob_affine} establishes that \ac{aSCED} achieves this property without compromising the average error performance of any path. 
\begin{figure}[t]
    \centering
    \scalebox{0.8}{
\begin{tikzpicture}[
    level 1/.style={sibling distance=20mm},
    edge from parent/.style={->,draw},
    node distance=2mm and 4mm,
    >=latex
]

\node (input) at (-1.5,-0.5) {$\bm{\lambda}$};

\node (debox_1) [rectangle, draw, thin, minimum width=2.4cm, minimum height=0.6cm,font=\small,align=center] at (2,1)   {BP($\bm{H}_\mathsf{s},\cdot$)};
\node (debox_2) [rectangle, draw, thin, minimum width=2.4cm, minimum height=0.6cm,font=\small,align=center] at (2,0) {$\text{BP}_{\bm{s}_\mathsf{a,1}}(\bm{H}_\mathsf{s},\cdot)$};
\node (debox_3) [rectangle, draw, thin, minimum width=2.4cm, minimum height=0.6cm,font=\small,align=center] at (2,-2) {$\text{BP}_{\bm{s}_\mathsf{a,2^\Delta-1}}(\bm{H}_\mathsf{s},\cdot)$};

\filldraw[black] (-0.435,0) circle (1.2pt);
\filldraw[black] (-0.435,-0.5) circle (1.2pt);

\node (ml_1) [right=0.7cm of debox_1] {};
\node (ml_2) [right=0.7cm of debox_2] {};
\node (ml_K) [right=0.7cm of debox_3] {};

\node (ml_box) [rectangle, draw, thin, minimum width=0.6cm, minimum height=3.6cm] 
      at ($(ml_2)+(0.19cm,-0.5)$) 
      {\rotatebox{90}{ML-in-the-list}};

\node (output) [right=1cm of ml_box] {$\hat{\bm{x}}$};

\node (dots_dec_box) [below=0.15cm of debox_2] {$\vdots$};

\myline{input}{debox_1};
\myline{input}{debox_2};
\myline{input}{debox_3};

\draw [-latex,thick] (debox_1.east) -- (ml_1);
\draw [-latex,thick] (debox_2.east) -- (ml_2);
\draw [-latex,thick] (debox_3.east) -- (ml_K);

\draw [-latex,thick] (ml_box.east) -- (output);

\node[ draw=KITblue,fill=KITblue,fill opacity=0.1, dashed, rounded corners, fit=(debox_1)(debox_2)(debox_3)(dots_dec_box), 
      inner sep=8pt, label={[font=\small]above:\textcolor{KITblue}{aSCED batch}}] {};

\end{tikzpicture}
}
    \caption{Block diagram of an \ac{aSCED} batch according to Design Principle~\ref{dsgn:cover_with_cosets}.}
    \label{fig:asced_block}
\end{figure}
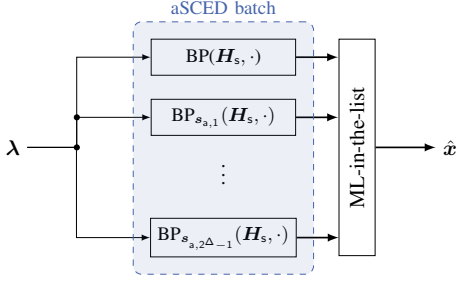

\begin{theorem}\label{th:error_prob_affine}
Let $\bm{H}_\mathsf{s}$ be the \ac{PCM} of a proper subcode $\sC(n,k')$ of $\mathcal{C}$ and let $\aC=\bm{x}_\mathsf{a}+\sC$ be an arbitrary coset of $\sC$. 
When transmitting via a BiMSC,
 the error probability of the decoder $\mathrm{BP}(\bm{H}_\mathsf{s})$,
 when averaged over all $\bm{x}\in \mathcal{C}$,
 is
identical to the average error probability of the affine decoder $\mathrm{BP}_{\bm{s}_\mathsf{a}}(\bm{H}_\mathsf{s})$.
\end{theorem}

\subsection{All-zero Assumption}

\reviv{
The all-zero codeword assumption is commonly used to accelerate Monte Carlo simulations by avoiding the need for encoding.
In the context of \ac{aSCED}, this assumption is particularly useful because, as shown below, it allows assessing the \ac{LEP} of an \ac{aSCED} batch by simulating only its linear subcode path, instead of all associated affine paths.
Besides reducing simulation complexity, the all-zero assumption can be helpful for rapid prototyping, training neural-based decoders \cite{NEURIPS2022_fcd3909d}, and for analytical tools \cite{MCT08}, thus, providing a basis for future analytical studies of \ac{aSCED}.}

\reviv{
The all-zero assumption does not by itself imply that the \ac{FEP} can always be obtained from such simulations. For \ac{aSCED}, the FE event consists of two parts: first, the transmitted codeword may be missing in the generated candidate list; and second, the transmitted codeword may be present in the list but not selected
by the final likelihood comparison. The latter event is an \ac{ML}-type selection error.  Hence, all-zero
simulations directly approximate the \ac{LEP}; they characterize the \ac{FEP} only when this \ac{ML}-type error term is negligible in the considered operating region. Importantly, this condition is not implied by the short block length of the code, but depends on the code, the channel quality, and the operating region under consideration.}

To formalize this, consider the transmission of the all-zero codeword $\bm{0}$, yielding the channel output $\bm{y}_0$. 
Define the events ${A := \{\bm{0} \in \mathcal{L}\}}$ (list success) and ${B := \{\arg\max_{\bm{x} \in \mathcal{L}}    (\bm{1}_{1\times n}-2\bm{x)}^\mathsf{T}\bm{\lambda}= \bm{0}\}}$ (frame success). \jmcol{}
\pbcol{Thus, we split the ensemble decoding process into the event $A$ that the transmitted codeword is contained in the list $\mathcal{L}$ and the \mbox{event $B$} that it is actually selected as the ensemble output, and define the \ac{LEP} and \ac{FEP}
\[
\mathrm{LEP}:=P(\bar{A}),\quad \mathrm{FEP}:=P(\bar{B}).\]}
Since ${P(B|\bar{A}) = 0}$ %
we 
have ${P(B)=P(A) (1 - P(\bar{B}|A))}$. 
\pbcol{For any outcome ${m \in A \cap \bar{B}}$, \reviv{the transmitted codeword is contained in
the list, but another candidate has a larger likelihood metric. Thus,} the \ac{ML} decoder also outputs an error, and we refer to this event as an \emph{\ac{ML} error}.}
 In the absence of \ac{ML} errors, ${P(\bar{B}|A) = 0\implies P(\bar{A}) = P(\bar{B})}$, i.e., $\mathrm{LEP}=\mathrm{FEP}$.

\begin{proposition}
    Consider an \ac{aSCED} batch associated with the linear subcode \ac{PCM} $\bm{H}_\mathsf{s}$. Then,
    \begin{align*}
    P(\mathrm{BP}(\bm{H}_\mathsf{s},\bm{y}_0)\neq \bm{0}  )
    &\geq  %
    P(\{\bm{0}\notin \mathcal{L}\}).
    \end{align*}
\end{proposition}
\begin{IEEEproof}
\pbcol{We neglect outcomes where at least one affine path estimates $\bm{0}$
 but $\mathrm{BP}(\bm{H}_\mathsf{s},\bm{y}_0)\neq \bm{0}$}.
Although theoretically possible due to the hard decision of \ac{BP}, this event is highly unlikely, which is confirmed by our simulations.    
\end{IEEEproof}

\begin{remark}\label{rem:az_sim_block}
The \ac{LEP} of an \ac{aSCED} batch can be (tightly) upper-bounded by $P(\mathrm{BP}(\bm{H}_\mathsf{s},\bm{y}_0) \neq \bm{0})$,
 which can be estimated as the \ac{LER} of Monte Carlo simulations under the all-zero assumption, simulating only the linear subcode path and ignoring the ${2^\Delta - 1}$ affine paths.
\end{remark}
\begin{remark}\label{rem:az_sim}
\pbcol{%
In the absence of \ac{ML} errors, the \ac{FER} of an \ac{aSCED} ensemble comprising $L$ \ac{aSCED} batches is identical to the \ac{LER} of Monte Carlo simulations using the all-zero codeword assumption and evaluating only the $L$ linear subcode paths, where the candidate list of the full ensemble is the union of the candidate lists generated by the individual \ac{aSCED} batches. }
\end{remark}

\section{
Structural Advantage of Subcodes for \ac{BP}
}\label{sec:structural_advantages}
\revi{
MacKay introduced the concept of a \ac{GPCM}, which extends conventional \acp{PCM} by allowing additional \emph{state variables} that are not transmitted \cite{mackay2000relationships}. These state variables behave similarly to punctured variables and provide flexibility in constructing Tanner graph representations of a code. Such representations can exhibit improved graph properties for iterative decoding, e.g., a reduced number of short cycles \cite{yedidia_4cyc_removal}. 
}

\subsection{Optimized \revi{Generalized} Parity-Check Matrices of Subcodes}
\pbcol{
In the very short block length regime, e.g., for $n \leq 128$, many codes with a large minimum Hamming distance are known, yet their dense PCMs often contain many short cycles of length $4$, which are particularly harmful for BP and are referred to as $4$-cycle.
Approaches to optimize PCMs for BP aim at improving their sparsity or eliminating $4$-cycles\cite{vasic_4cyc_orthogon,yifei_sspcm,yedidia_4cyc_removal,kumar_milen_4cyc_rem,yifei_univ_bp}. 
In such approaches, starting from the PCM $\bm{H}\in \mathbb{F}_2^{m\times n}$ of a code, one constructs an optimized \revi{\ac{GPCM}} $\bm{H}^\star\in \mathbb{F}_2^{m^\star\times n^\star}$, $m^\star\geq m$, $n^\star\geq n$, with improved properties for BP. 
If $m^\star> m$ or $n^\star> n$, we say that the approach introduces $m^\star- m$ auxiliary CNs (ACNs) or $n^\star- n$ auxiliary VNs (AVNs), respectively,
which ensure that ${\forall\bm{x}\in \mathcal{C}(\bm{H}^\star)\implies \bm{x}_{[n]}\in \mathcal{C}(\bm{H})}$. 
}

\pbcol{
In particular, we are interested in finding optimized \acp{GPCM} of subcodes, i.e., starting from $\bm{H}_\ell$ according to \eqref{eq:inducing_subcode}, identify $\bm{H}^\star_\ell$.
We revisit the construction of \acp{ssPCM} \cite{yifei_sspcm} and use \acp{ssPCM} as optimized \acp{GPCM} $\bm{H}^\star_\ell$ of the aSCED batches. }

\revi{
The AVNs correspond to the state variables in MacKay's \ac{GPCM} mentioned above. Nevertheless, we will use the term AVNs to emphasize their interaction with ACNs.
}

\subsection{Revisiting Structured Sparse Parity-Check Matrices}\label{sec:ssPCM}

\pbcol{
\revi{Building on MacKay's \ac{GPCM} representation \cite{mackay2000relationships},
Yedidia et al. proposed an algorithm for transforming a \ac{PCM} of a code into an equivalent \ac{GPCM} representation with a reduced number of $4$-cycles~\cite{yedidia_4cyc_removal}. Subsequently, \cite{kumar_milen_4cyc_rem} and \cite {vasic_4cyc_orthogon} independently proposed a more efficient $4$-cycle removal strategy, for which, with a single (ACN, AVN) pair, multiple $4$-cycles can be removed.}
\reviv{The $4$-cycle removal procedure proposed in \cite{kumar_milen_4cyc_rem} and \cite{vasic_4cyc_orthogon} can be formalized as follows.}
Consider a \ac{PCM} $\bm{H}$. 
If there exist index sets $\mathcal{R}\subset[m]$ and $\mathcal{T}\subset[n]$ with $|\mathcal{R}|,|\mathcal{T}|\geq2$ such that
${\bm{H}_{\mathcal{R},\mathcal{T}}=\bm{1}_{|\mathcal{R}|\times|\mathcal{T}|}}$, then the Tanner graph associated with $\bm{H}$ contains $4$-cycles.
The process of removing the associated $4$-cycles is detailed in Algorithm~\ref{alg:4cycle_elim}, which returns a (first) optimized \revi{\ac{GPCM}} $\bm{H}^\star\in \mathbb{F}_2^{(m+1)\times (n+1)}$ 
 with $\bm{H}^\star_\mathcal{R,T}= \bm{0}_{|\mathcal{R}|\times |\mathcal{T}|}$ and removes all $4$-cycles associated with the pair of index sets $\mathcal{R}$ and $\mathcal{T}$.
 The Tanner graph associated with $\bm{H}^\star$ may still have (other) $4$-cycles, which can be eliminated by iteratively repeating this process, introducing more auxiliary nodes\cite{vasic_4cyc_orthogon}.
 }

 \pbcol{If ${m^\star>m}$, \ac{BP} using $\bm{H}^\star$ operates on a \ac{GPCM} of increased size that contains AVNs.}
 \revi{As AVNs behave similarly to punctured bits, their channel \acp{LLR} are initialized to zero; consequently, they act as erasures that must be resolved during iterative decoding.}
We thus have a tradeoff in finding a suitable number of AVNs:
Fewer may leave residual cycles, while too many may degrade the \ac{BP} performance. 
Recent refinements further optimize \acp{GPCM} by avoiding low-degree CNs and maximizing the heuristic approximate cycle extrinsic message degree of 6-cycles \cite{wesel_ace}, which correlates with a larger minimum size of trapping sets\cite{yifei_univ_bp}.
In the follow-up work \cite{yifei_sspcm}, this strategy was further optimized, introducing \acp{ssPCM}.

\begin{algorithm}[bt!]
\caption{Elimination of $4$-cycles \cite{vasic_4cyc_orthogon,kumar_milen_4cyc_rem}}\label{alg:4cycle_elim}
\begin{algorithmic}[1]
\STATE \textbf{Input:} PCM $\bm{H}\in\mathbb{F}_2^{m\times n}$, index sets $\mathcal R\subseteq[m]$, $\mathcal T\subseteq[n]$
\STATE \textbf{Output:} Optimized GPCM $\bm{H}^\star\in\mathbb{F}_2^{(m+1)\times (n+1)}$
\STATE \textbf{Assume:} $\bm{H}_{\mathcal R,\mathcal T}= \bm{1}_{|\mathcal R|\times|\mathcal T|}$ with $|\mathcal R|,|\mathcal T|\geq 2$
\STATE $\bm{H}\gets\begin{pmatrix}     \bm{H} &\bm{0}_{m\times 1} \end{pmatrix}\in \mathbb{F}_2^{m\times (n+1)}$ (introduce AVN)
\STATE $\bm{H}^\star\gets[\bm{H} ; \bm{1}_{\mathcal{T}\cup \{n+1\}}]\in \mathbb{F}_2^{(m+1)\times (n+1)}$(introduce ACN)
\FOR{each $r\in\mathcal R$}
    \STATE $\bm{H}^\star_{r,:} \gets \bm{H}^\star_{r,:} + \bm{H}^\star_{m+1,:}$
\ENDFOR
\STATE \textbf{Return} $\bm{H}^\star$
\end{algorithmic}
\end{algorithm}

\begin{algorithm}[bt!]
\caption{Generation of search space matrix $\bm{S}$ \cite{yifei_sspcm}}\label{alg:search_space}
\begin{algorithmic}[1]
\STATE \textbf{Input:} Max. search space size $S_\mathrm{max}$, PCM $\bm{H}$ of $\mathcal{C}$
\STATE \textbf{Output:} Search space matrix $\bm{S}$
\STATE Initialize $d \gets d_\mathrm{min}^\bot$; $\bm{S}\gets \emptyset$
\WHILE{$\mathrm{number\_rows}(\bm{S})\!< \! S_\mathrm{max}$ and ${\mathrm{rank}(\bm{S}) < \!n-k}$}
\STATE  $\bm{s},d_\mathrm{new}\gets \text{find\_next\_min\_weight\_codeword}(\bm{H})$
\STATE $\bm{S} \gets [\bm{S}; \bm{s}]$
    \IF{$ d_\mathrm{new}>d$ \AND $\mathrm{rank}(\bm{S})=n-k$}
        \STATE \textbf{Return} $\bm{S}$
    \ENDIF
    \STATE $d\gets d_\mathrm{new}$
\ENDWHILE
\STATE \textbf{Return} $\bm{S}$
\end{algorithmic}
\end{algorithm}

\begin{algorithm}[bt!]
\caption{Generation of \ac{ssPCM} \cite{yifei_sspcm}}\label{alg:sspcm}
\begin{algorithmic}[1]
\STATE \textbf{Input:} Search space matrix $\bm{S}$, \revi{max. edge count} $w_\mathrm{max}$
\STATE \textbf{Output:} $\bm{H}_{\mathrm{ssPCM-I}}$ and $\bm{H}_{\mathrm{ssPCM-II}}$
\STATE Initialize $\bm{H}_\mathrm{ssPCM}\gets\emptyset,\,\, \reviii{\bm{P}^\star} \gets\emptyset$, check\_rank$\gets$\TRUE
\WHILE{$\mathrm{wt}(\bm{H}_\mathrm{ssPCM})<w_\mathrm{max}$ and $\bm{S}\neq\emptyset$}
\FOR{$1\leq i<j\leq \mathrm{number\_rows}(\bm{S})$}%
    \STATE $\bm{P}_\mathrm{cand.},\bm{P}^\star_\mathrm{cand.}\gets \mathrm{identify\_pcrb}(i,j,\bm{S})$
    \STATE $\bm{P}^\star\gets \mathrm{select}(\bm{P}^\star, \bm{P}^\star_\mathrm{cand.})$
\ENDFOR
\STATE $\bm{H}_\mathrm{ssPCM}\gets [ \bm{H}_\mathrm{ssPCM} ; \bm{P}^\star] $ 
\STATE $\bm{S}\gets \bm{S}.\mathrm{remove\_rows}(\reviii{\bm{P}_\mathrm{cand.}})$
\IF{check\_rank \AND $\mathrm{rank}(\bm{H}_\mathrm{ssPCM})=n-k$ }
\STATE $\bm{H}_{\mathrm{ssPCM-I}}\gets \bm{H}_\mathrm{ssPCM}$; 
\STATE check\_rank$\gets$\FALSE 
\ENDIF
\ENDWHILE
\STATE $\bm{H}_{\mathrm{ssPCM-II}}\gets \bm{H}_\mathrm{ssPCM}$
\STATE \textbf{Return} $\bm{H}_{\mathrm{ssPCM-I}}$ and $\bm{H}_{\mathrm{ssPCM-II}}$
\end{algorithmic}
\end{algorithm}

The construction of \acp{ssPCM} \reviv{proposed in \cite{yifei_sspcm}} starts by generating a search space of low-weight dual codewords, organized as a search matrix $\bm{S}$, as summarized in Algorithm~\ref{alg:search_space}.
Ideally, the matrix $\bm{S}$ should consist of $S_\mathrm{max}$ minimum-weight dual codewords.
If required to fulfill the rank condition ${\mathrm{rank}(\bm{S})=n-k}$, we also include higher-weight rows.
The function find\_next\_min\_weight\_codeword() returns the next dual codeword such that its Hamming weight is minimum among all dual codewords not considered so far, along with its weight. \pbcol{This function can be implemented using, e.g., \cite{leon_search}.}

\pbcol{\reviv{According to \cite{yifei_sspcm}}, the basic step of designing \acp{ssPCM} described in Algorithm~\ref{alg:sspcm} is to identify 
\emph{\acp{PCRB}} $\bm{P}$ in $\bm{S}$ as follows:
Consider a set of column indices ${\mathcal{T}\subset[n]}$ with ${t=|\mathcal{T}|}$.
A PCRB $\bm{P}$ consists of $s$ rows of $\bm{S}$ such that ${\bm{P}_{[s],\mathcal{T}}= \bm{1}_{s\times t}}$ and $\mathcal{T}$ is \emph{maximal} with this property, i.e., for any 
$\mathcal{T}'\subseteq[n]$ with $\mathcal{T}\subset \mathcal{T}'$  and ${t'=|\mathcal{T}'|}$:
  $\bm{P}_{[s],\mathcal{T}'}\not= \bm{1}_{s\times t'}$.
 Given such a PCRB $\bm{P}$, Algorithm~\ref{alg:4cycle_elim} with $(\mathcal{R}=[s],\mathcal{T})$ yields a 4-cycle free $\bm{P}^\star$\cite{yifei_sspcm}.
 }

\jmcol{An \ac{ssPCM} is generated from $\bm{S}$ by successively identifying such $\bm{P}$ and $\bm{P}^\star$, as detailed in Algorithm~\ref{alg:sspcm}: 
To identify candidate $\bm{P}$ of different size, i.e., with different $s, t$, we
iterate over all row pairs of $\bm{S}$ and apply $\mathrm{identify\_pcrb(\cdot)}$, which identifies all additional ${s-2}$ rows intersecting in the same $t$ positions and returns the associated candidate $\bm{P}$ (and $\bm{P}^\star$) denoted $\bm{P}_\text{cand.}$ (and $\bm{P}^\star_\text{cand.}$).
Among all candidates $\bm{P}^\star_\text{cand.}$, $\mathrm{select(\cdot)}$ selects the best PCRB according to the following criteria, in order of priority, such that criterion~$i$+1) only is used if several \acp{PCRB} are equivalent regarding criterion $i$):
}

\begin{enumerate}
    \item Maximize rank of $[\bm{H}_{\mathrm{ssPCM}}; \bm{P}^\star]$
    \item Maximize block size $s$;
    \item Minimize number of $4$-cycles in $[\bm{H}_{\mathrm{ssPCM}}; \bm{P}^\star]$
    \item Minimize variance of information column weights
\end{enumerate}
Note that, in contrast to \cite{yifei_sspcm}, we reordered rules 1) and 2), i.e., we prioritize maximizing the rank over maximizing the block size $s$, which resulted in a more stable behavior of the implemented algorithm.

After iterating over all row pairs in $\bm{S}$, the selected $\bm{P}^\star$ is appended to $\bm{H}_{\mathrm{ssPCM}}$, and each row in the search space is extended by one $0$ to preserve the structure when appending future $\bm{P}^\star$.
The dual codewords forming the selected $\bm{P}$ are then removed from $\bm{S}$, and the process repeats by iterating over all row pairs until an \ac{ssPCM} $\bm{H}^\star\in\mathbb{F}_2^{m'\times n'}$ satisfying ${\mathrm{rank}(\bm{H})=n-k}$ is obtained, which is denoted $\bm{H}_{\mathrm{ssPCM-I}}$.
Algorithm~\ref{alg:sspcm} continues beyond $\bm{H}_{\mathrm{ssPCM-I}}$ to construct larger \acp{ssPCM} with enhanced cycle connectivity, stopping once a predefined criterion is met, e.g., a maximum \revi{edge count} $w_\mathrm{max}$.

Our simulations show that \ac{NMSA} combined with \acp{ssPCM} provides strong error-correction performance.
In contrast, \ac{SPA} requires an additional normalization constant; therefore, when using \acp{ssPCM}, we employ \ac{NSPA} with normalization factor $\alpha$.%

\subsection{Structural Advantage of Subcodes}\label{subsec:structual_adv}

\pbcol{We propose employing \acp{ssPCM} as optimized \ac{GPCM} $\bm{H}_\ell^\star$ for an \ac{aSCED} batch.
Starting from the \ac{PCM} $\bm{H}_\ell$ of a linear subcode, Algorithm~\ref{alg:search_space} yields a search space matrix $\bm{S}_\ell$ consisting of dual codewords of the subcode.
Then, we construct \acp{ssPCM} tailored to the subcode $\mathcal{C}_\ell$ using Algorithm~\ref{alg:sspcm}, i.e., yielding $\bm{H}_\ell^\star$.}

\pbcol{
Using \acp{ssPCM} tailored to subcodes provides a structural advantage for \ac{BP}. Since ${\sC\subset_\mathrm{s}\mathcal{C}\implies \sC^\bot\supset_\mathrm{s}\mathcal{C}^\bot }$,
the minimum Hamming weight $\mathrm{d}_\text{min}^\bot$ of $\mathcal{C}^\bot$ bounds the minimum Hamming weight $\mathrm{d}_{\ell,\text{min}}^\bot$ of $\sC^\bot$ from above.
Hence, via linear combinations, the linearly appended row in \eqref{eq:inducing_subcode} can introduce additional dual codewords of the subcode $\bm{x}\notin \mathcal{C}^\bot$ with $\mathrm{wt}(\bm{x})<\mathrm{d}_\text{min}^\bot$.
By appending a linearly independent row with weight ${d_\mathsf{c}<\mathrm{d}_\text{min}^\bot}$ to $\bm{H}$, we can ensure the existence of such vectors and $\bm{S}_\ell$ contains rows with weight lower than all rows in $\bm{S}$. 
Indeed, we observe the following:
The minimum weight $\mathrm{d}_\text{min}^\bot$ of the dual code of the \ac{BCH} code $\mathcal{C}(63,30)$ is $12$. 
After appending to $\bm{H}$ a randomly selected linearly independent row of weight $d_\mathsf{c}=6$ and running Algorithm~\ref{alg:search_space} with $S_\mathrm{max}=\num{1000}$,
 the resulting search space matrix $\bm{S}_\ell$ contained one row of weight $6$, and $10$ and $220$ rows of weight $8$ and $10$, respectively. 
 Consequently, Algorithm~\ref{alg:sspcm} using $\bm{S}_\ell$  rather than $\bm{S}$ typically results in $\bm{H}_\ell^\star$ being sparser than $\bm{H}_\ell$, and thus, providing a structural advantage for BP.
}

\section{Numerical Results}\label{sec:res}
To evaluate the performance of different decoding schemes, we conduct Monte Carlo simulations using a \ac{BI-AWGN} channel. 
We use the notation $\{\text{aSCED},\text{SCED},\text{MBBP},\text{AED}\}$-$M$ to denote the respective ensemble decoding technique employing $M$ parallel paths.
Unless stated otherwise, each stand-alone decoder and each path employs \ac{BP} with a flooding schedule and the early stopping criterion $\bm{H}\hat{\bm{x}}^\mathsf{T}=\bm{0}$. For every simulated SNR point, we collect at least $200$ FEs.
\revii{If available, we include the simulated ML performance reported in~\cite{channelcodes}, obtained using the linear-programming-based methodology described in~\cite{8464884}. 
These results are Monte Carlo estimates of the true ML performance. Otherwise, if available,  we include the performance of \ac{OSD} with order~$t$, denoted as \ac{OSD}-$t,$ as an estimate of the \ac{ML} performance. If the selected order $t$ is sufficiently large such that OSD-$t$ is known to achieve near-ML performance, we indicate this by appending ``(ML)'' to the respective OSD label.}

We first evaluate the performance of \ac{aSCED} for LDPC codes. Subsequently, we consider \ac{BCH} codes, highlighting \pbcol{the structural advantage for BP} arising from the use of subcodes.

\begin{remark}\label{rem:github}
\revii{All used PCMs and the} C++ implementation are available at  \href{https://github.com/JonathanMandelbaum/aSCED.git}{https://github.com/kit-cel/aSCED.git} \cite{aSCED_sim_2026}.
\end{remark}

\subsection{Performance \ac{aSCED}-\ac{BP} for LDPC Codes}

We consider two short-length LDPC codes, namely the 5G LDPC code $\mathcal{C}_{\mathrm{5G}}(132,66)$ and the CCSDS LDPC code $\mathcal{C}_\mathrm{CCSDS}(256,128)$.
Both codes are quasi-cyclic, which, \pbcol{using techniques from \cite{stuttgart_ldpc_aed}}, allows us to include the performance of \ac{AED} as a reference.
\pbcol{
To show that aSCED can improve different BP variants}, we employ the \ac{NMSA} with normalization constant ${\alpha=\frac{3}{4}}$ for $\mathcal{C}_{\mathrm{5G}}(132,66)$,
whereas for $\mathcal{C}_\mathrm{CCSDS}(256,128)$, we use the \ac{SPA}.
\revii{In the following, we consider two design processes for selecting linearly independent rows, one for general LDPC codes and one for \ac{RL} LDPC codes, and then report the performance of ensembles constructed accordingly.}

\subsubsection{\revii{Selection for General LDPC Codes}}\label{sec:maxcov_selection}
\revi{In general, for an LDPC code with \ac{PCM} $\bm{H}$, one may use \cite[Algorithm~1]{mandelbaum_subcode_2025} for finding new linearly independent rows that do not contribute additional $4$-cycles beyond the ones already existing in $\bm{H}$.} \revii{While such graph-preserving constructions are appealing from a structural perspective, our simulations indicate that the ensemble performance also strongly depends on the complementarity of the different decoding paths. Therefore, we use a maximum-coverage selection procedure that balances the raw performance of each path against how well the paths perform jointly as an ensemble.}

 \revii{For the maximum coverage procedure,} in line with \cite{mandelbaum_subcode_2025}, since well-performing PCMs $\bm{H}$ are known for LDPC codes by design, we fix ${\mathrm{Dec}_1=\mathrm{BP}(\bm {H},\cdot)}$ and consider ensembles of size $M=\tilde{M}+1$, where $\tilde{M}$ is the number of additional auxiliary paths.
\pbcol{We use $\Delta=1$}, i.e., we append a single linearly independent row to $\bm{H}$ and each \ac{aSCED} batch contributes two decoders.
Thus, the ensemble size is ${M=2\tilde{L}_\mathsf{s}+1}$, where $\tilde{L}_\mathsf{s}$ denotes the number of constructed PCMs $\bm{H}_\ell$ of the linear subcodes constituting an \ac{aSCED} batch.
To select well-performing ensembles for a general LDPC code, we follow the \revii{maximum coverage} procedure described in \cite[Sec.~VI-A]{mandelbaum_subcode_2025}. \revii{Briefly stated, this procedure constructs a pool of candidate rows, which are linearly independent,  and selects those rows whose corresponding auxiliary paths correct the largest number of previously uncovered FEs from an offline design set of stand-alone BP decoding failures. Unless the no-new-$4$-cycle constraint above is explicitly imposed, this coverage-based selection is not guaranteed to preserve $4$-cycle-freeness; instead, it is a reproducible performance-driven criterion for promoting diversity among the ensemble paths.}

\subsubsection{\revii{Selection for Raptor-Like LDPC Codes}}\label{sec:rl_selection}
\revii{
In the following, we propose a more \emph{principled} selection approach for
\ac{RL} LDPC codes \cite{protograph_raptor_like_ldpc}. The \ac{PCM}
of an \ac{RL} LDPC code is composed of a highest-rate code $\bm{H}_\mathrm{HR}$
and an incremental-redundancy (IR) code $\bm{H}_\mathrm{IR}$ and is of the form
\begin{equation*}
\bm{H}_\mathrm{RL} =
\begin{bmatrix}
     \bm{H}_\mathrm{HR}&\bm{0}\\
     \bm{H}_\mathrm{IR}&\bm{I}
\end{bmatrix}.
\end{equation*}
Starting from the highest-rate code, \ac{RL} LDPC codes enable rate adaptivity by
incrementally extending $\begin{bmatrix}
    \bm{H}_\mathrm{HR}&\bm{0}
\end{bmatrix}$ by adding rows from the matrix
$\bm{H}_\mathrm{L}:=
\begin{bmatrix}
     \bm{H}_\mathrm{IR}&\bm{I}
\end{bmatrix}$.
Due to the identity matrix, all newly added rows are linearly independent of
the previous rows. Thus, the number of information bits remains fixed, while
the block length and number of check nodes are increased, resulting in a lower-rate code.
In the special case of protograph-based \ac{RL} LDPC codes, such as the $5$G LDPC
code, and when using a quasi-cyclic lifting process with lifting factor
$Z_ \mathrm{c}$, the matrix $\bm{H}_\mathrm{L}$ consists of row blocks of size $Z_ \mathrm{c}$,
where the rows within one block are obtained by quasi-cyclic shifts of the
corresponding protograph row.
}
\revii{
The \ac{RL} structure of such codes therefore allows for a simple, principled, and
reproducible selection of linearly independent rows: we use subsequent rows
from the IR part that are not used for transmission at
the considered operating rate. In contrast to the maximum-coverage procedure
used for general LDPC codes, this selection does not require an offline
design set. Instead, the additional rows are determined directly by the
nested \ac{RL} code structure. As for general LDPC codes, this structural choice
does not guarantee that no additional short cycles are
introduced; however, the rows remain sparse and inherit the protograph/QC
structure of the mother code.}

\revii{
Note that the selected IR rows stem from $\bm{H}_\mathrm{L}$. However, \ac{aSCED} is applied to the fixed operating-rate code, thereby keeping the number of variable nodes fixed, i.e., 
discarding any additional identity columns.
These discarded columns correspond to IR symbols that are not transmitted at the considered operating rate. Thus, linear independence is no longer guaranteed by the identity part. Hence, for every selected batch, we explicitly check that they are linearly independent.
For all codes considered in this work, the appended rows were found to be linearly independent.
}

\subsubsection{Numerical Results}

Fig.~\ref{fig:5G_LDPC_132_66} and Fig.~\ref{fig:ccsds256_128} show the \ac{FER} performance of stand-alone \ac{BP}, \ac{AED}-$M$, and \revi{different configurations of} the novel \ac{aSCED}-$M$
for the code $\mathcal{C}_{\mathrm{5G}}(132,66)$ and code $\mathcal{C}_\mathrm{CCSDS}(256,128)$, respectively.
 All \ac{BP} decodings employ a maximum of ${I_\mathrm{max}=32}$ iterations. 
Additionally, in Fig.~\ref{fig:5G_LDPC_132_66}, we include the performance of SCED-$11$, \reviv{SCED-$43$}, and stand-alone \ac{BP} with ${I_\mathrm{max}=352}$ iterations, as reported in~\cite{mandelbaum_subcode_2025}. \reviv{Note that, as shown in \cite{mandelbaum_subcode_2025}, the performance of SCED with $\Delta=1$ saturates at around $43$ paths, i.e., no further gains are obtained by further increasing the ensemble size.}

Under the assumption that all ensemble paths are executed in parallel, BP using ${I_\mathrm{max}=32}$ iterations serves as a reference for worst-case latency $\max_{i\in[M]} I_{\mathrm{max},i}$,
whereas NMSA-$352$ represents a worst-case computational complexity scenario.
Since all PCMs have a comparable number of edges, the worst-case decoding complexity is governed by the maximum total number of iterations $\sum_{i\in[M]} I_{\mathrm{max},i}$.

\reviv{
For the $5$G LDPC code $\mathcal{C}_{\mathrm{5G}}(132,66)$, the ensemble of size $M=11$ is designed according to Sec.~\ref{sec:maxcov_selection} and includes the original \ac{PCM} path. In contrast, the larger \ac{aSCED}-$M$ ensembles with $M\in \{48,384\}$ leverage the \ac{RL} structure and are constructed according to Sec.~\ref{sec:rl_selection}; these ensembles do not include the original \ac{PCM} path. The $5$G code is QC with lifting factor $Z_ \mathrm{c}=11$, i.e., one QC block consists of $11$ rows. 
Due to the RL structure specified in the 5G standard, the QC blocks have a natural order \cite{5gstandard_2025_01}. Starting from the PCM corresponding to the considered operating rate, subsequent QC blocks are those that are appended according to this order.
We use the following design procedure:
\begin{itemize}
    \item \ac{aSCED}-$48$: we take the two subsequent QC blocks, i.e., a total of $22$ rows. Each of the two QC blocks consisting of $11$ rows is subdivided into four sub-blocks of two rows and one sub-block of three rows, resulting in a total of eight batches with $\Delta=2$ and two batches with $\Delta=3$. Hence, the total number of paths is $8\cdot 2^2+2\cdot 2^3=48$.
    \item \ac{aSCED}-$384$: we take the four subsequent QC blocks, i.e., a total of $44$ rows. Each of the four QC blocks consisting of $11$ rows is subdivided into one sub-block of five rows and one sub-block of six rows, resulting in a total of four batches with $\Delta=5$ and four batches with $\Delta=6$. Hence, the total number of paths is $4\cdot 2^5+4\cdot 2^6=384$.
\end{itemize}
Both configurations (${M\in\{17,31\}}$) for the code $\mathcal{C}_\mathrm{CCSDS}(256,128)$ are constructed according to Sec.~\ref{sec:maxcov_selection}.}

In Fig.~\ref{fig:5G_LDPC_132_66}, the performance of \ac{aSCED}-$11$ matches or surpasses that of \ac{SCED}-11.
Importantly, this gain is accompanied by advantages from an implementation perspective: 
all decoders within one \ac{aSCED} batch share the same graph with slightly different CN updates, enabling potential hardware reuse and increased architectural flexibility. 
Furthermore, the construction of an \ac{aSCED} ensemble only requires half the number of suitable PCMs compared to \ac{SCED}, thereby simplifying the design process. 
Finally, we verify coverage properties by evaluating the fraction of codewords that are not decoded by any auxiliary path of the ensemble of \ac{SCED}-$11$, \pbcol{which consists of decoding on $\bm{H}$ and on $10$ auxiliary paths using proper subcodes~\cite{mandelbaum_subcode_2025}.
Out of $10{,}000$ randomly sampled codewords from code $\mathcal{C}_{\mathrm{5G}}$, a total of $11$ codewords were not covered by any of the $10$ proper subcodes. 
Thus, there exist codewords that are \emph{less protected} since they can only be recovered by the path using $\bm{H}$.
}
 However, by Design~Principle~\ref{dsgn:cover_with_cosets}, \ac{aSCED} guarantees \emph{uniform protection} across all codewords.

\pbcol{In Fig.~\ref{fig:5G_LDPC_132_66} and Fig.~\ref{fig:ccsds256_128}, \ac{aSCED}-$M$ consistently outperforms stand-alone \ac{BP} decoding and \ac{AED}-$M$ under equal worst-case latency and equal worst-case complexity. Specifically, at an FER $10^{-3}$, in Fig.~\ref{fig:5G_LDPC_132_66}, \ac{aSCED}-11 achieves gains of $0.4$\,dB and $0.2$\,dB over NMSA and \ac{AED}-11, respectively. In Fig.~\ref{fig:ccsds256_128}, the lower-complexity aSCED-$17$ matches AED-$32$, while at equal complexity, aSCED-$32$ outperforms AED.}

\revii{Finally, Fig.~\ref{fig:5G_LDPC_132_66} shows that \ac{aSCED}-$48$ and \ac{aSCED}-$384$, which are constructed solely from \ac{RL}-based batches, further improve the decoding performance when using $\Delta>1$. In particular, \ac{aSCED}-$384$ achieves a gain of $0.9$\,dB compared to stand-alone decoding at an \ac{FER} of $10^{-3}$. In the next section, we further elaborate on the performance--complexity tradeoff and compare different $\Delta$-configurations.}

\begin{figure}[t]
    \centering
    \begin{tikzpicture}[scale=0.92,spy using outlines={rectangle, magnification=2}]

\begin{axis}[%
width=.9\columnwidth,
height=1.2\ferfigheight,
at={(0.758in,0.645in)},
scale only axis,
xmin=1,
xmax=5,
xlabel style={font=\color{white!15!black}},
xlabel={$E_{\mathrm{b}}/N_0$ ($\si{dB}$)},
ymode=log,
ymin=1e-06,
ymax=1,
yminorticks=true,
ylabel style={font=\color{white!15!black}},
ylabel={FER},
axis background/.style={fill=white},
xmajorgrids,
ymajorgrids,
legend style={at={(0.004,0.005)}, anchor=south west, legend cell align=left, align=left, draw=white!15!black,font=\scriptsize}
]

\addplot[color=kit-royalblue80,solid,line width = 1pt,mark=x, mark options={solid}]
table[row sep=crcr]{
 1.0  5.467980e-01\\
 1.5  3.553223e-01\\
 2.0  1.798365e-01\\
 2.5  7.278189e-02\\
 3.0  2.111038e-02\\
 3.5  4.332920e-03\\
 4.0  7.442976e-04\\
 4.5  1.148381e-04\\
 5.0  1.290430e-05\\
};
\addlegendentry{NMSA \cite{mandelbaum_subcode_2025}};

\addplot[color=kit-royalblue80,dotted,line width = 1pt,mark=x, mark options={solid}]
table[row sep=crcr]{
 1.00  5.141388e-01\\
 1.50  2.551020e-01\\
 2.00  1.264223e-01\\
 2.50  4.380201e-02\\
3.00  1.090275e-02\\
 3.50  2.144956e-03\\
 4.00  3.055329e-04\\
 4.50  4.053300e-05\\
 5.00  4.830237e-06\\
};
\addlegendentry{NMSA \cite{mandelbaum_subcode_2025}\\${I_\mathrm{max}=352}$ };

\addplot[color=KITgreen,dashed,line width = 1pt,mark=triangle, mark options={solid}]
table[row sep=crcr]{
 1.00  4.926108e-01\\
 1.50  2.998501e-01\\
 2.00  1.362398e-01\\
 2.50  4.621072e-02\\
 3.00  1.192677e-02\\
 3.50  2.103359e-03\\
 4.00  2.896100e-04\\
4.50  3.245119e-05\\
 5.00  3.464770e-06\\
};
\addlegendentry{AED-$11$ \cite{mandelbaum_subcode_2025}};

\addplot[color=KITorange,line width = 1pt,mark=diamond*, mark options={solid}, dashed]
table[row sep=crcr]{
 1.00  4.464286e-01\\
 1.50  2.743484e-01\\
 2.00  1.147447e-01\\
 2.50  4.145078e-02\\
 3.00  8.008329e-03\\
 3.50  1.445713e-03\\
 4.00  1.989583e-04\\
4.50  1.514168e-05\\
 5.00  1.494075e-06\\
};
\addlegendentry{SCED-$11$ \cite{mandelbaum_subcode_2025}};

\addplot[color=KITpurple,dashed,line width = 1pt,mark=square, mark options={solid}]
table[row sep=crcr]{
 1.00  3.960396e-01\\
 1.50  2.242152e-01\\
 2.00  7.855460e-02\\
 2.50  2.734108e-02\\
 3.00  5.499945e-03\\
 3.50  6.574903e-04\\
 4.00  6.791694e-05\\
4.50  5.026574e-06\\
};
\addlegendentry{SCED-43 \cite{mandelbaum_subcode_2025}};

\addplot[color=KITred,line width = 1pt,mark=diamond*, mark options={solid}]
table[row sep=crcr]{
 1.00  5.122951e-01\\
 1.50  2.467917e-01\\
 2.00  1.062473e-01\\
 2.50  3.351206e-02\\
 3.00  7.658470e-03\\ %
3.50  1.160389e-03\\
 4.00  1.392763e-04\\
 4.50  1.281535e-05\\
 5.00  1.007620e-06\\
};
\addlegendentry{aSCED-11};

\addplot[color=KITblue!60,line width = 1pt,mark=triangle, mark options={solid}]
table[row sep=crcr]{
1 0.3571035747021082\\
1.5 0.18488151184881513\\
2 0.06884949348769899\\
2.5 0.019854137385318563\\
3 0.0036196964605847448\\
3.5 0.000491052420269672\\
4 4.4328988197752285e-05\\
4.5  2.965920104771095e-06\\
5 1.3379296624288492e-07\\
};
\addlegendentry{aSCED-48}

\addplot[color=KITpurple!30,line width = 1pt,mark=pentagon, mark options={solid}]
table[row sep=crcr]{
1 0.23617896631524812\\
1.5 0.10047980361526444\\
2 0.030786870657493925\\
2.5 0.006419873173834153\\
3 0.0009913115828252789\\
3.5 9.581785784534925e-05\\
4 7.385968565865376e-06\\
};
\addlegendentry{aSCED-384}

\addplot[color=black,line width = 1pt, solid]
table[col sep=comma]{
1.00, 1.120e-01
1.50, 3.609e-02
2.00, 9.891e-03
2.50, 1.623e-03
3.00, 2.514e-04
};
\label{plot:5g132_osd}
\addlegendentry{OSD-4 \cite{stuttgart_ldpc_aed}};

\end{axis}

\end{tikzpicture}%
    \caption{Decoder performances for the 5G LDPC code $\mathcal{C}_\mathrm{5G}(132,66)$. All \ac{BP} decodings employ the \ac{NMSA} using a normalization constant of $\alpha=\frac{3}{4}$. Unless specified otherwise, all other \ac{BP} decoders use $I_\mathrm{max}=32$ iterations.} %
    \label{fig:5G_LDPC_132_66}
\end{figure}
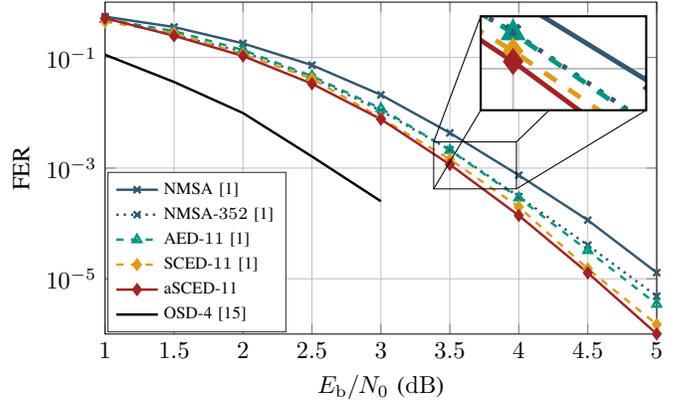

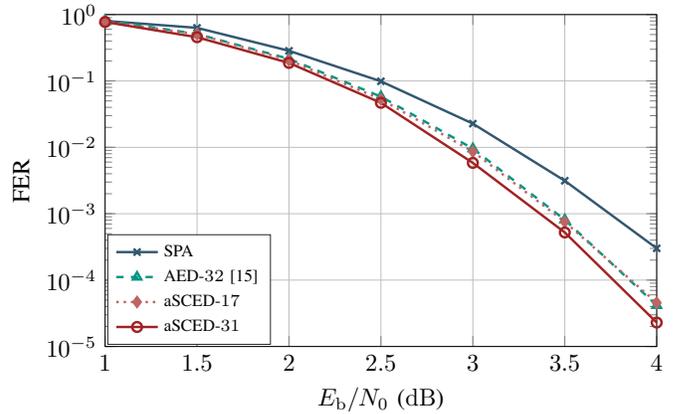
\begin{figure}[t]
    \centering
    \begin{tikzpicture}[scale=0.92,spy using outlines={rectangle, magnification=2}]

\begin{axis}[%
width=.9\columnwidth,
height=\ferfigheight,
at={(0.758in,0.645in)},
scale only axis,
xmin=1,
xmax=4,
xlabel style={font=\color{white!15!black}},
xlabel={$E_{\mathrm{b}}/N_0$ ($\si{dB}$)},
ymode=log,
ymin=1e-05,
ymax=1,
yminorticks=true,
ylabel style={font=\color{white!15!black},at={(-0.115,0.4)}},
ylabel={FER},
axis background/.style={fill=white},
xmajorgrids,
ymajorgrids,
legend style={at={(0.004,0.005)}, anchor=south west, legend cell align=left, align=left, draw=white!15!black,font=\scriptsize}
]

\addplot[color=kit-royalblue80,line width = 1pt,mark=x, mark options={solid}]
table[row sep=crcr]{
 1.00  8.032129e-01\\
 1.50  6.309148e-01\\
 2.00  2.844950e-01\\
 2.50  9.881423e-02\\
 3.00  2.275313e-02\\
 3.50  3.136320e-03\\
 4.00  3.009492e-04\\
 4.50  2.150000e-05\\
};
\addlegendentry{SPA};

\addplot[color=KITgreen,dashed,line width = 1pt,mark=triangle, mark options={solid}]
table[row sep=crcr]{
1.00 8.203e-01\\
1.50 5.111e-01\\
2.00 2.153e-01\\
2.50 5.811e-02\\
3.00 9.559e-03\\
3.50 8.103e-04\\
4.00 4.165e-05\\
};
\addlegendentry{AED-$32$ \cite{stuttgart_ldpc_aed}};

\addplot[color=KITred!70,line width = 1pt,dotted ,mark=diamond*, mark options={solid}]
table[row sep=crcr]
{
 1.00  7.886435e-01\\
 1.50  5.070994e-01\\
 2.00  2.055921e-01\\
 2.50  5.298855e-02\\
 3.00  8.643042e-03\\
 3.50  7.548993e-04\\
 4.00  4.552593e-05\\
};
\addlegendentry{aSCED-$17$};

\addplot[color=KITred,line width = 1pt,mark=o, mark options={solid}]
table[row sep=crcr]{
 1.00  7.731959e-01\\
 1.50  4.559271e-01\\
 2.00  1.869159e-01\\
 2.50  4.682379e-02\\
 3.00  5.830224e-03\\
 3.50  5.220742e-04\\
 4.00  2.294130e-05\\
};
\addlegendentry{aSCED-$31$};

\end{axis}

\end{tikzpicture}%
    \caption{Performance of \ac{BP}-based decoders for the CCSDS code $\mathcal{C}_\mathrm{CCSDS}(256,128)$.
    All \ac{BP} decoders employ the \ac{SPA} with a maximum of $I_\mathrm{max}=32$ iterations. Both \ac{aSCED} ensembles use $\Delta=1$.}
    \label{fig:ccsds256_128}
\end{figure}

\subsubsection{Performance vs complexity}

\begin{figure*}
    \centering
    \pgfplotstableread[
    col sep=space,
]{figures/5G_LDPC/scatter_data/n78_paths_vs_required_snr.dat}\datatableA

\pgfplotstableread[
    col sep=space,
]{figures/5G_LDPC/scatter_data/n180_paths_vs_required_snr.dat}\datatableB

\begin{tikzpicture}
\begin{axis}[
    xlabel={Required $E_{\mathrm{b}}/N_0$ ($\si{dB}$) for a target FER $10^{-3}$},
    ylabel={Number of Paths},
    minor y tick num=4, %
    minor x tick num=1,
    ymin=-2,
    xmin=2.7,
    xmax=5.8,
    yminorticks=true,
    xminorticks=true,
    xmajorgrids,
    ymajorgrids,
    legend pos=outer north east,
    legend style={
        font=\scriptsize,
        legend cell align=left,
        cells={align=left},
        at={(rel axis cs: 0.4, 0.85)}, %
        anchor=north, %
        yshift=-1em, %
    },
    width=0.9\textwidth,
    height=6cm,
]

\draw[black, thick] (5.01, \pgfkeysvalueof{/pgfplots/ymin}) -- (5.01, \pgfkeysvalueof{/pgfplots/ymax});
\node[
    anchor=east,
    font=\footnotesize\bfseries,
] at (5.01, 170) {OSD-2 (ML) n=78};

\draw[black, thick] (2.735, \pgfkeysvalueof{/pgfplots/ymin}) -- (2.735, \pgfkeysvalueof{/pgfplots/ymax});
\node[
    anchor=west,
    font=\footnotesize\bfseries,
] at (2.72, 170) {OSD-4 n=180};
\addplot[
    only marks,
    mark=x,
    mark size=3pt,
    opacity=0.3,
    draw opacity=1,
    color=kit-royalblue80,
] table[
    col sep=space,
    x=requiredsnr,
    y=numpaths,
    restrict expr to domain={\thisrow{split}}{0:0}
] {\datatableA};
\addlegendentry{NMSA}

\addplot[
    only marks,
    mark=*,
    mark size=2.5pt,
    opacity=0.3,
    color=KITorange,
] table[
    col sep=space,
    x=requiredsnr,
    y=numpaths,
    restrict expr to domain={\thisrow{split}}{5:5}
] {\datatableA};
\addlegendentry{AED}

\addplot[
    only marks,
    mark=diamond*,
    mark size=2.5pt,
    opacity=0.3,
    draw opacity=1,
    color=KITred,    
] table[
    col sep=space,
    x=requiredsnr,
    y=numpaths,
    restrict expr to domain={\thisrow{split}}{1:1}
] {\datatableA};
\addlegendentry{$\Delta=1$}%

\addplot[
    only marks,
    mark=triangle*,
    mark size=2.5pt,
    opacity=0.3,
    draw opacity=1,
    color=KITpalegreen
] table[
    col sep=space,
    x=requiredsnr,
    y=numpaths,
    restrict expr to domain={\thisrow{split}}{2:2}
] {\datatableA};
\addlegendentry{$\Delta=2$}%

\addplot[
    only marks,
    mark=square*,
    mark size=2pt,
    opacity=0.3,
    draw opacity=1,
    color=KITblue
] table[
    col sep=space,
    x=requiredsnr,
    y=numpaths,
    restrict expr to domain={\thisrow{split}}{3:3}
] {\datatableA};
\addlegendentry{$\Delta=3$}

\addplot[
    only marks,
    mark=pentagon*,
    mark size=2.5pt,
    opacity=0.3,
    draw opacity=1,
    color=KITgreen
] table[
    col sep=space,
    x=requiredsnr,
    y=numpaths,
    restrict expr to domain={\thisrow{split}}{4:4}
] {\datatableA};
\addlegendentry{$\Delta=6$}

\addplot[
    only marks,
    mark=x,
    mark size=3pt,
    opacity=0.3,
    draw opacity=1,
    color=kit-royalblue80,
] table[
    col sep=space,
    x=requiredsnr,
    y=numpaths,
    restrict expr to domain={\thisrow{split}}{0:0}
] {\datatableB};

\addplot[
    only marks,
    mark=*,
    mark size=2.5pt,
    opacity=0.3,
    draw opacity=1,
    color=KITorange,
] table[
    col sep=space,
    x=requiredsnr,
    y=numpaths,
    restrict expr to domain={\thisrow{split}}{5:5}
] {\datatableB};

\addplot[
    only marks,
    mark=diamond*,
    mark size=2.5pt,
    opacity=0.3,
    draw opacity=1,
    color=KITred
] table[
    col sep=space,
    x=requiredsnr,
    y=numpaths,
    restrict expr to domain={\thisrow{split}}{1:1}
] {\datatableB};

\addplot[
    only marks,
    mark=triangle*,
    mark size=2.5pt,
    opacity=0.3,
    draw opacity=1,
    color=KITpalegreen
] table[
    col sep=space,
    x=requiredsnr,
    y=numpaths,
    restrict expr to domain={\thisrow{split}}{2:2}
] {\datatableB};

\addplot[
    only marks,
    mark=square*,
    mark size=2pt,
    opacity=0.3,
    draw opacity=1,
    color=KITblue
] table[
    col sep=space,
    x=requiredsnr,
    y=numpaths,
    restrict expr to domain={\thisrow{split}}{3:3}
] {\datatableB};

\addplot[
    only marks,
    mark=pentagon*,
    mark size=2.5pt,
    opacity=0.3,
    draw opacity=1,
    color=KITgreen
] table[
    col sep=space,
    x=requiredsnr,
    y=numpaths,
    restrict expr to domain={\thisrow{split}}{4:4}
] {\datatableB};
\end{axis}
\end{tikzpicture}
    \caption{\revii{Number of paths vs. required $E_{\mathrm{b}}/N_0$ ($\si{dB}$) for a target FER $10^{-3}$ of different \ac{aSCED}-configurations for two $5$G LDPC codes $\mathcal{C}_{\mathrm{5G}}(n,60)$, $n\in \{78,180\}$ obtained from base graph $2$ with lifting factor $Z_ \mathrm{c}=6$. All \ac{aSCED}-configurations leverage the \ac{RL} structure as described in Sec.~\ref{sec:rl_selection}.}}
    \vspace{-8pt}
    \label{fig:scatter_num_path_reqsnr}
\end{figure*}
\revii{In the following, we consider two $5$G LDPC codes $\mathcal{C}_{\mathrm{5G}}(n,60)$, $n\in \{78,180\}$ obtained from base graph~$2$ with lifting factor $Z_ \mathrm{c}=6$. The code $\mathcal{C}_{\mathrm{5G}}(78,60)$ corresponds to the highest-rate code, whereas $\mathcal{C}_{\mathrm{5G}}(180,60)$ includes $17$ additional QC blocks of the IRC part, resulting in a lower rate $r=\frac{1}{3}$.
In particular, the low-rate code is of interest for applications such as IoT and URLLC.
We compare different \ac{aSCED} configurations constructed according to Sec.~\ref{sec:rl_selection}, with varying $\Delta\in \{1,2,3,6\}$ and varying number of batches. 
To this end, we consider the following ensembles:
\begin{itemize}
    \item For $\Delta\in\{1,2,3\}$, we first increase the number of batches until all rows of the first QC row block are used. This corresponds to $L=6$, $L=3$, and $L=2$ batches for $\Delta=1$, $\Delta=2$, and $\Delta=3$, respectively. Afterwards, we add complete QC row blocks, i.e., we increase $L$ in steps of $6$, $3$, and $2$, respectively, up to
    a maximum of $48$, $24$, and $16$ batches.
    \item For $\Delta=6$, each batch corresponds to appending one complete QC row block from the IR part, and we consider
    $L\in\{1,2,3\}$ batches.
\end{itemize}
}

\revii{
Fig.~\ref{fig:scatter_num_path_reqsnr} shows the total number of paths over the $E_{\mathrm{b}}/N_0$ required for a target FER of $10^{-3}$ for the different \ac{aSCED} configurations. All \ac{BP} decodings use a maximum of ${I_\mathrm{max}=32}$ iterations. As baselines, we include \ac{NMSA} and \ac{AED}-$6$ according to \cite{stuttgart_ldpc_aed}. Furthermore, we include OSD references, namely OSD-$2$ for the shorter code, which coincides with \ac{ML} performance, and OSD-$4$ for the longer code.
}

\revii{
The required $E_{\mathrm{b}}/N_0$ was determined using a binary search over the SNR. For each candidate operating point, Monte Carlo simulations were performed until at least $1000$ FEs had been collected. The search terminated when the measured \ac{FER} was within $\pm 2\,\%$ of the target \ac{FER} or when the SNR search interval was smaller than $10^{-4}$\,dB. For the latter, the operating point with the \ac{FER} closest to the target was selected.}

\revii{
For both codes, Fig.~\ref{fig:scatter_num_path_reqsnr} shows a similar behavior. For a small number of paths, the $\Delta=1$ configurations offer the most favorable performance--complexity trade-off. However, after roughly $15$--$20$ paths, adding more $\Delta=1$ batches yields diminishing returns. For larger numbers of paths, the $\Delta=2$ configurations first offer the most favorable performance--complexity trade-off, before the $\Delta=3$ configurations become more competitive. This behavior indicates that increasing the number of batches is beneficial initially, but that increasing $\Delta$ becomes more effective once the gains from additional low-$\Delta$ batches saturate.}

\revii{
For the higher-rate code, several configurations approach the OSD reference, whereas for the lower-rate code a gap to the OSD-$4$ reference remains. Nevertheless, for the lower-rate code, several configurations with fewer than $50$ paths yield a gain of more than $0.9$\,dB compared to \ac{NMSA}. Overall, the results suggest the following practical
design rule for \ac{RL}-based \ac{aSCED}: keep $\Delta$ small and increase the number of batches until the gains start to diminish; only then increase $\Delta$.}

\subsection{\ac{BCH} Codes}

Throughout this section, unless stated otherwise, all \ac{BP} decoders employ the \ac{NMSA} with normalization constant $\alpha=\frac{1}{2}$ and a maximum of $I_\mathrm{max}=20$ iterations.

\subsubsection{Analysis of \acp{ssPCM} of Subcodes}\label{sec:structural_advantages_results}
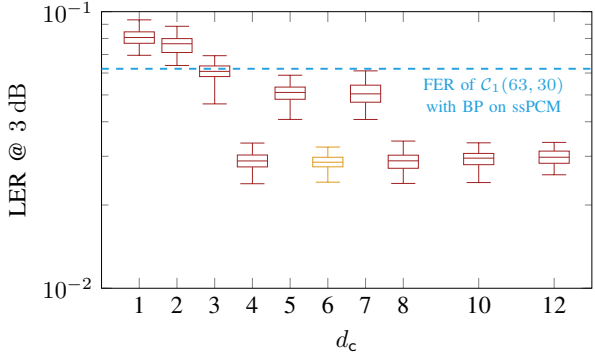
\begin{figure}[t]
    \centering
    \begin{tikzpicture}[scale=0.92,spy using outlines={rectangle, magnification=2}]
\begin{axis}[
    width=0.8\columnwidth,
    height=4cm,
    at={(0.758in,0.645in)},
    scale only axis,
    boxplot/draw direction=y,
    cycle list name= rbr,%
    ylabel style={font=\color{white!15!black}},
    xlabel style={font=\color{white!15!black}},
    boxplot/draw direction=y,
    ylabel={LER @ $3\;\mathrm{dB}$},
    xlabel={$d_\mathsf{c}$},
    xmin=0,
    xmax=13,
    ymode=log,
    ymin=1e-02,
    ymax=1e-01,
    yminorticks=true,
    xtick={10},
    xtick={1,2,3,4,5,6,7,8,10,12},       %
    xticklabels={1,2,3,4,5,6,7,8,10,12},  %
]

\draw[KITcyanblue, thick, dashed] 
    (axis cs:0,6.216972e-02) -- (axis cs:13,6.216972e-02) 
    node[
    pos=0.8,
    anchor=north,
    align=center,
    inner sep=3pt
]
{\scriptsize FER of $\mathcal{C}_1(63,30)$\\[-3pt]\scriptsize with BP on ssPCM};

\addplot+[
    boxplot prepared={
        median=0.08064516,
        upper quartile=0.0845311,
        lower quartile=0.07684941,
        upper whisker=0.09341429,
        lower whisker=0.06954103
    },
    boxplot/draw position=1
] coordinates {}; %

\addplot+[
    boxplot prepared={
        median=0.07656968,
        upper quartile=0.07992912,
        lower quartile=0.07117440,
        upper whisker=0.08853475,
        lower whisker=0.06389776
    },
    boxplot/draw position=2
] coordinates {}; %

\addplot+[
    boxplot prepared={
        median=0.06083654,
        upper quartile=0.06364867,
        lower quartile=0.0582581,
        upper whisker=0.06932409,
        lower whisker=0.04640371
    },
    boxplot/draw position=3
] coordinates {};

\addplot+[
    boxplot prepared={
        median=0.02886214,
        upper quartile=0.03032494,
        lower quartile=0.02746595,
        upper whisker=0.03349523, %
        lower whisker=0.02387490
    },
        boxplot/draw position=4
] coordinates {}; %

\addplot+[
    boxplot prepared={
        median=0.05104645,
        upper quartile=0.05337256,
        lower quartile=0.04822184,
        upper whisker=0.05891016,
        lower whisker=0.04077472
    },
    boxplot/draw position=5
] coordinates {};

\addplot+[
    boxplot prepared={
        median=0.02855715,
        upper quartile=0.02976097,
        lower quartile=0.02748110,
        upper whisker=0.03239916,
        lower whisker=0.02418672
    },
        boxplot/draw position=6
] coordinates {}; %

\addplot+[
    boxplot prepared={
        median=0.05045424,
        upper quartile=0.05423367,
        lower quartile=0.04705895,
        upper whisker=0.06108735,
        lower whisker=0.04075810
    },
    boxplot/draw position=7
] coordinates {}; 

\addplot+[
    boxplot prepared={
        median=0.02888938,
        upper quartile=0.03027909,
        lower quartile=0.02714902,
        upper whisker=0.03405415,
        lower whisker=0.02392917
    },
        boxplot/draw position=8
] coordinates {}; %

\addplot+[
    boxplot prepared={
        median=0.02954210,
        upper quartile=0.03072669,
        lower quartile=0.02797986,
        upper whisker=0.03355705,
        lower whisker=0.02411091
    },
        boxplot/draw position=10
] coordinates {}; %

\addplot+[
    boxplot prepared={
        median=0.02977077,
        upper quartile=0.03132465,
        lower quartile=0.02830511,
        upper whisker=0.03368137,
        lower whisker=0.02572016
    },
        boxplot/draw position=12
] coordinates {}; %

\end{axis}
\end{tikzpicture}
    \caption{List Error Rate (LER) of \ac{aSCED} batches of subcodes of the \ac{BCH} code $\mathcal{C}_1(63,30)$ using different check node degrees $d_\mathsf{c}$ for $\bm{M}_\ell$ with $m_\ell=1$ and ssPCMs. The box plot at $d_\mathsf{c}=6$ is depicted in \textcolor{KITorange}{orange} to indicate that it corresponds to the same box plot shown in Fig.~\ref{fig:bch63_30:performance_over_edges}.}
    \label{fig:bch63_30:sweep_dc}
    \vspace{-4pt}
\end{figure}
\begin{figure}[t]
    \centering
    \begin{tikzpicture}[scale=0.92,spy using outlines={rectangle, magnification=2}]
\begin{axis}[
    at={(0.758in,0.645in)},
    scale only axis,
    boxplot/draw direction=y,
    cycle list name= brr,%
    ylabel style={font=\color{white!15!black}},
    xlabel style={font=\color{white!15!black}},
    ylabel={LER @ $3\;\mathrm{dB}$},
    xlabel={$m_\ell$},%
    ymode=log,
    ymin=1e-05,
    ymax=1e-01,
    yminorticks=true,
    xmin=0,
    xmax=13,
    xtick={1,2,3,4,5,6,7,8,9,10,11,12,14},
    xticklabels={1,2,3,4,5,6,7,8,9,10,11,12},
    width=0.8\columnwidth,
    height=4cm
]

\draw[KITcyanblue, thick, dashed] 
    (axis cs:0,6.216972e-02) -- (axis cs:13,6.216972e-02) 
        node[
    pos=0.8,
    anchor=north,
    align=center,
    inner sep=3pt
]
{\scriptsize FER of $\mathcal{C}_1(63,30)$\\[-3pt]\scriptsize with BP on ssPCM};

\draw[black, thick, dashed] 
    (axis cs:0,2.200e-03) -- (axis cs:13,2.200e-03) 
        node[
    pos=0.2,
    anchor=north,
    align=center,
    inner sep=3pt
]
{\scriptsize FER of $\mathcal{C}_1(63,30)$\\[-3pt]\scriptsize with ML};

\addplot+[
    boxplot prepared={
        median=0.02855715,
        upper quartile=0.02976097,
        lower quartile=0.02748110,
        upper whisker=0.03239916,
        lower whisker=0.02418672
    },
        boxplot/draw position=1
] coordinates {}; %

\addplot+[
    boxplot prepared={
        median=0.025795617140345922,
        upper quartile=0.02687631,
        lower quartile=0.02422334,
        upper whisker=0.029859659599880562,
        lower whisker=0.02206287920573635
    },
    boxplot/draw position=2
] coordinates {};%

\addplot+[
    boxplot prepared={
        median=0.023803897025452055,
        upper quartile=0.02517405,
        lower quartile=0.02306141,
        upper whisker=0.028105677346824058,
        lower whisker=0.020159258139300473
    },
    boxplot/draw position=3
] coordinates {};%

\addplot+[
    boxplot prepared={
        median=0.01634796,
        upper quartile=0.01762621,
        lower quartile=0.01509314,
        upper whisker=0.02121341,
        lower whisker=0.01376652
    },
    boxplot/draw position=4
] coordinates {};%

\addplot+[
    boxplot prepared={
        median=0.00872243746659591,
        upper quartile=0.00940347,
        lower quartile=0.00810356,
        upper whisker=0.01040582726326743,
        lower whisker=0.007032348804500703
    },
    boxplot/draw position=5
] coordinates {};%

\addplot+[
    boxplot prepared={
        median=0.00568359,
        upper quartile=0.00622464,
        lower quartile=0.00513059,
        upper whisker=0.00725058,
        lower whisker=0.00383120
    },
    boxplot/draw position=6
] coordinates {};%

\addplot+[
    boxplot prepared={
        median=0.003471291864402857,
        upper quartile=0.00368241,
        lower quartile=0.00320859,
        upper whisker=0.00398358761900968,
        lower whisker=0.0026479193973335453
    },
    boxplot/draw position=7
] coordinates {};%
\addplot+[
    boxplot prepared={
        median=0.00237240,
        upper quartile=0.00253425,
        lower quartile=0.00224779,
        upper whisker=0.00295840,
        lower whisker=0.00172209
    },
    boxplot/draw position=8
] coordinates {};%

\addplot+[
    boxplot prepared={
        median=0.0020749444704878223,
        upper quartile=0.00231299,
        lower quartile=0.00190029,
        upper whisker=0.00286820593718629,
        lower whisker=0.0013743720837542349
    },
    boxplot/draw position=9
] coordinates {};%

\addplot+[
    boxplot prepared={
        median=0.00140261,
        upper quartile=0.00161557,
        lower quartile=0.00120080,
        upper whisker=0.00188136,
        lower whisker=0.00084306
    },
    boxplot/draw position=10
] coordinates {};%

\addplot+[
    boxplot prepared={
        median=0.0004867580993192732,
        upper quartile=0.00056116,
        lower quartile=0.00043083,
        upper whisker=0.0006917064397869544,
        lower whisker=0.0003486379586550245
    },
    boxplot/draw position=11
] coordinates {};

\addplot+[
    boxplot prepared={
        median=0.000107,
        upper quartile=0.000119,
        lower quartile=0.000097,
        upper whisker=0.000139,
        lower whisker=0.000070%
    },
    boxplot/draw position=12
] coordinates {};%

\end{axis}
\end{tikzpicture}
    \caption{LER of \ac{aSCED} batches of subcodes of the \ac{BCH} code $\mathcal{C}_1(63,30)$ using a varying number of checks $m_\ell$ of degree $d_\mathsf{c}=6$ in $\bm{M}_\ell$ and ssPCMs.
    \pbcol{We ensure that all appended rows are linearly independent, i.e., $m_\ell=\Delta_\ell$.}}
    \label{fig:bch63_30:sweep_rows}
\end{figure}
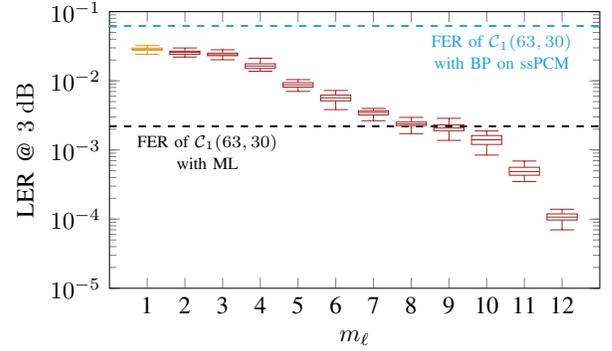

In the following, we analyze the structural benefit of decoding subcodes of the \ac{BCH} code $\mathcal{C}_1(63,30)$ using their \acp{ssPCM}.
We first investigate the impact of appending a single linearly independent row of varying Hamming weight $d_\mathsf{c}$ to the original \ac{PCM} $\bm{H}$.

To this end, for varying $d_\mathsf{c}$,
we sample $50$ distinct linearly independent rows and append each row to $\bm{H}$. 
\pbcol{Following Sec.~\ref{subsec:structual_adv}, we append the rows \emph{before} running Algorithm~\ref{alg:search_space}, which improves the sparsity of the respective search matrix.
Subsequently, we use Algorithm~\ref{alg:sspcm} with ${w_{\mathrm{max}}=2000}$ to construct the corresponding $\bm{H}_{\mathrm{ssPCM-II}}$}.

\pbcol{
For Fig.~\ref{fig:bch63_30:sweep_dc}, for each of the $50$ linear subcodes, we run a Monte Carlo simulation under the all-zero codeword assumption at $E_{\mathrm{b}}/N_0=3\,\mathrm{dB}$.
As per Remark~\ref{rem:az_sim}, this yields an estimate on the LER of the respective aSCED batch.
Therefore, Fig.~\ref{fig:bch63_30:sweep_dc} presents box plots of the resulting \acp{LER}.  
As stated in Remark~\ref{rem:az_sim_block}, in the absence of ML errors, the LER equals the FER of the full \ac{aSCED} batch.}
As a baseline, we include the \ac{FER} of \ac{BP} decoding using $\bm{H}_{\mathrm{ssPCM-II}}$ of the original \ac{BCH} code \cite{yifei_sspcm}.
For $d_\mathsf{c} \geq 3$, we consistently observe substantial gains, indicating that the enriched search space introduced by the appended row improves the structure of the \ac{ssPCM}.
Notably, for $\mathcal{C}_1(63,30)$, even check node degrees yield lower \ac{FER}.
However, this behavior is code-dependent: simulations on $\mathcal{C}(63,36)$ suggest the opposite trend, where odd $d_\mathsf{c}$ appears \revii{(slightly)} beneficial.
\pbcol{This implies that $d_\mathsf{c}$ must be optimized per code.
While all $d_\mathsf{c} \geq 3$ yield substantial gains, the interquartile ranges suggest that further optimization is possible.}

Next, we fix $d_\mathsf{c}=6$, and vary the number of appended rows $m_\ell$ while ensuring $\Delta_\ell = m_\ell$.
For every $m_\ell$, we sample $50$ subcodes and construct $\bm{H}_{\mathrm{ssPCM-II}}$ with $w_\mathrm{max}=2000$.

Fig.~\ref{fig:bch63_30:sweep_rows} displays the \ac{LER} box plots at $E_{\mathrm{b}}/N_0=3\,\mathrm{dB}$ for the different $m_\ell$ under the all-zero codeword assumption (Remark~\ref{rem:az_sim_block}).
For larger $m_\ell$, this significantly reduces the simulation complexity by assessing only one path of each \ac{aSCED} batch, which comprises $2^{m_\ell}$ paths.
For comparison, we include the performance of \ac{BP} decoding on $\bm{H}_{\mathrm{ssPCM-II}}$ of the original \ac{BCH} code and its \ac{ML} performance \cite{channelcodes}.

With increasing $m_\ell$, the average \ac{LER} decreases monotonically.
Notably, for $m_\ell \geq 8$, the \ac{LER} of the \ac{aSCED} batch can be lower than the \ac{ML} performance of the original \ac{BCH} code.
This is \emph{not} contradictory: as $m_\ell$ increases, the assumption in Remark~\ref{rem:az_sim}---no \ac{ML} errors occur---becomes invalid, and an estimate of the performance of the ensemble can only be attained by transmitting an arbitrary codeword and using the full \ac{aSCED} batch.
Still, the all-zero case provides insights: for $m_\ell \geq 8$, the dominant source of errors shifts from individual \ac{BP} decoding failures to \ac{ML} failures. This suggests that \ac{aSCED} with $2^8 = 256$ paths already approaches \ac{ML} performance,
despite possessing a significantly smaller number of paths than the full $2^{30}$‑path ensemble referenced in Proposition~\ref{lemma:aSCED-ML}.

\pbcol{In summary, the use of subcodes enables the construction of \acp{ssPCM} from sparser rows, yielding optimized \acp{GPCM} that are better suited for \ac{BP} than those of the original code.}
  
\subsubsection{Analysis of the Efficiency of \ac{aSCED} for Classical Codes}

\begin{figure}[t]
    \centering
    \input{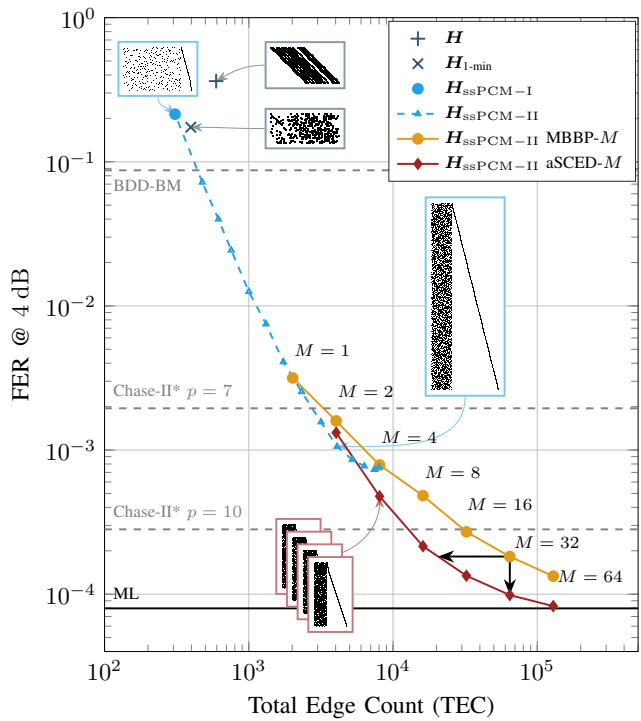}
    \caption{Performance of \ac{BP}-based decoders for the \ac{BCH} code $\mathcal{C}_1(63, 30)$  over TEC, and visualization of PCMs and \acp{ssPCM} used in those \ac{BP} decodings. All
\ac{BP} decodings employ the \ac{NSPA} with $\alpha=\frac{1}{2}$ and $I_\mathrm{max}=20$. 
Chase-II* denotes a performance estimate for the Chase-II decoder \cite{971758}.}%
\vspace{-2pt}
    \label{fig:bch63_30:performance_over_edges}
\end{figure}
In the following, we consider the \ac{BCH} code $\mathcal{C}_1(63,30)$, to analyze the performance of \ac{BP}-based \ac{aSCED} for classical codes, i.e., for non-LDPC codes. 
For constructing an \mbox{\ac{aSCED}-$M$} ensemble, we use $L=\frac{M}{2}$ distinct \ac{aSCED} batches with $m_\ell=1$ ($\Delta_\ell=1$), unless stated otherwise. 

\pbcol{First, we demonstrate that, due to the structural benefit of subcodes for \ac{GPCM} optimization, \ac{aSCED} can exploit decoding complexity more efficiently.}
For \ac{BP}-based decoding with a fixed number of maximum iterations and \revii{scheduling}, the worst-case \revii{computational} complexity is mainly dominated by the total edge count (TEC)
\[
\text{TEC}:=\sum_{\ell\in[M]} \mathrm{wt}(\bm{H}_\ell),
\] 
\revii{since the number of message updates scales approximately linearly with the TEC. 
Yet, in terms of hardware metrics such as area and throughput, the TEC only serves as an estimate \cite[Fig.~1]{burg_vlsi_tradeoffs}\cite[Sec.~5.5.3]{beuschel2010fully}, since these metrics depend on additional architectural design choices.
The exact analysis of required hardware resources is beyond the scope of this paper.
}

In Fig.~\ref{fig:bch63_30:performance_over_edges}, we plot the \ac{FER} at ${E_{\mathrm{b}}/N_0=4\,\mathrm{dB}}$ of \ac{NSPA} with normalization constant ${\alpha=\frac{1}{2}}$ and maximum number of ${I_\mathrm{max}=20}$ iterations, versus the TEC for different \ac{BP}-based schemes, namely \pbcol{\ac{BP} using $\bm{H}$, $\bm{H}_\text{1-min}$ from \cite{channelcodes}, and \acp{ssPCM}, as well as} \ac{MBBP}-$M$ and \ac{aSCED}-$M$.
The respective \ac{GPCM} designs are detailed below and visualized in Fig.~\ref{fig:bch63_30:performance_over_edges}. 
\pbcol{We conducted the same analysis using \ac{NMSA}; however, \ac{NSPA} consistently yields slightly better performance. A comparison between NSPA and NMSA is provided later.}
Note that only $\bm{H}$, $\bm{H}_{1-\text{min}}$ and $\bm{H}_{\mathrm{ssPCM-I}}$ are not overcomplete. As a baseline performance of classical decoders,
 we include the \ac{FER} performance of bounded-distance decoding based on the Berlekamp–Massey (BDD-BM) algorithm 
 and a performance estimate for Chase-II decoding using $2^p$ test patterns\cite{Massey1969ShiftregisterSA,971758}. 

For stand-alone \ac{BP} decoding, we use the \ac{PCM} $\bm{H}_\text{1-min}$ with minimized weight \cite{channelcodes}, and
for the \acp{ssPCM} of the \ac{BCH} code $\mathcal{C}_1(63,30)$, we first construct a large \ac{ssPCM}-II of the \ac{BCH} code  
based on a search space matrix $\bm{S}_\text{min}$ consisting of all minimum weight dual codewords and using Algorithm~\ref{alg:sspcm} with $w_\mathrm{max}=8000$.
 From this large \ac{ssPCM}-II, we derive a series of smaller \ac{ssPCM}-IIs by iteratively removing the last PCRB and its associated AVNs, 
 down to the \ac{ssPCM}-I.

For \ac{MBBP}, we construct $64$ different \ac{ssPCM}-IIs of $\mathcal{C}_1(63,30)$ by running  Algorithm~\ref{alg:sspcm} with ${w_\mathrm{max}=2000}$ using $64$ randomly sampled subsets of $\bm{S}_\text{min}$.
Then, \ac{MBBP}-$M$ consists of $M$ randomly selected \acp{ssPCM}.
Note that the idea of using \acp{ssPCM} within \ac{MBBP} was already mentioned in \cite{yifei_sspcm}.

\pbcol{When constructing subcode \acp{ssPCM} for code $\mathcal{C}_1(63,30)$, we empirically observe a saturation effect around $w_\mathrm{max}=2000$, i.e., increasing $w_\mathrm{max}$ beyond this value does not yield further \ac{FER} improvements for the corresponding \ac{aSCED} batch, while increasing decoding complexity.}
Hence, for \ac{aSCED}-$M$, we select $32$ distinct $(k-1)$-dimensional linear subcodes obtained by appending a single linearly independent row with random $d_\mathsf{c} \in \{6,8,10\}$ and running Algorithm~\ref{alg:sspcm} with $w_\mathrm{max}=2000$ to construct the respective \ac{ssPCM}-II.
The considered ensembles consist of $L$
randomly selected \ac{aSCED} batches from the $32$ constructed ones, yielding a total ensemble size of ${M=2L}$.

The \ac{aSCED}-$M$ proposed in this work outperforms \ac{MBBP}-$M$ at all evaluated ensemble sizes.
For a fixed TEC, \ac{aSCED} achieves a lower \ac{FER}, and for a fixed \ac{FER}, it requires a lower TEC.
More specifically, as highlighted in Fig.~\ref{fig:bch63_30:performance_over_edges}, \ac{aSCED}-$32$ reduces the TEC by a factor of around $3$ in order to match the \ac{FER} performance of \ac{MBBP}-$32$.
Similarly, \ac{aSCED}-$32$ reduces the \ac{FER} by a factor of roughly $2$ for a fixed TEC.
For a ${\text{TEC}>6000}$, \ac{aSCED}-$M$ yields lower \ac{FER} compared to stand-alone \ac{BP} decoding using \ac{ssPCM}-II, for which a saturation effect can be observed.
Notably, \ac{aSCED}-$64$ approaches \ac{ML} performance using only $2^6$ paths of \ac{NSPA}, i.e., using even fewer paths than the single \ac{aSCED} batch with $\Delta_\ell=8$ from Fig.~\ref{fig:bch63_30:sweep_rows}, which uses \ac{NMSA}.

\pbcol{As for LDPC codes, this raises the question: is it preferable to append multiple single rows, i.e., use $\Delta=1$ and increase $L$, or can a single larger \ac{aSCED} batch, i.e., let $L=1$ and increase $\Delta$, yield better trade-offs? We address this in the next section.}

\subsubsection{Implementation Benefits vs. Performance Trade-off}
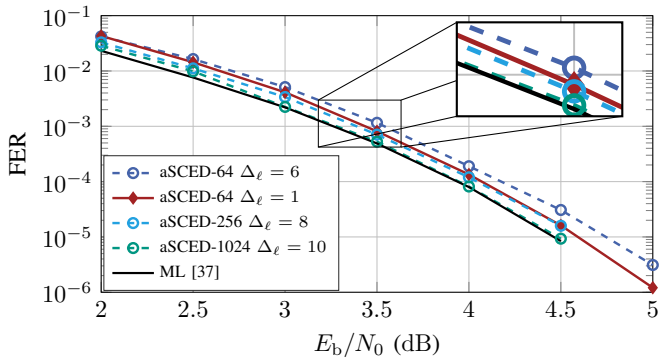
\begin{figure}[t]
    \centering
    \begin{tikzpicture}[scale=0.92,spy using outlines={rectangle, magnification=2}]

\begin{axis}[%
width=.9\columnwidth,
height=4cm,
at={(0.758in,0.645in)},
scale only axis,
xmin=2,
xmax=5,
xlabel style={font=\color{white!15!black}},
xlabel={$E_{\mathrm{b}}/N_0$ ($\si{dB}$)},
ymode=log,
ymin=1e-06,
ymax=1e-01,
yminorticks=true,
ylabel style={font=\color{white!15!black}},
ylabel={FER},
axis background/.style={fill=white},
xmajorgrids,
ymajorgrids,
legend style={at={(0.004,0.005)}, anchor=south west, legend cell align=left, align=left, draw=white!15!black,font=\scriptsize}
]

\addplot[color=KITblue,line width = 1pt,mark=o,dashed, mark options={solid}]
table[row sep=crcr]{
 2.00  4.209640e-02\\
 2.50  1.642441e-02\\
 3.00  5.112344e-03\\
 3.50  1.156972e-03\\
 4.00  1.898940e-04\\
 4.50  3.071564e-05\\
 5.00  3.110526e-06\\
};\addlegendentry{aSCED-64 $\Delta_\ell=6$}

\addplot[color=KITred,solid,line width = 1pt,mark=diamond*, mark options={solid}]
table[row sep=crcr]{
  0.00  4.322000e-01\\
  0.50  2.879000e-01\\
  1.00  1.837000e-01\\
  1.50  9.789000e-02\\
  2.00  4.296000e-02\\
  2.50  1.439000e-02\\
  3.00  4.119000e-03\\
  3.50  8.116000e-04\\
  4.00  1.345000e-04\\
  4.50  1.615000e-05\\
  5.00  1.208000e-06\\
   };
\addlegendentry{aSCED-64 $\Delta_\ell=1$};%

 \addplot[color=KITcyan,line width = 1pt,mark=o,dashed, mark options={solid}]
table[row sep=crcr]{
 2.00  3.285691e-02\\
 2.50  1.114454e-02\\
 3.00  3.380206e-03\\
 3.50  7.075538e-04\\
 4.00  1.196210e-04\\
4.5 1.5930735930735932e-05\\
};\addlegendentry{aSCED-256 $\Delta_\ell=8$}

 \addplot[color=KITgreen,line width = 1pt,mark=o,dashed, mark options={solid}]
table[row sep=crcr]{
 2.00  2.862459e-02\\
 2.50  9.862906e-03\\
 3.00  2.241650e-03\\
 3.50  5.290180e-04\\
 4.00  8.156936e-05\\
4.5 9.297625063163214e-06\\
};\addlegendentry{aSCED-1024 $\Delta_\ell=10$}

\addplot [color=black, line width=0.9pt]
  table[row sep=crcr]{%
  0.00   3.215e-01\\
  0.50  1.969e-01\\
  1.00   1.161e-01\\
  1.50  5.631e-02\\
  2.00   2.315e-02\\
  2.50   7.628e-03\\
  3.00   2.200e-03\\
  3.50   4.896e-04\\
  4.00  7.973e-05\\
  4.50   8.396e-06\\
};
\addlegendentry{ML \cite{channelcodes}}%
\coordinate (spypoint) at (axis cs:3.25,0.0005);
			\coordinate (spyviewer) at (axis cs:4.15,0.004);	
			\spy[width=2.2cm,height=1.25cm, thin, spy connection path={
            \draw (tikzspyonnode.south west) -- (tikzspyinnode.south west);
            \draw (tikzspyonnode.south east) -- (tikzspyinnode.south east);
			\draw (tikzspyonnode.north west) -- (tikzspyinnode.north west);\draw (tikzspyonnode.north east) -- (intersection of  tikzspyinnode.north east--tikzspyonnode.north east and tikzspyinnode.north west--tikzspyinnode.south west);
			;}] on (spypoint) in node at (spyviewer);
		\coordinate (a) at ($(axis cs:-10.8/1.4,-0.12)+(spyviewer)$);
		\coordinate[label={[font=\small,text=black]right:$10^{-2}$}] (b) at ($(axis cs:+10.8/1.4,-0.12)+(spyviewer)$);
        
\end{axis}

\end{tikzpicture}%
    \caption{Performance of \ac{aSCED} for the \ac{BCH} code $\mathcal{C}_1(63,30)$ using different $\Delta_\ell$ and ensemble sizes.
    }
    \label{fig:bch63_30_subcode}
\end{figure}

Fig.~\ref{fig:bch63_30_subcode} plots the \ac{FER} versus $E_{\mathrm{b}}/N_0$ performance for the \ac{BCH} code $\mathcal{C}_1(63,30)$ of various
\ac{aSCED}-$M$ decoders consisting of \ac{BP} decoders of subcodes with varying rank deficiencies $\Delta_\ell$. 

For the ensembles consisting of a single \ac{aSCED} batch, i.e., $M = 2^{\Delta_\ell}$, the best-performing batch was selected for each $\Delta_\ell$ from the corresponding $50$ candidates shown in Fig.~\ref{fig:bch63_30:sweep_rows}.
The ensemble of $32$ \ac{aSCED} batches with $\Delta_\ell=1$ is the one from Fig.~\ref{fig:bch63_30:performance_over_edges} evaluated using \ac{NMSA}.

\pbcol{As $\Delta_\ell$ increases, the \ac{FER} of \ac{aSCED}-$2^{\Delta_\ell}$, i.e., an ensemble consisting of a single large \ac{aSCED} batch $(L=1)$, decreases monotonically, fully closing the gap to \ac{ML} performance at ${\Delta_\ell=10}$.
When comparing the two depicted $64$-path ensembles, i.e., one composed of ${L=32}$ \ac{aSCED} batches with ${\Delta_\ell=1}$, and the other of a single \ac{aSCED} batch with ${L=1}$, ${\Delta_\ell=6}$, the former outperforms the latter at comparable TEC.}

This indicates that the diversity gain from using multiple distinct \ac{aSCED} batches (fixing $\Delta_\ell=1$ and increasing $L$) provides a performance advantage over a single large batch (fixed $L=1$ and increasing $\Delta_\ell=1$) and can be reasoned as follows:
Assume we randomly sample $L$ rows $\bm{M}_\ell$, $m_\ell=1$, that are jointly linearly independent of the rows of $\bm{H}$.
Appending each of those rows to $\bm{H}$ according to \eqref{eq:inducing_subcode} results in $L$ distinct linear subcodes $\mathcal{C}_\ell\subset_\mathsf{s}\mathcal{C}$ 
and their respective $(k+1)$-dimensional dual codes $\mathcal{C}_\ell^\bot\supset_\mathsf{s}\mathcal{C}^\bot$, $\ell\in[L]$.
Since the appended rows are jointly linearly independent of the rows of $\bm{H}$, 
it follows that $\bigcap_{\ell\in[L]}\mathcal{C}_\ell^\bot=\mathcal{C}^\bot$ and $\bigcap_{\ell\in[L]}(\mathcal{C}_\ell^\bot\setminus\mathcal{C}^\bot)=\emptyset$.
Hence, each dual code of a linear subcode contains a subset $\mathcal{C}_\ell^\bot\setminus\mathcal{C}^\bot$ consisting of $2^k$ dual codewords that are not elements of $\mathcal{C}_{\ell'}$, ${\ell'\in[L]\setminus\{\ell\}}$.
Then, if appending a row $\bm{M}_\ell$ with $\mathrm{wt}(\bm{M}_\ell)<d_\mathrm{min}^\bot$ and following the reasoning in Sec.~\ref{subsec:structual_adv},
the search space matrix $\bm{S}_\ell$ used for constructing the \ac{ssPCM} tailored to subcode $\mathcal{C}_\ell$ typically mostly consists 
of rows which are elements of $\mathcal{C}_\ell^\bot\setminus\mathcal{C}^\bot$.
Consequently, the \ac{ssPCM} for each of the \ac{aSCED} batches is built mostly out of different rows and likely possesses distinct 
resulting in \emph{inter-path diversity}, which boosts performance when using increased $L$.

\jmcol{However, when latency is not the limiting factor and some loss in decoding performance is acceptable}, fixing $L$ and increasing $\Delta$ potentially offers a key advantage in hardware implementation: since all decodings within a single \ac{aSCED} batch share the same underlying graph, the same hardware could be reused in a serialized fashion, thereby reducing area overhead.
Thus, \ac{aSCED} can allow for a balanced trade-off among performance, hardware area, and latency.

\begin{remark} 
As expected, when dropping the all-zero codeword assumption, the \ac{FER} of ensembles with \ac{LER} below \ac{ML} in Fig.~\ref{fig:bch63_30:sweep_rows} is lower-bounded by the \ac{ML} performance.
Notably, those ensembles achieve performance close to \ac{ML}, underlining that the all-zero assumption enables a practically efficient simulation during ensemble design.
\end{remark}

\subsubsection{\ac{FER} Performance}

\begin{figure}[t]
    \centering
    \begin{tikzpicture}[scale=0.92,spy using outlines={rectangle, magnification=2}]

\begin{axis}[%
width=.9\columnwidth,
height=\ferfigheight,
at={(0.758in,0.645in)},
scale only axis,
xmin=2,
xmax=7,
xlabel style={font=\color{white!15!black}},
xlabel={$E_{\mathrm{b}}/N_0$ ($\si{dB}$)},
ymode=log,
ymin=1e-06,
ymax=1,
yminorticks=true,
ylabel style={font=\color{white!15!black}},
ylabel={FER},
axis background/.style={fill=white},
xmajorgrids,
ymajorgrids,
legend style={
    at={(1,-0.25)},
    anchor=north east,
    legend columns=4,
    legend cell align=left,
    draw=white!15!black,
    font=\scriptsize
}
]

   \addplot[color=kit-royalblue80,line width = 1pt,mark=x, mark options={solid}]
table[row sep=crcr]{
  2.00  6.766000e-01\\
  2.50  5.369000e-01\\
  3.00  3.936000e-01\\
  3.50  2.567000e-01\\
  4.00  1.445000e-01\\
  4.50  7.172000e-02\\
  5.00  2.901000e-02\\
  5.50  9.900000e-03\\
  6.00  2.751000e-03\\
  6.50  7.059000e-04\\
  7.00  1.567000e-04\\
};
\addlegendentry{$\bm{H}_\text{1-min}$};

\addplot[color=KITorange!70,dotted,line width = 1pt,mark=o, mark options={solid}]
table[row sep=crcr]{
  0.00  6.095000e-01\\
  0.50  4.582000e-01\\
  1.00  3.171000e-01\\
  1.50  1.891000e-01\\
  2.00  9.009000e-02\\
  2.50  4.396000e-02\\
  3.00  1.520000e-02\\
  3.50  3.701000e-03\\
  4.00  7.881000e-04\\
  4.50  1.064000e-04\\
  5.00  1.124000e-05\\
  5.50  7.036000e-07\\
  };
\addlegendentry{MBBP-8};%

\addplot[color=KITorange,line width = 1pt,mark=o, mark options={solid}]
table[row sep=crcr]{
  0.00  4.783000e-01\\
  0.50  3.490000e-01\\
  1.00  2.227000e-01\\
  1.50  1.101000e-01\\
  2.00  5.056000e-02\\
  2.50  2.030000e-02\\
  3.00  5.149000e-03\\
  3.50  1.230000e-03\\
  4.00  1.965000e-04\\
  4.50  2.487000e-05\\
  5.00  1.667000e-06\\
  };
\addlegendentry{MBBP-64};%

\addplot[color=black!50, line width = 1pt]
table[row sep=crcr]{
0.00 8.810573e-01 \\
0.50 9.090909e-01 \\
1.00 8.163265e-01 \\
1.50 6.688963e-01 \\
2.00 5.494505e-01 \\
2.50 4.310345e-01 \\
3.00 2.773925e-01 \\
3.50 1.672241e-01 \\
4.00 8.729812e-02 \\
4.50 4.073320e-02 \\
5.00 1.666528e-02 \\
5.50 5.132153e-03 \\
6.00 1.151192e-03 \\
6.50 2.390266e-04 \\
};
\addlegendentry{BDD-BM};

\addplot[color=KITcyanblue,dashed,line width = 1pt,mark=triangle, mark options={solid}]
table[row sep=crcr]{
 2.00  2.1e-01\\
 2.50  1.1e-01\\
 3.00  5e-02\\
 3.50  1.8e-02\\
 4.00  5.5e-03\\
 4.5   1.1e-3\\
 5     1.8e-4\\
 5.5   2e-5\\
 6     1.7e-6\\
   };
\addlegendentry{ssPCM-II };%

\addplot[color=KITred!70,dotted, mark=*, line width=1pt, mark=diamond*, mark options={solid}]
table[row sep=crcr]{
  0.00  4.990000e-01\\
  0.50  3.627000e-01\\
  1.00  2.351000e-01\\
  1.50  1.392000e-01\\
  2.00  6.522000e-02\\
  2.50  2.641000e-02\\
  3.00  8.306000e-03\\
  3.50  1.856000e-03\\
  4.00  3.755000e-04\\
  4.50  4.802000e-05\\
  5.00  4.076000e-06\\
};
\addlegendentry{aSCED-8}; %

\addplot[color=KITred,solid,line width = 1pt,mark=diamond*, mark options={solid}]
table[row sep=crcr]{
  0.00  4.322000e-01\\
  0.50  2.879000e-01\\
  1.00  1.837000e-01\\
  1.50  9.789000e-02\\
  2.00  4.296000e-02\\
  2.50  1.439000e-02\\
  3.00  4.119000e-03\\
  3.50  8.116000e-04\\
  4.00  1.345000e-04\\
  4.50  1.615000e-05\\
  5.00  1.208000e-06\\
   };
\addlegendentry{aSCED-64 NMSA};%

\addplot [color=black!30,dashed, line width=0.9pt]
  table[row sep=crcr]{%
0.0   0.6008491679318061\\
0.5   0.47105110616957496\\
1.0   0.33803007408699765\\
1.5   0.21745768541589\\
2.0   0.12257279040119042\\
2.5   0.05905089181685509\\
3.0   0.023671458773727696\\
3.5   0.007671025454008002\\
4.0   0.001948235164018641\\
4.5   0.0003750875697260048\\
5.0   5.282397509470336e-05\\
5.5  5.237771613413546e-06\\
6.0   3.5102733167881437e-07\\
6.5   1.522188975540518e-08\\
};
\addlegendentry{Chase-II* $p=7$}

\addplot [color=black, line width=0.9pt]
  table[row sep=crcr]{%
  0.00   3.215e-01\\
  0.50  1.969e-01\\
  1.00   1.161e-01\\
  1.50  5.631e-02\\
  2.00   2.315e-02\\
  2.50   7.628e-03\\
  3.00   2.200e-03\\
  3.50   4.896e-04\\
  4.00  7.973e-05\\
  4.50   8.396e-06\\
};
\addlegendentry{ML \cite{channelcodes}}

\addplot[color=KITpalegreen,dotted,  line width=1pt, mark=square, mark options={solid}]
table[row sep=crcr]{
2 0.07323076923076922\\
2.5 0.02909581977779142\\
3 0.011286303076775616\\
3.5 0.0026654820079964462\\
4 0.0005306573761916778\\
4.5 7.832965279157499e-05\\
5 6.346430037907227e-06\\
5.5 7.199436716071334e-07\\
};
\addlegendentry{SCED-8};

\addplot[color=KITgreen,solid,line width = 1pt,mark=diamond*, mark options={solid}]
table[row sep=crcr]{
  0.00  6.693000e-01\\
  0.50  4.449000e-01\\
  1.00  2.379000e-01\\
  1.50  1.085000e-01\\
  2.00  4.108000e-02\\
  2.50  1.145000e-02\\
  3.00  2.964000e-03\\
  3.50  5.802000e-04\\
  4.00  8.262000e-05\\
  4.5 9.626714202817796e-06\\
  5 8.154186823426652e-07\\
   };
\addlegendentry{aSCED-64 NSPA};

\addplot [color=black!50,dashed, line width=0.9pt]
  table[row sep=crcr]{%
0.0   0.4267674645183627\\
0.5   0.30008288898554103\\
1.0   0.1889732327858183\\
1.5   0.10420193386182185\\
2.0   0.04908426507249242\\
2.5   0.019226858551741874\\
3.0   0.0060822193730184\\
3.5   0.001505228857460796\\
4.0   0.00028155383102404177\\
4.5   3.834597487179901e-05\\
5.0   3.6519128105744038e-06\\
5.5   2.3279615366139582e-07\\
6.0   9.474747215791802e-09\\
6.5   2.339544158474056e-10\\
};
\addlegendentry{Chase-II* $p=10$}

\coordinate (spypoint) at (axis cs:3.35,0.0004);
			\coordinate (spyviewer) at (axis cs:2.75,0.000008);	
			\spy[width=2.2cm,height=1.25cm, thin, spy connection path={\draw (tikzspyonnode.south west) -- (intersection of  tikzspyonnode.south west--tikzspyinnode.south west and tikzspyinnode.north east--tikzspyinnode.north west);\draw (tikzspyonnode.south east) -- (tikzspyinnode.south east);
			\draw (tikzspyonnode.north west) -- (tikzspyinnode.north west);\draw (tikzspyonnode.north east) -- (tikzspyinnode.north east);
			;}] on (spypoint) in node at (spyviewer);
		\coordinate (a) at ($(axis cs:-10.8/1.4,-0.12)+(spyviewer)$);
		\coordinate[label={[font=\small,text=black]right:$10^{-2}$}] (b) at ($(axis cs:+10.8/1.4,-0.12)+(spyviewer)$);

\end{axis}

\end{tikzpicture}%
    \caption{Performances of \ac{BP}-based decoders and classical decoders for the \ac{BCH} code $\mathcal{C}_1(63,30)$. Unless stated otherwise, \ac{NMSA} with $\alpha=\frac{1}{2}$ and $I_\mathrm{max}=20$ is used for \ac{BP}. All \ac{aSCED}-$M$ use $\Delta=1$ and $L=\frac{M}{2}$. Chase-II* denotes the performance estimate for the Chase-II decoder derived in \cite{971758}.}
    \label{fig:bch63_30}
\end{figure}
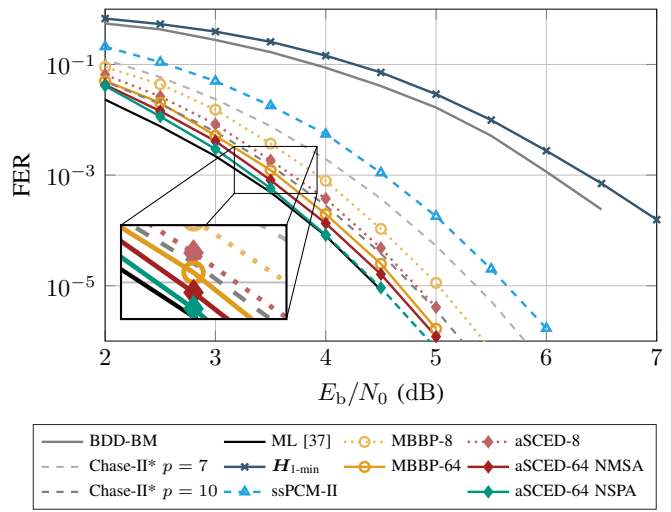

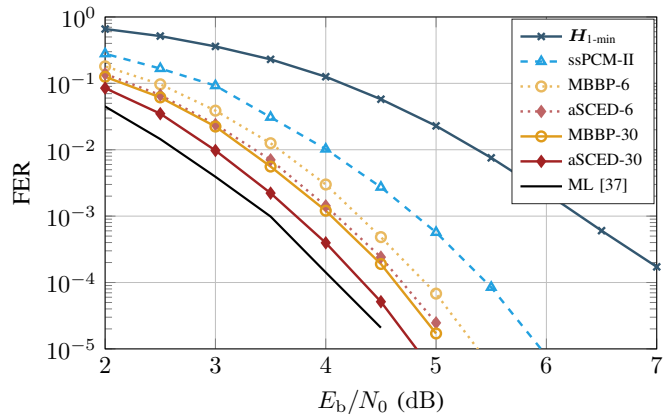
\begin{figure}[t]
    \centering
    \begin{tikzpicture}[scale=0.92,spy using outlines={rectangle, magnification=2}]

\begin{axis}[%
width=.9\columnwidth,
height=\ferfigheight,
at={(0.758in,0.645in)},
scale only axis,
xmin=2,
xmax=7,
xlabel style={font=\color{white!15!black}},
xlabel={$E_{\mathrm{b}}/N_0$ ($\si{dB}$)},
ymode=log,
ymin=1e-05,
ymax=1,
yminorticks=true,
ylabel style={font=\color{white!15!black}},
ylabel={FER},
axis background/.style={fill=white},
xmajorgrids,
ymajorgrids,
legend style={at={(0.996,0.995)}, anchor=north east, legend cell align=left, align=left, draw=white!15!black,font=\scriptsize}
]

\addplot[color=kit-royalblue80,line width = 1pt,mark=x, mark options={solid}]
table[row sep=crcr]{
  2.00  6.590000e-01\\
  2.50  5.157000e-01\\
  3.00  3.603000e-01\\
  3.50  2.299000e-01\\
  4.00  1.259000e-01\\
  4.50  5.766000e-02\\
  5.00  2.289000e-02\\
  5.50  7.558000e-03\\
  6.00  2.216000e-03\\
  6.50  6.040000e-04\\
  7.00  1.722000e-04\\
};
\addlegendentry{$\bm{H}_\text{1-min}$};

\addplot[color=KITcyanblue,dashed,line width = 1pt,mark=triangle, mark options={solid}]
table[row sep=crcr]{
    1.00  5.822416e-01\\
 1.50  4.094166e-01\\
 2.00  2.791347e-01\\
 2.50  1.672940e-01\\
 3.00  9.229349e-02\\
 3.50  3.086658e-02\\
 4.00  1.023830e-02\\
 4.50  2.726858e-03\\
 5.00  5.688331e-04\\
 5.50  8.466123e-05\\
 6.00  7.897723e-06\\
  };
\addlegendentry{ssPCM-II};%

\addplot[color=KITorange!70, dotted, line width = 1pt,mark=o, mark options={solid}]
table[row sep=crcr]{
  0.00  7.654000e-01\\
  0.50  6.161000e-01\\
  1.00  4.817000e-01\\
  1.50  3.385000e-01\\
  2.00  1.805000e-01\\
  2.50  9.683000e-02\\
  3.00  3.879000e-02\\
  3.50  1.254000e-02\\
  4.00  2.996000e-03\\
  4.50  4.829000e-04\\
  5.00  6.799000e-05\\
  5.50  5.356000e-06\\
  };
\addlegendentry{MBBP-6};%

\addplot[color=KITpalegreen,dotted,  line width=1pt, mark=square, mark options={solid}]
table[row sep=crcr]{
2 0.15622076707202995\\
2.5 0.07478397486252945\\
3 0.030316289536192025\\
3.5 0.009225382587828002\\
4 0.002250266478925136\\
4.5 0.0004592890279925551\\
5 7.220287571579134e-05\\
5.5 8.411848215182379e-06\\
 };
 \addlegendentry{SCED-6};

\addplot[color=KITred!70, dotted, line width = 1pt,mark=diamond*, mark options={solid}]
table[row sep=crcr]{
 2.00  1.390821e-01\\
 2.50  6.668890e-02\\
 3.00  2.401537e-02\\
 3.50  7.042006e-03\\
 4.00  1.427399e-03\\
 4.50  2.410861e-04\\
 5.00  2.466282e-05\\
 };
 \addlegendentry{aSCED-6};

\addplot[color=KITorange,line width = 1pt,mark=o, mark options={solid}]
table[row sep=crcr]{
  0.00  7.088000e-01\\
  0.50  5.907000e-01\\
  1.00  4.084000e-01\\
  1.50  2.536000e-01\\
  2.00  1.257000e-01\\
  2.50  6.146000e-02\\
  3.00  2.235000e-02\\
  3.50  5.579000e-03\\
  4.00  1.207000e-03\\
  4.50  1.901000e-04\\
  5.00  1.707000e-05\\
  };
\addlegendentry{MBBP-30};

\addplot[color=KITred,line width = 1pt,mark=diamond*, mark options={solid}]
table[row sep=crcr]{
 2.00  8.478169e-02\\
 2.50  3.493450e-02\\
 3.00  9.743740e-03\\
 3.50  2.219509e-03\\
 4.00  3.959502e-04\\
 4.50  5.135809e-05\\
 5.00  4.492969516549288e-06\\ %
 5.5 2.3402024526693324e-07\\
 };
 \addlegendentry{aSCED-30};

\addplot [color=black, line width=0.9pt]
  table[row sep=crcr]{%
  0.00   4.505e-01\\
  0.50   3.145e-01\\
  1.00   1.739e-01\\
  1.50   1.018e-01\\
  2.00   4.482e-02\\
  2.50   1.447e-02\\
  3.00   3.916e-03\\
  3.50   9.864e-04\\
  4.00   1.419e-04\\
  4.50   2.086e-05\\
};
\addlegendentry{ML \cite{channelcodes}}

\end{axis}

\end{tikzpicture}%
    \caption{Performances of \ac{BP}-based decoders for the \ac{BCH} code $\mathcal{C}_2(63,36)$. All \ac{BP} decoders employ the \ac{NMSA} with normalization constant $\alpha=\frac{1}{2}$ and a maximum of $I_\mathrm{max}=20$ iterations. %
    }
    \label{fig:bch63_36}
\end{figure}
\begin{figure}[t]
    \centering
    \begin{tikzpicture}[scale=0.92]
\begin{axis}[%
width=.9\columnwidth,
height=\ferfigheight,
at={(0.758in,0.645in)},
scale only axis,
xmin=1,
xmax=5,
xlabel style={font=\color{white!15!black}},
xlabel={$E_{\mathrm{b}}/N_0$ ($\si{dB}$)},
ymode=log,
ymin=1e-05,
ymax=1,
yminorticks=true,
ylabel style={font=\color{white!15!black}},
ylabel={FER},
axis background/.style={fill=white},
xmajorgrids,
ymajorgrids,
legend style={at={(0.01,0.01)}, anchor=south west, legend cell align=left, align=left, draw=white!15!black,font=\scriptsize}
]

\addplot[color=kit-royalblue80,line width = 1pt,mark=x, mark options={solid}
  ]
  table[row sep=crcr] {
 1 0.9607843137254902\\
1.5 0.9131175468483816\\
2 0.8502415458937198\\
2.5 0.7893175074183977\\
3 0.6716981132075471\\
3.5 0.5483870967741935\\
4 0.40793878527020566\\
4.5 0.2898550724637681\\
5 0.17856500629383204\\
  };
\addlegendentry{$\bm{H}$}

\addplot[color=KITcyanblue,dashed,line width = 1pt,mark=triangle, mark options={solid}
]
table[row sep=crcr] {
1 0.3682256867112101\\
1.5 0.2587341772151899\\
2 0.16735602927139676\\
2.5 0.08631429079907457\\
3 0.0460173882497585\\
3.5 0.01879339004595678\\
4 0.006080565974445766\\
4.5 0.0017536246636058596\\
5 0.0003753071269811175\\
};
\addlegendentry{$\text{ssPCM-II}$}

\addplot[
  color=KITorange, dotted, line width = 1pt,mark=o, mark options={solid}
]
table[row sep=crcr] {
1 0.2045732952225398\\
1.5 0.12349397590361445\\
2 0.061209879489321085\\
2.5 0.027284023988330002\\
3 0.010230507458853242\\ 
3.5 0.003550111474242216\\
4 0.0009047232495665666\\
4.5 0.00017236712230581056\\
5 2.3279160786253657e-05\\
};
\addlegendentry{$\mathrm{MBBP}\text{-}8$}

\addplot[
  color=KITred!70, dashed, line width = 1pt,mark=diamond*, mark options={solid}
]
table[row sep=crcr] {
1 0.10417107770484861\\
1.5 0.05537383177570093\\
2 0.025630635194002636\\
2.5 0.009240885568654916\\
3 0.002921672505970638\\
3.5 0.0007399266602104438\\
4 0.0001458240732349713\\
4.5 2.06213747338567e-05\\
5 2.1796299146754064e-06\\
};
\addlegendentry{aSCED$\text{-}8$}

\addplot[
  color=KITorange!70, solid, line width = 1pt,mark=o, mark options={solid}
]
table[row sep=crcr] {
    1 0.1594773299748111\\
    1.5 0.09268559948863117\\
    2 0.04697953317805639\\
    2.5 0.019020286986640277\\
    3 0.006844473934743098\\
    3.5 0.0020249217928812193\\
    4 0.00047794497010405444\\
    4.5 8.598492356786651e-05\\
    5  9.823157835800825e-06\\
};
\addlegendentry{$\mathrm{MBBP}\text{-}16$}

  \addplot[
      color=KITred, solid, line width = 1pt,mark=diamond*, mark options={solid}
  ]
  table[row sep=crcr] {
    1 0.0852161537093044\\
    1.5 0.040089278049617996\\
    2 0.01964010496938393\\
    2.5 0.006863319952490491\\
    3 0.0020892420755962257\\
    3.5 0.000482844441384836\\
    4 8.907827351702929e-05\\
    4.5 1.2352300120959476e-05\\
    5 1.420040062047511e-06\\
  };
  \addlegendentry{aSCED$\text{-}16$}

  \addplot[
      color=black,
      solid,
      line width=0.9pt,
      mark=none
  ]
  table[row sep=crcr] {
 1.00  6.377551e-02\\
 1.50  3.297609e-02\\
 2.00  1.286794e-02\\
 2.50  4.973825e-03\\
 3.00  1.386016e-03\\
 3.50  3.495144e-04\\
 4.00  6.260739e-05\\
 4.50  9.161622e-06\\
  };
  \addlegendentry{\minrevi{$\mathrm{OSD}\text{-}5$ (ML)}}

\end{axis}
\end{tikzpicture}
    \vspace*{-1.2em}
    \caption{\revii{Performances of \ac{BP}-based decoders for the \ac{BCH} code $\mathcal{C}_3(63,18)$. \ac{NMSA} with normalization constant $\alpha=\frac{1}{2}$, $I_\mathrm{max}=20$ iterations. \minrevi{\ac{OSD}-$5$ serves as approximation of \ac{ML}}.%
    }}
    \vspace{-5pt}
    \label{fig:bch63_18}
\end{figure}
\begin{figure}[t]
    \centering
    \begin{tikzpicture}[scale=0.92]
\begin{axis}[%
width=.9\columnwidth,
height=\ferfigheight,
at={(0.758in,0.645in)},
scale only axis,
xmin=1,
xmax=5,
xlabel style={font=\color{white!15!black}},
xlabel={$E_{\mathrm{b}}/N_0$ ($\si{dB}$)},
ymode=log,
ymin=1e-05,
ymax=1,
yminorticks=true,
ylabel style={font=\color{white!15!black}},
ylabel={FER},
axis background/.style={fill=white},
xmajorgrids,
ymajorgrids,
legend style={at={(0.01,0.01)}, anchor=south west, legend cell align=left, align=left, draw=white!15!black,font=\scriptsize}
]

\addplot[
    color=kit-royalblue80,
    line width=1pt,
    mark=x,
    mark options={solid}
]
table[row sep=crcr] {
1 0.9874739039665971\\
1.5 0.9716312056737588\\
2 0.9428571428571428\\
2.5 0.9037162162162162\\
3 0.8378787878787879\\
3.5 0.7871443624868283\\
4 0.660741885625966\\
4.5 0.5234676974047487\\
5 0.41019514137793706\\
};
\addlegendentry{$\bm{H}$}

\addplot[
    color=KITcyanblue,
    dashed,
    line width=1pt,
    mark=triangle,
    mark options={solid}
]
table[row sep=crcr] {
1 0.1921598770176787\\
1.5 0.11373639101227705\\
2 0.07265660038869787\\
2.5 0.03064699205448354\\
3 0.012582340315298646\\
3.5 0.004148060638386973\\
4 0.0012174459225254635\\
4.5 0.0002823799633822674\\
5 4.542539930121391e-05\\
};
\addlegendentry{$\text{ssPCM-II}$}

\addplot[
    color=KITorange,
    dotted,
    line width=1pt,
    mark=o,
    mark options={solid}
]
table[row sep=crcr] {
1 0.10796621775850263\\
1.5 0.04985541928407618\\
2 0.02445109780439122\\
2.5 0.009592371690156897\\
3 0.0034009881841204504\\
3.5 0.0008734680588823515\\
4 0.00019817614496426379\\
4.5 3.04320632218105e-05\\
5 3.6763479659998725e-06\\
};
\addlegendentry{$\mathrm{MBBP}\text{-}8$}

\addplot[
    color=KITred!70,
    dotted,
    line width=1pt,
    mark=diamond*,
    mark options={solid}
]
table[row sep=crcr] {
1 0.07866468352338457\\
1.5 0.037824512749402975\\
2 0.01780213482341142\\
2.5 0.007025658928259731\\
3 0.0021438385899897505\\
3.5 0.000569295553681874\\
4 0.00011780009426948868\\
4.5 1.9448166764820117e-05\\
5 2.4187717889741474e-06\\
};
\addlegendentry{aSCED$\text{-}8$}

\addplot[
    color=KITorange!70,
    solid,
    line width=1pt,
    mark=o,
    mark options={solid}
]
table[row sep=crcr] {
1 0.07625636132315522\\
1.5 0.041947749294738136\\
2 0.01910957099722374\\
2.5 0.007481975241463747\\
3 0.0023100698721777863\\
3.5 0.0006294120716018103\\
4 0.0001318521174409123\\
4.5  2.0972431308105125e-05\\
5 2.967159274574628e-06\\
};
\addlegendentry{$\mathrm{MBBP}\text{-}16$}

\addplot[
    color=KITred,
    solid,
    line width=1pt,
    mark=diamond*,
    mark options={solid}
]
table[row sep=crcr] {
1 0.06544996853366897\\
1.5 0.0350923682616596\\
2 0.015804136766568173\\
2.5 0.006096129813944342\\
3 0.0018822491696249086\\
3.5 0.0004838386815446223\\
4 0.00010021674785000117\\
4.5  1.583113854281657e-05\\
5 2.0882416091638545e-06\\
};
\addlegendentry{aSCED$\text{-}16$}

\addplot[
    color=black,
    solid,
    line width=0.9pt,
    mark=none
]
table[row sep=crcr] {
 1.00  6.787714e-02\\
 1.50  3.058572e-02\\
 2.00  1.437401e-02\\
 2.50  5.570720e-03\\
 3.00  1.726050e-03\\
 3.50  4.594121e-04\\
 4.00  9.337020e-05\\
 4.50  1.421363e-05\\
};
\addlegendentry{\minrevi{$\mathrm{OSD}\text{-}5$ (ML)}}

\end{axis}
\end{tikzpicture}
    \vspace*{-1.2em}
    \caption{\revii{Performances of \ac{BP}-based decoders for the \ac{BCH} code $\mathcal{C}_4(63,16)$. \ac{NMSA} with normalization constant $\alpha=\frac{1}{2}$, $I_\mathrm{max}=20$ iterations. \minrevi{\ac{OSD}-$5$ serves as approximation of \ac{ML}}.%
    }}
    \label{fig:bch63_16}
    \vspace{-3pt}
\end{figure}

In the following, we consider two medium-rate \ac{BCH} codes, i.e., codes for which no efficient \ac{ML} decoding exists \revii{and two low-rate \ac{BCH} codes with rates ${r\leq 1/3}$, which are of particular interest for applications such as IoT and URLLC.} \revi{All \ac{aSCED} configurations use ${\Delta=1}$.}

In Fig.~\ref{fig:bch63_30}, we summarize the findings for the \ac{BCH} code $\mathcal{C}_1(63,30)$. %
As reference, we include \ac{ML} decoding, BDD-BM, and Chase-II performance \cite{channelcodes,Massey1969ShiftregisterSA,971758}. 
As \ac{BP} baselines, we include decoding using $\bm{H}_{1-\text{min}}$ and \ac{ssPCM}-II \cite{channelcodes,yifei_sspcm}.

Across the simulated SNR range, \ac{aSCED} consistently outperforms its equal-complexity (equal TEC) \ac{MBBP} counterpart. 
Specifically, using \ac{NMSA}, at a target \ac{FER} of $10^{-3}$, \ac{aSCED}-$8$ and \ac{aSCED}-$64$ achieve gains of $0.2\,\mathrm{dB}$ and $0.1\,\mathrm{dB}$,
 respectively, over their \ac{MBBP} counterparts. 
 Moreover, for  $E_{\mathrm{b}}/N_0\geq3.6\,\mathrm{dB}$ and using \ac{NSPA}, \ac{aSCED}-$64$ fully closes the gap to \ac{ML} performance.

\revi{Fig.~\ref{fig:bch63_36}--\ref{fig:bch63_16} plot the \ac{FER} versus SNR for various \ac{BP}-based decoders for the \ac{BCH} codes $\mathcal{C}_2(63,36)$, $\mathcal{C}_3(63,18)$, and $\mathcal{C}_4(63,16)$, respectively.}
As baselines, we include \ac{BP} decoding using $\bm{H}_{1-\text{min}}$ and \ac{ssPCM}-II \cite{yifei_sspcm,channelcodes}.

 \revi{For the \ac{aSCED} construction of the medium-rate code $\mathcal{C}_2(63,36)$, we append rows with Hamming weights in the range $26\leq \mathrm{wt}(\bm{M}_\ell)\leq 38$. For the lower rate code $\mathcal{C}_3(63,18)$ we append rows of weight ${\mathrm{wt}(\bm{M}_\ell)=4}$, whereas for \minrevi{$\mathcal{C}_4(63,16)$} the appended rows have weight ${4\leq \mathrm{wt}(\bm{M}_\ell)\leq 7}$.
 We construct the corresponding search spaces using Algorithm~\ref{alg:search_space}, which generates lower-weight dual codewords via linear combinations. These are subsequently used for the \ac{ssPCM}-II construction, where we set ${w_{\max}=2000}$. The exact appended rows and the resulting \acp{ssPCM} are provided in the accompanying repository \cite{aSCED_sim_2026} (see Remark~\ref{rem:github}), enabling reproduction of the reported results.
 }

For the medium-rate code $\mathcal{C}_2(63,36)$ in Fig.~\ref{fig:bch63_36}, at a target \ac{FER} of $10^{-3}$, \ac{aSCED}-$6$ and \ac{aSCED}-$30$ achieve gains of $0.2\,\mathrm{dB}$ and $0.3\,\mathrm{dB}$ compared to their equal-complexity \ac{MBBP} counterparts. 
Note that the performance gains of \ac{aSCED} over \ac{MBBP} that are observed for $\mathcal{C}_2(63,36)$ are larger than the gains observed for  $\mathcal{C}_1(63,30)$. This fact may be attributed to the structural constraints of the respective dual codes.
Specifically, constructing \acp{ssPCM} for $\mathcal{C}_2$ requires higher-weight dual codewords 
to achieve full-rank search spaces\cite{yifei_sspcm}. This can degrade the sparsity and structural quality of the resulting \ac{ssPCM}.
 In contrast, appending linearly independent rows, i.e., as done for \ac{aSCED}, mitigates this degradation, preserving or even enhancing the decoding graph structure.%

 \revii{To further position \ac{aSCED} against modern short-code decoding approaches, Table~\ref{table:ber_comparison} includes a bit error rate (BER) comparison for the code $\mathcal{C}_2(63,36)$ with a Transformer-based decoder from \cite{NEURIPS2022_fcd3909d}. The results show that \ac{aSCED}-$30$ reduces the BER by approximately two orders of magnitude compared to the reported Transformer-based decoder.}

\revi{When comparing the two lower-rate codes in Fig.~\ref{fig:bch63_18} and  Fig.~\ref{fig:bch63_16} at a target FER of $10^{-3}$, we observe that for $\mathcal{C}_4(63,16)$, BP decoding using an ssPCM already reduces the gap to ML to below $1$ dB, whereas a larger gap remains for $\mathcal{C}_3(63,18)$. This can be attributed to the different structures of the available \acp{PCM}:
The PCM of $\mathcal{C}_4(63,16)$ consists of rows of degree $6$, which are well suited for BP decoding and consequently allow the construction of effective ssPCMs.
In contrast, the PCM of $\mathcal{C}_3(63,18)$ contains solely rows of degree $8$.}

\revi{For $\mathcal{C}_4(63,16)$, the code therefore already provides sufficiently well-structured checks and the additional subcode checks do not substantially improve the resulting ssPCM structure.
Consequently, the gains of aSCED-$M$ over the respective equal-complexity MBBP-$M$, for $M\in\{8,16\}$, are marginal.
In contrast, for code $\mathcal{C}_3(63,18)$, the use of subcodes introduces lower-weight rows into the search space of the ssPCM construction, enabling the construction of more effective \acp{ssPCM}.
As a result, aSCED-$8$ and aSCED-$16$ provide gains of $0.5$ dB and $0.4$ dB, respectively, compared to the corresponding equal-complexity MBBP configurations.
Notably, aSCED-$8$ already outperforms MBBP-$16$.}

\reviv{
Overall, the results indicate that aSCED consistently matches or improves upon the performance of the corresponding MBBP configurations, while the gains depend on the structure of the original code and the extent to which the subcode provides access to beneficial parity checks.
}

\begin{table}[t]
        \centering
        \caption{\revii{BER of \ac{aSCED}-$30$ and the Transformer-based neural decoder from \cite{NEURIPS2022_fcd3909d} for the BCH code $\mathcal{C}_2(63,36)$}}\label{table:ber_comparison}
        \begin{tabular}{ccc}
                \toprule
                Decoder &   $4$ dB & $5$ dB \\
                \midrule                
                Transformer architecture ($N=6$) &   $7.75\cdot10^{-3}$  & $1.29\cdot10^{-3}$ \\
                \ac{aSCED}-$30$ &   $7.6128\cdot 10^{-5}$ & $7.6124\cdot10^{-7}$  \\
                \bottomrule
        \end{tabular}
\end{table}

\section{Conclusion}\label{sec:conclusion}
In this work, we introduce \ac{aSCED} for binary linear codes.
This approach extends the previously proposed \ac{SCED} by incorporating affine subcodes alongside linear subcodes.
To this end, we develop a \ac{BP} decoder for affine subcodes.
Leveraging the fact that an affine subcode is a coset of a linear subcode, we show that the decoding graph structure remains unchanged, only requiring sign flips at specific check nodes.

We further provide theoretical results on \ac{aSCED} and, in particular, demonstrate that by appending linearly independent rows \emph{before} search-space generation, we gain a structural advantage over \ac{MBBP} ensembles. Consequently, subcodes enable the design of high-performance subcode ssPCMs.

Compared to \ac{SCED}, \ac{aSCED} offers four key advantages:  
\begin{enumerate}
    \item \textbf{Simplified design}: Fewer PCMs must be designed to achieve the same ensemble size. 
     \item \textbf{Uniform protection}: The inclusion of affine subcodes ensures that all codewords are equally protected.
     \item \textbf{Improved performance}: \ac{aSCED} consistently outperforms other ensemble methods at equal complexity.  
     \item \textbf{Hardware efficiency}: All decodings within an \ac{aSCED} batch share the same graph structure, potentially enabling hardware reuse.
\end{enumerate}

To demonstrate its effectiveness, we evaluate \ac{aSCED} using four LDPC codes and four \ac{BCH} codes, achieving significant gains over stand-alone \ac{BP} and other equal-complexity ensemble decoders.
 Furthermore, when combining \ac{aSCED} with \acp{ssPCM}, we report \ac{ML} performance for the $(63,30)$ BCH code with \ac{BP}-based decoding using only $64$ decoding paths.

{\appendix[Proof of Theorem~\ref{th:error_prob_affine}]%
In the following proof, we denote the \ac{BPSK} codeword as $\ddot{\bm x}:=1-2\bm x$ \cite{MCT08}. 
\begin{figure}[t]
    \centering
\definecolor{KITblue}{rgb}{.27,.39,.66}
\definecolor{greencolor}{rgb}{0,.59,.51}

\newlength{\CNsize} 
\setlength{\CNsize}{0.2cm} %

\begin{tabular}[ht]{@{}c@{\,}c@{}@{}c@{}}

	\begin{tikzpicture}[scale=1.1] \label{fig:illustration:a}

\draw [fill=KITblue!10, draw=none, very thick] (-1.05, -0.5) rectangle ++ (2.65,3.3);
\draw [decorate,thick,decoration={brace,amplitude=5pt,raise=4pt}] (-0.85,2.5) --(1.1,2.5) node[midway,yshift=14pt] {\scriptsize$\mathrm{BP}_{\bm{s}_\mathsf{a}}\left(\bm{H_\mathsf{s}},\bm{\lambda}^{[2]}\right)$};

\node[draw=none, anchor=center] at (-0, +1.7) {$\vdots$};	
\node[draw=none, anchor=center] at (-0, +0.4) {$\vdots$};	
\node[draw=none, anchor=center] at (0.5, +0.4) {$\vdots$};	
\node[draw=none, anchor=center] at (1.0, +1.7) {$\vdots$};	
\node[draw=none, anchor=center] at (1.0, +0.4) {$\vdots$};	
\node[draw=none, anchor=center] at (-0.7, +1.7) {$\vdots$};	
\node[draw=none, anchor=center] at (-0.7, +0.4) {$\vdots$};	

\draw [black!50, very thick] (0.115,2.5) -- (0.9,2.4);
\draw [black!50, very thick] (0.115,2.1) -- (0.9,2.4);
\draw [black!50, very thick] (0.115,1.1) -- (0.9,2.4);
\draw [ very thick] (0.115,2.1) -- (0.9,0.5);
\draw [ very thick] (0.115,1.1) -- (0.9,0.5);

\node[draw=none, anchor=center] at (-0.7, +2.5) {\textcolor{KITblue}{\scriptsize$-\lambda^{[1]}_{1}$}};	
\node[draw=none, anchor=center] at (-0.7, +2.1) {\textcolor{KITblue}{\scriptsize$-\lambda^{[1]}_{2}$}};	

\node[draw=none, anchor=center] at (-0.7, +1.1) {\textcolor{KITblue}{\scriptsize$-\lambda^{[1]}_{\beta}$}};	
\node[draw=none, anchor=center] at (-0.7, +0.7) {\textcolor{KITblue}{\scriptsize$\lambda^{[1]}_{\beta+1}$}};	

\node[draw=none, anchor=center] at (-0.7, -0.2) {\textcolor{KITblue}{\scriptsize$\lambda^{[1]}_{n}$}};	

\foreach \y in {+2.5,+2.1, +1.1,+0.7,-0.2}
{
    \draw[ -latex] (-0.4, \y) --(0-0.1, \y);
    \draw [ fill=white, very thick] (0, \y) circle (0.1);
}

\foreach \y in {+2.4, +1.3}
{
    \draw [KITred, -latex] (1.0,\y-0.35) -- (1.0,\y-0.1);
	\draw [black!50,fill=white, very thick] (1.0-0.1, \y-0.1) rectangle ++(\CNsize,\CNsize);
   \node[black!50,draw=none, anchor=center] at (1.0, \y) {\footnotesize$+$};
   \node[draw=none, anchor=center] at (1.0, \y-0.4) {\textcolor{KITred}{\tiny$-\infty$}};
   }

\draw [decorate,thick,decoration={brace,amplitude=5pt,raise=4pt}] (1.05,2.55) --(1.05,0.8) node[midway,right=6pt] {$\alpha$};

\foreach \y in {0.5,-0.2}
{
	\draw [fill=white,  very thick] (1.0-0.1, \y-0.1) rectangle ++(\CNsize,\CNsize);
   \node[draw=none, anchor=center] at (1.0, \y) {\footnotesize$+$};
}

\end{tikzpicture}	
 \hspace*{-0.8em}    & 

	\begin{tikzpicture}[scale=1] \label{fig:illustration:b} 

\draw [fill=KITblue!10, draw=none, very thick] (-1.05, -0.5) rectangle ++ (2.3,3.3);
\node[draw=none, anchor=center] at (-0, +1.7) {$\vdots$};	
\node[draw=none, anchor=center] at (-0, +0.4) {$\vdots$};	
\node[draw=none, anchor=center] at (0.5, +0.4) {$\vdots$};	
\node[draw=none, anchor=center] at (1.0, +1.7) {$\vdots$};	
\node[draw=none, anchor=center] at (1.0, +0.4) {$\vdots$};	
\node[draw=none, anchor=center] at (-0.7, +1.7) {$\vdots$};	
\node[draw=none, anchor=center] at (-0.7, +0.4) {$\vdots$};	

\draw [KITred, very thick] (0.115,2.5) -- (0.9,2.4) node [pos=0.7, yshift=3pt, sloped]  {\tiny\textcolor{KITred}{$-1$}};
\draw [KITred, very thick] (0.115,2.1) -- (0.9,2.4) node [pos=0.4, yshift=-3.5pt, sloped]  {\tiny\textcolor{KITred}{$-1$}};
\draw [KITred, very thick] (0.115,1.1) -- (0.9,2.4) node [pos=0.7, yshift=-3.5pt, sloped]  {\tiny\textcolor{KITred}{$-1$}};
\draw [KITgreen, very thick] (0.115,2.1) -- (0.9,0.5) node [pos=0.7, yshift=3pt, sloped]  {\tiny\textcolor{KITgreen}{$-1$}};
\draw [KITgreen, very thick] (0.115,1.1) -- (0.9,0.5) node [pos=0.7, yshift=-3.5pt, sloped]  {\tiny\textcolor{KITgreen}{$-1$}};

\node[draw=none, anchor=center] at (-0.7, +2.5) {\textcolor{KITblue}{\scriptsize$-\lambda^{[1]}_{1}$}};	
\node[draw=none, anchor=center] at (-0.7, +2.1) {\textcolor{KITblue}{\scriptsize$-\lambda^{[1]}_{2}$}};	
\node[draw=none, anchor=center] at (-0.7, +1.1) {\textcolor{KITblue}{\scriptsize$-\lambda^{[1]}_{\beta}$}};	
\node[draw=none, anchor=center] at (-0.7, +0.7) {\textcolor{KITblue}{\scriptsize$\lambda^{[1]}_{\beta+1}$}};	
\node[draw=none, anchor=center] at (-0.7, -0.2) {\textcolor{KITblue}{\scriptsize$\lambda^{[1]}_{n}$}};	

\foreach \y in {+2.5,+2.1, +1.1,+0.7,-0.2}
{
    \draw[, -latex] (-0.4, \y) --(0-0.1, \y);
    \draw [ fill=white, very thick] (0, \y) circle (0.1);
}

\foreach \y in {+2.4, +1.3}
{
\draw [black!50,fill=white, very thick] (1.0-0.1, \y-0.1) rectangle ++(\CNsize,\CNsize);
\node[black!50,draw=none, anchor=center] at (1.0, \y) {\footnotesize$+$};
}

\foreach \y in {0.5,-0.2}
{
	\draw [fill=white,  very thick] (1.0-0.1, \y-0.1) rectangle ++(\CNsize,\CNsize);
   \node[draw=none, anchor=center] at (1.0, \y) {\footnotesize$+$};
}

\end{tikzpicture}	
\hspace*{-0.6em}
 
    &

	\begin{tikzpicture}[scale=1] \label{fig:illustration:c} 

\draw [fill=KITblue!10, draw=none, very thick] (-2.1, -0.5) rectangle ++ (2.1,3.3);
\draw [fill=KITgreen!15, draw=none, very thick] (-1.05, -0.5) rectangle ++ (2.3,3.3);

\node[draw=none, anchor=center] at (-0, +1.7) {$\vdots$};	
\node[draw=none, anchor=center] at (-0, +0.4) {$\vdots$};	
\node[draw=none, anchor=center] at (0.5, +0.4) {$\vdots$};	
\node[draw=none, anchor=center] at (1.0, +1.7) {$\vdots$};	
\node[draw=none, anchor=center] at (1.0, +0.4) {$\vdots$};	
\node[draw=none, anchor=center] at (-0.7, +1.7) {$\vdots$};	
\node[draw=none, anchor=center] at (-0.7, +0.4) {$\vdots$};	
\node[draw=none, anchor=center] at (-1.8, +0.4) {$\vdots$};	

\draw [decorate,thick,decoration={brace,amplitude=5pt,raise=4pt}] (-0.85,2.5) --(1.1,2.5) node[midway,yshift=14pt] {\scriptsize$\mathrm{BP}\left(\bm{H}_\mathsf{s},\bm{\lambda}^{[1]}\right)$};
\draw [decorate,thick,decoration={brace,amplitude=5pt,raise=4pt}] (-2,3.05) --(1.2,3.05) node[midway,yshift=14pt] {\scriptsize$\mathrm{BP}_{\bm{s}_\mathsf{a}}\left(\bm{H_\mathsf{s}},\bm{\lambda}^{[2]}\right)$};

\draw [black!50, very thick] (0.115,2.5) -- (0.9,2.4);
\draw [black!50, very thick] (0.115,2.1) -- (0.9,2.4);
\draw [black!50, very thick] (0.115,1.1) -- (0.9,2.4);
\draw [ very thick] (0.115,2.1) -- (0.9,0.5);
\draw [ very thick] (0.115,1.1) -- (0.9,0.5);

\node[draw=none, anchor=center] at (-1.8, +2.5) {\textcolor{KITblue}{\scriptsize$-\lambda^{[1]}_{1}$}};	
\node[draw=none, anchor=center] at (-1.8, +2.1) {\textcolor{KITblue}{\scriptsize$-\lambda^{[1]}_{2}$}};	
\node[draw=none, anchor=center] at (-1.8, +1.1) {\textcolor{KITblue}{\scriptsize$-\lambda^{[1]}_{\beta}$}};	
\node[draw=none, anchor=center] at (-1.8, +0.7) {\textcolor{KITblue}{\scriptsize$\lambda^{[1]}_{\beta+1}$}};	
\node[draw=none, anchor=center] at (-1.8, -0.2) {\textcolor{KITblue}{\scriptsize$\lambda^{[1]}_{n}$}};	

\node[draw=none, anchor=center] at (-0.7, +2.5) {\textcolor{KITgreen}{\scriptsize$\lambda^{[1]}_{1}$}};	
\node[draw=none, anchor=center] at (-0.7, +2.1) {\textcolor{KITgreen}{\scriptsize$\lambda^{[1]}_{2}$}};	
\node[draw=none, anchor=center] at (-0.7, +1.1) {\textcolor{KITgreen}{\scriptsize$\lambda^{[1]}_{\beta}$}};	
\node[draw=none, anchor=center] at (-0.7, +0.7) {\textcolor{KITgreen}{\scriptsize$\lambda^{[1]}_{\beta+1}$}};	
\node[draw=none, anchor=center] at (-0.7, -0.2) {\textcolor{KITgreen}{\scriptsize$\lambda^{[1]}_{n}$}};	

\foreach \y in {+2.5,+2.1, +1.1}
{
    \draw[] (-1.4, \y) --(-1, \y) node [pos=0.5, yshift=+3pt, sloped]  {\tiny$-1$};;
    \draw[ -latex] (-0.4, \y) --(0-0.1, \y);
    \draw [ fill=white, very thick] (0, \y) circle (0.1);
}
\foreach \y in {+0.7,-0.2}
{
   \draw[] (-1.4, \y) --(-1, \y);
    \draw[ -latex] (-0.4, \y) --(0-0.1, \y);
    \draw [ fill=white, very thick] (0, \y) circle (0.1);
}

\foreach \y in {+2.4, +1.3}
{
    \draw [black!50,fill=white, very thick] (1.0-0.1, \y-0.1) rectangle ++(\CNsize,\CNsize);
   \node[black!50,draw=none, anchor=center] at (1.0, \y) {\footnotesize$+$};
   }

\foreach \y in {0.5,-0.2}
{
	\draw [fill=white,  very thick] (1.0-0.1, \y-0.1) rectangle ++(\CNsize,\CNsize);
   \node[draw=none, anchor=center] at (1.0, \y) {\footnotesize$+$};
}

\end{tikzpicture}	
	\\
	{\footnotesize (a) } & {\footnotesize (b) } & {\footnotesize (c) }
\end{tabular}
    \caption{Illustration of the key transformation steps in the proof of Theorem~\ref{th:error_prob_affine}: (a) Affine decoding problem; (b) Transferring the signs to edges; (c) Interrelation of affine and subcode decoding problem. Remaining edges are not depicted for readability.}
    \label{fig:illustration_symmetry}
\end{figure}
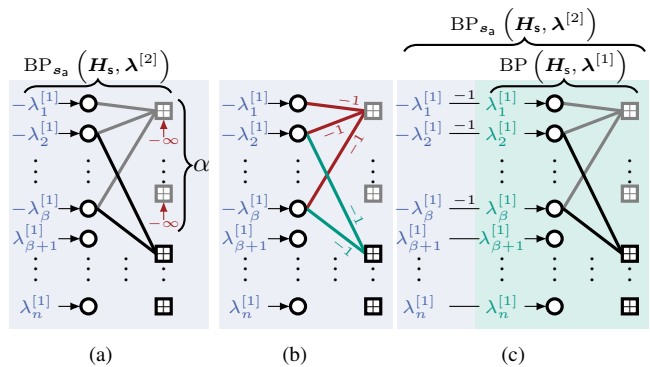
\begin{IEEEproof}
Let $\bm H_{\mathsf s}\in\mathbb F_{2}^{\,m_{\mathsf s}\times n}$ denote a \ac{PCM} of a fixed proper subcode
$\sC(n,k')\subset \mathcal C$, \jmcol{with decoder $\mathrm{BP}(\bm H_{\mathsf s})$}. 
After selecting $\sC$, there exist ${2^{k-k'}-1}$ disjoint affine subcodes from which we arbitrarily select one \jmcol{affine subcode $\aC=\bm x_{\mathsf a}+\mathcal C_{\mathsf s}$
with affine offset $\bm x_{\mathsf a}$, affine syndrome $ \bm s_{\mathsf a}= \bm H_{\mathsf s}\,\bm x_{\mathsf a}^{\mathsf T}$, and decoder $\mathrm{BP}_{\bm{s}_\mathsf{a}}(\bm H_{\mathsf s})$}.
Without loss of generality, we may reorder the rows and columns of $\bm H_{\mathsf s}$ such that  
\begin{equation}\label{eq:appendix:s_a_x_a}
\bm s_{\mathsf a}= \bigl(\underbrace{1,\ldots,1}_{\alpha},
\underbrace{0,\ldots,0}_{m_{\mathsf s}-\alpha}\bigr)^{\mathsf T},
\quad
\bm x_{\mathsf a}= \bigl(\underbrace{1,\ldots,1}_{\beta},
\underbrace{0,\ldots,0}_{n-\beta}\bigr),
\end{equation}
with $1\!\le\alpha\!\le m_{\mathsf s}$ and $1\!\le\beta\!\le n$; we consider BiMSCs, so this re-ordering does not affect the performance of flooding \ac{BP}.

Consider the transmission of vector $\ddot{\bm x}^{[1]}$ with ${\bm{x}^{[1]}\in\mathcal C}$.
Let the resulting channel output be ${\bm y^{[1]}\in\mathcal Y^{n}}$ with corresponding \ac{LLR} vector
${\bm \lambda^{[1]}}\in \mathbb{R}^n$.
We do not assume $\bm x^{[1]}\in\mathcal C_{\mathsf s}$ since in \ac{aSCED}, subcode decoders may be asked to decode a received word that originates from a codeword within an arbitrary coset $\aC$, %
not necessarily coinciding with the corresponding subcode. %

Now, consider the codeword $\bm x^{[2]} := \bm x^{[1]} + \bm x_{\mathsf a}$. %
Note that $\bm x^{[2]}\in \aC$ iff $\bm x^{[1]}\in \sC$. 
Next, we introduce the bijection ${f_\beta(\cdot):\mathbb{R}^n\to\mathbb{R}^n}$, inverting the sign of the first $\beta$ components
\[
\left( f_{\beta}(\bm{v}) \right)
_i
=
\begin{cases}
-v_i, & 1\leq i\le \beta,\\
v_i,  & \beta < i\le n .
\end{cases}
\]
By construction, we have ${\ddot{\bm x}^{[2]}=f_{\beta}(\ddot{\bm x}^{[1]})}$ and we define ${\bm \lambda^{[2]}:= f_{\beta}(\bm \lambda^{[1]})}$.

Next, we show that the average decoding error probabilities of $\mathrm{BP}(\bm H_{\mathsf s})$ and $\mathrm{BP}_{\bm{s}_\mathsf{a}}(\bm H_{\mathsf s})$ are identical. 
Define the set of bit-wise \ac{LLR} vectors successfully decoded by $\mathrm{BP}(\bm H_{\mathsf s})$
\[
\mathcal S^{[1]} :=\bigl\{\bm \lambda^{[1]}\in\mathbb R^{n}:\mathrm{BP}\bigl(\bm H_{\mathsf s},\bm \lambda^{[1]}\bigr)=\bm x^{[1]}\bigr\},
\]
where we assume the transmission of ${\bm{x}^{[1]}\in\mathcal C}$. 

When defining $\mathcal S^{[2]} := \{f_{\beta}(\bm \lambda):\bm{\lambda}\in \mathcal S^{[1]}\}$, \jmcol{and using the symmetry property of BiMSCs, it follows that}
\begin{equation}\label{eq:prob_a_s}
P(\mathcal S^{[2]}|\bm{X}=\bm{x}^{[2]})=P(\mathcal S^{[1]}|\bm{X}=\bm{x}^{[1]}),
\end{equation}
where $\bm{X}$ is random codeword uniformly distributed over $\mathcal{C}$.

Let $\bm \lambda ^{[2]} = f_{\beta}(\bm \lambda^{[1]})\in\mathcal S^{[2]}$.
We show that decoding $\bm x^{[1]}$ from $\bm \lambda^{[1]}$ using $\mathrm{BP}(\bm H_{\mathsf s})$ is
equivalent to decoding $\bm x^{[2]}$ from $\bm \lambda^{[2]}$ using $\mathrm{BP}_{\bm{s}_\mathsf{a}}(\bm H_{\mathsf s})$.

To this end, as illustrated in Fig.~\ref{fig:illustration_symmetry}, we transform the Tanner graph of the affine decoding problem based on the following observations: 
For the first $\alpha$ CNs, the affine syndrome component is non-zero by construction, thus the CN update in (\ref{eq:affine_update_eq}) introduces the factor $(-1)$ which we refer to as \emph{signed} CN and which is indicated in Fig.~\ref{fig:illustration_symmetry} by \emph{injecting} $\color{KITred}{-\infty}$ to the signed CNs.
For all remaining CNs, the affine syndrome component is zero, and they behave like standard CN updates. 

Due to reordering, $\bm s_{\mathsf a}$ and $\bm x_{\mathsf a}$ are non‑zero precisely in their first $\alpha$ and $\beta$ positions, respectively.
\jmcol{In order for $\bm s_{\mathsf a}= \bm H_{\mathsf s}\,\bm x_{\mathsf a}^{\mathsf T}$ to hold}, each VN $\mathsf{v}_{i}$ with $i\le\beta$ shares an edge with
\begin{itemize}
    \item  an odd number of the $\alpha$ \emph{signed} CNs, and  
    \item  an even number of the remaining $m_{\mathsf s}-\alpha$ standard CNs.
\end{itemize}

Since there is an odd number of edges connecting the first $\beta$ VNs with any of the first $\alpha $ CNs, the $(-1)$ factor of each signed CN can be transferred to every edge that
connects that CN with a VN $\mathsf{v_i}$, $i\le\beta$. In Fig.~\ref{fig:illustration_symmetry}~(b), these edges are drawn in \textcolor{KITred}{red} and are labeled with an additional weight $(-1)$.
After transferring the $(-1)$ factor onto the edges, the signed CNs revert to
standard CNs.

Similarly, each of the first $\beta$ VNs is linked to the remaining $m_{\mathsf s}-\alpha$ CNs by an even number of edges. 
Thus, an auxiliary factor $(-1)$ can be attached to every one of these
edges without altering the update rule of the connected CN.
In Fig.~\ref{fig:illustration_symmetry}~(b), these are shown in \textcolor{KITgreen}{green} and carry the extra $(-1)$ weight.

Summarizing, after those two \jmcol{sign transfers} shown in Fig.~\ref{fig:illustration_symmetry}~(b), \emph{all} edges that connect a VN $v_i$ with $i\le\beta$ bear a $(-1)$ factor and every CN has now become a standard CN. 
Next, in a final transformation step from the graph in Fig.~\ref{fig:illustration_symmetry}~(b) to the graph illustrated in Fig.~\ref{fig:illustration_symmetry}~(c), 
 the $(-1)$ weights can be shifted onto the first $\beta$ incoming LLRs.

As highlighted in Fig.~\ref{fig:illustration_symmetry}~(c), the resulting graph interrelates the affine‑decoding problem (blue background) with the subcode‑decoding problem (green background).
\jmcol{This representation shows that the affine decoder correctly recovers $\bm{x}^{[2]}$ from $\bm \lambda ^{[2]}\in \mathcal S^{[2]}$ if and only if the subcode decoder correctly recovers $\bm{x}^{[1]}$ from $\bm \lambda ^{[1]}\in \mathcal S^{[1]}$. 
Thus, by definition of $\mathcal S^{[1]}$, $\mathcal S^{[2]}$ is the set of bit-wise \ac{LLR} vectors successfully decoded by $\mathrm{BP}(\bm H_{\mathsf s})$}.
Hence, using \eqref{eq:prob_a_s} and averaging over all codewords $\bm{x}\in \mathcal{C}$ completes the proof.
\end{IEEEproof}

}

\newpage

\vfill

\end{document}